%% file: main.tex
\documentclass[final,3p,times]{elsarticle}
\input{setup.tex}

\begin{document}

\input{contents/acronyms.tex}

\begin{frontmatter}

\journal{Journal of Computational Physics}
\title{Extension of compressive sampling to binary vector recovery for model-based defect imaging}

\author{Wei-Chen Li}
\author{Chun-Yeon Lin\corref{cor1}}
\cortext[cor1]{Corresponding author.\ead{chunyeonlin@ntu.edu.tw}}

\affiliation{organization={Department of Mechanical Engineering, National Taiwan University},
             city={Taipei},
             postcode={106},
             country={Taiwan}}

\input{contents/abstract.tex}

\end{frontmatter}

\input{contents/section1.tex}
\input{contents/section2.tex}
\input{contents/section3.tex}
\input{contents/section4.tex}
\input{contents/section5.tex}
\input{contents/section6.tex}
\input{contents/acknowledgment.tex}

\appendix
\input{contents/appendix1.tex}
\input{contents/appendix2.tex}
\input{contents/appendix3.tex}

\bibliographystyle{elsarticle-num}
\bibliography{contents/references}

\end{document}

%% file: setup.tex
\usepackage{amsmath,amsthm,amssymb}
\usepackage{graphicx}
\usepackage{float}
\usepackage{subcaption}
\usepackage{bm}
\usepackage{tikz}
\usepackage{pgfplots}
\usepackage{array}
\usepackage{multirow}
\usepackage{booktabs}
\usepackage{enumitem}
\usepackage{silence}
\usepackage[nolist,nohyperlinks]{acronym}
\usepackage[ruled,linesnumbered]{algorithm2e}
\usepackage[hidelinks]{hyperref}
\usepackage[capitalize]{cleveref}
\usepackage[numbered,open]{bookmark}
\usepackage[notransparent]{svg}
\usepackage[export]{adjustbox}

\renewcommand{\vec}[1]{\bm{#1}}
\newcommand{\mat}[1]{\boldsymbol{\mathbf{#1}}}
\newcommand{\support}[1]{\operatorname{supp}(#1)}
\newcommand{\diag}{\operatorname{diag}}
\newcommand{\trace}{\operatorname{Tr}}

\newcommand{\gaussdist}{\mathcal{N}}
\newcommand{\KL}{D_{\text{KL}}}
\newcommand{\prob}{\mathbb{P}}
\newcommand{\expval}{\mathbb{E}}
\newcommand{\const}{\textit{const}}
\newcommand{\set}[1]{\mathsf{#1}}
\newcommand{\cpm}{\complement}
\newcommand{\argmin}{\mathop{\mathrm{arg\,min}}}
\newcommand{\argmax}{\mathop{\mathrm{arg\,max}}}
\newcommand{\minimize}{\mathop{\mathrm{minimize}}}

\newcommand{\subjectto}{\mathop{\text{subject to}}}

\newtheorem{theorem}{Theorem}
\newtheorem{proposition}{Proposition}
\newtheorem{lemma}{Lemma}
\newtheorem{claim}{Claim}

\theoremstyle{definition}
\newtheorem{definition}{Definition}
\newtheorem{remark}{Remark}

\usetikzlibrary{arrows.meta}
\usetikzlibrary{3d}
\usetikzlibrary{calc}

\usepgfplotslibrary{colorbrewer}
\pgfplotsset{compat=1.9}
\pgfplotsset{legend cell align={left}}
\pgfplotsset{legend style={font=\normalsize}}
\pgfplotsset{tick label style={font=\small}}
\pgfplotsset{every axis plot/.append style={line width=0.8pt}}
\pgfplotsset{cycle list/Set1-5}

\newcommand{\thickhline}{\specialrule{0.12em}{-1pt}{0em}}
\newcommand{\semithickhline}{\specialrule{0.08em}{0em}{0em}}
\newcolumntype{C}[1]{ >{\centering\arraybackslash}m{#1} }

%% file: contents/acronyms.tex
\begin{acronym}

\acro{IID}{independent and identically distributed}
\acro{MRI}{magnetic resonance imaging}
\acro{MIT}{magnetic induction tomography}
\acro{EIT}{electrical impedance tomography}
\acro{RF}{radio frequency}
\acro{SBL}{sparse Bayesian learning}
\acro{SSP}{subset sum problem}
\acro{EM}{expectation-maximization}
\acro{MAP}{maximum a posteriori}
\acro{KL}{Kullback-Liebler}
\acro{ADMM}{alternating direction method of multipliers}
\acro{MFD}{magnetic flux density}
\acro{SOVP}{second order vector potential}
\acro{SNR}{signal-to-noise ratio}
\acro{IoU}{intersection over union}
\acro{RIP}{restricted isometry property}
\acro{BRIP}{bounded restricted isometry property}

\end{acronym}

%% file: contents/abstract.tex
\begin{abstract}
Common imaging techniques for detecting structural defects typically require sampling at more than twice the spatial frequency to achieve a target resolution. This study introduces a novel framework for imaging structural defects using significantly fewer samples. In this framework, defects are modeled as regions where physical properties shift from their nominal values to resemble those of air, and a linear approximation is formulated to relate these binary shifts in physical properties with corresponding changes in measurements.
Recovering a binary vector from linear measurements is generally an NP-hard problem. To address this challenge, this study proposes two algorithmic approaches. The first approach relaxes the binary constraint, using convex optimization to find a solution. The second approach incorporates a binary-inducing prior and employs approximate Bayesian inference to estimate the posterior probability of the binary vector given the measurements.
Both algorithmic approaches demonstrate better performance compared to existing compressive sampling methods for binary vector recovery. The framework's effectiveness is illustrated through examples of eddy current sensing to image defects in metal structures.
\end{abstract}



%% file: contents/section1.tex
\section{Introduction}

Some common imaging problems aiming to reconstruct the physical properties of a specimen in 3D space include \ac{MRI}, \ac{MIT}, and \ac{EIT}. In \ac{MRI}, hydrogen nuclei align with an external magnetic field. When exposed to \ac{RF} pulses, these nuclei absorb the \ac{RF} energy, causing their spins to become excited and temporarily disoriented from the magnetic field. As they relax back to their original alignment, they release energy in the form of \ac{RF} waves, which are then detected by antennas. In \ac{MIT}, a time-varying magnetic field induces eddy currents in the specimen. These eddy currents create secondary magnetic fields, which induce an electromotive force in the sensing coils. In \ac{EIT}, currents are injected into the specimen through electrodes placed on its surface, and the resulting voltages are measured at the boundary using the same or different electrodes. Despite their drastic difference in sensing modalities, these three imaging problems can all be formulated as linear inverse problems \cite{RN1159, ma2017, RN119}, represented by the equation:
\begin{equation*}
    \vec{y} = \mat{\Phi} \vec{x} + \vec{n} ,
\end{equation*}
where $\vec{y}$ represents the measurements, $\mat{\Phi}$ the measurement matrix, $\vec{x}$ the physical properties to be reconstructed, and $\vec{n}$ the measurement noise. Additionally, the vector $\vec{x}$ is typically sparse. \ac{MRI} images, for example, are sparse in a wavelet basis, whereas in \ac{MIT} and \ac{EIT}, the regions with electrical conductivity anomalies are small compared to the overall specimen size. To recover sparse vectors from underdetermined linear systems, compressive sampling algorithms are often employed.

\subsection{Compressive sampling}
Compressive sampling or compressed sensing, as coined by \cite{RN271, RN276}, enables the recovery of signals using significantly fewer samples than dictated by the Nyquist sampling theorem. Comprehensive surveys of compressive sampling algorithms can be found in \cite{RN1226, RN1217}.
Among these algorithms, two primary classes are frequently employed.

The first class of methods, known as convex relaxation or $\ell_1$ minimization, relaxes the $\ell_0$ minimization problem
\begin{equation*}
    \begin{array}{ll}
        \minimize  & \|\vec{x}\|_0  \\
        \subjectto & \| \mat{\Phi} \vec{x} - \vec{y} \|_2 \leq \epsilon
    \end{array}
\end{equation*}
as the $\ell_1$ minimization problem
\begin{equation*}
    \begin{array}{ll}
        \minimize  & \|\vec{x}\|_1  \\
        \subjectto & \| \mat{\Phi} \vec{x} - \vec{y} \|_2 \leq \epsilon .
    \end{array}
\end{equation*}
The $\ell_0$ norm counts the number of nonzero entries, but optimizing it is NP-hard \cite{natarajan95}. In contrast, the $\ell_1$ norm serves as the convex envelope of the $\ell_0$ norm \cite{RN1099}, allowing the $\ell_1$ minimization problem to be efficiently solved using convex optimization algorithms. Theoretical results ensure exact recovery under specific conditions \cite{RN227, RN230, RN1200}. Notably, if the vector is sufficiently sparse and the measurement matrix satisfies the \ac{RIP}--meaning the matrix columns are approximately mutually orthogonal--$\ell_1$ minimization can yield exact recovery.

The second class of methods is known as \ac{SBL} \cite{RN64, RN211}, which formulates the problem within a probabilistic framework. In \ac{SBL}, a sparse-inducing prior is assigned to the vector $\vec{x}$, and the posterior probability of the vector given the measurements is calculated. The posterior probability is expressed as:
\begin{equation*}
    p(\vec{x} | \vec{y}) = \frac{p(\vec{y} | \vec{x}) p(\vec{x})}{p(\vec{y})} ,
\end{equation*}
where $p(\vec{x})$ represents the sparse-inducing prior, $p(\vec{y} | \vec{x})$ denotes the measurement likelihood, and $p(\vec{y})$ is the marginal density of $\vec{y}$. Although \ac{SBL} do not provide any recovery guarantees, they empirically perform better than $\ell_1$ minimization \cite{RN1226}. Another advantage of \ac{SBL} is its ability to accommodate more complex structures through prior selection, such as spatial clustering \cite{RN33, RN10, RN127} or temporal correlation \cite{RN96}. For certain prior choices, the marginal density is tractable, allowing the use of the expectation-maximization algorithm for inference \cite{RN33, RN10, RN96}. In other cases, more advanced inference algorithms, such as variational inference \cite{RN127} or approximate message passing \cite{RN1440, RN1444}, are necessary.

\subsection{Extending to binary vector recovery}
The extension of compressive sampling to recover not only sparse but also binary vectors is motivated by certain imaging applications--particularly structural defect imaging--where physical properties take on binary values. For instance, in the inspection of metal structures, a region may exhibit either the nominal electrical conductivity of the metal or zero conductivity, indicating an air-filled crack or void.

While few studies have addressed binary vector recovery, most existing approaches rely on convex optimization methods \cite{RN1219, RN1157, RN1155}. Common formulations include:
\begin{equation*}
    \begin{array}{ll}
        \minimize  & \vec{0}^\top \vec{x}  \\
        \subjectto & \mat{\Phi} \vec{x} = \vec{y}  \\
                   & \vec{0} \preceq \vec{x} \preceq \vec{1} ,
    \end{array}
    \qquad
    \begin{array}{ll}
        \minimize  & \vec{1}^\top \vec{x}  \\
        \subjectto & \mat{\Phi} \vec{x} = \vec{y}  \\
                   & \vec{0} \preceq \vec{x} \preceq \vec{1} ,
    \end{array}
    \qquad
    \begin{array}{ll}
        \minimize  & \| \vec{x} - 0.5\cdot\vec{1} \|_\infty  \\
        \subjectto & \mat{\Phi} \vec{x} = \vec{y} .  \\
                   &
    \end{array}
\end{equation*}
These formulations, however, do not account for cases where the measurement $\vec{y}$ is noisy, nor do they explore solutions through Bayesian inference methods.

This study extends compressive sampling by focusing on binary vector recovery, with particular application to imaging challenges in structural defect detection. The main contributions are as follows:
\begin{itemize}[nosep, leftmargin=*]
    \item \textbf{Convex optimization formulation (\cref{sec:convex-optimization}).} This study establishes that the native binary-constrained optimization problem is NP-hard. To tackle this, the binary constraint is relaxed to its convex hull, enabling convex optimization techniques to approximate a solution. The tightness of the convex relaxation is proved.

    \item \textbf{Approximate Bayesian inference formulation (\cref{sec:bayesian-inference}).} The study demonstrates that inferring the posterior probability of binary variables within a Bayesian network is also NP-hard. To address this, two approximate inference algorithms are developed, utilizing mean field approximation and approximate message passing.

    \item \textbf{Application to eddy current sensing (\cref{sec:application}).} The study shows how defect imaging with eddy current sensing can be formulated as a binary vector recovery problem. Examples are provided to demonstrate its application to defect imaging in metal plates and pipes.
\end{itemize}

%% file: contents/section2.tex
\section{Convex optimization for binary vector recovery}
\label{sec:convex-optimization}

The aim is to find a sparse and binary vector $\vec{x}$ that satisfies
\begin{equation}  \label{eq:single-measurement}
    \vec{y} = \mat{\Phi} \vec{x} + \vec{n},
\end{equation}
where $\vec{y} \in \mathbb{R}^M$ is the measured vector, $\mat{\Phi} \in \mathbb{R}^{M\times N}$ is the measurement matrix, and $\vec{n} \in \mathbb{R}^M$ is the \ac{IID} measurement noise. A natural optimization problem to solve is
\begin{equation}  \label{eq:opt-binary-constraint}
    \begin{array}{ll}
        \minimize  & \|\vec{x}\|_1  \\
        \subjectto & \| \mat{\Phi} \vec{x} - \vec{y} \|_2 \leq \epsilon  \\
                   & \vec{x} \in \{0,1\}^N .
    \end{array}
\end{equation}
However, this optimization problem is NP-hard. To show this, consider another NP-hard problem.
\begin{definition}  \label{def:subset-sum-problem}
    The \ac{SSP} asks whether, given integers $w_1, \dots, w_N$ and a target number $W$, there exists a subset of $\{ w_1, \dots, w_N \}$ that sums exactly to $W$.
\end{definition}
\noindent The \ac{SSP} is NP-complete \cite{RN1467}.
\begin{theorem}  \label{thm:opt-binary-constraint-NP-hard}
    The binary constrained $\ell_1$ minimization problem \eqref{eq:opt-binary-constraint} is NP-hard.
\end{theorem}
\begin{proof}
    To prove that the binary-constrained $\ell_1$ minimization problem \eqref{eq:opt-binary-constraint} is NP-hard, it can be shown that the \ac{SSP} reduces to it.
    Consider an instance of the optimization problem where $\mat{\Phi} = \bigl[ \begin{matrix} w_1 & \cdots & w_N \end{matrix} \bigr]$, $\vec{y} = \bigl[ W \bigr]$, and $\epsilon = 0$. In this case, the constraint $\|\mat{\Phi} \vec{x} - \vec{y}\|_2 \leq \epsilon$ simplifies to
    \begin{equation*}
        \Bigl\| \sum_{j=1}^N w_j x_j - W \Bigr\|_2 = 0 ,
    \end{equation*}
    which leads to the equation
    \begin{equation*}
        \sum_{i=1}^N w_j x_j = W .
    \end{equation*}
    The binary constraint $\vec{x} \in \{0,1\}^N$ directly corresponds to selecting a subset in the \ac{SSP}: $x_j = 1$ means $w_j$ is included in the subset, while $x_j = 0$ means it is not.
    While the \ac{SSP} aims to find a feasible subset, it can be incorporated into the optimization framework by minimizing $\|\vec{x}\|_1$, which corresponds to minimization the number of integers selected.
\end{proof}

This simple reduction shows that, because of the introduced binary constraint, problem \eqref{eq:opt-binary-constraint} is a hard problem. Nonetheless, one might wonder if additional assumptions on the problem could turn the binary constrained problem into a problem that is solvable in polynomial time.

\subsection{Relaxing the binary constraint}
A possible convex relaxation is to replace the binary constraint $\vec{x} \in \{0,1\}^N$ with its convex hull. This leads to the following optimization problem:
\begin{equation}  \label{eq:opt-unit-interval-constraint}
    \begin{array}{ll}
        \minimize  & \|\vec{x}\|_1  \\
        \subjectto & \| \mat{\Phi} \vec{x} - \vec{y} \|_2 \leq \epsilon  \\
                   & \vec{0} \preceq \vec{x} \preceq \vec{1} .
    \end{array}
\end{equation}
This relaxed problem can be solved to global optimum in polynomial time using convex optimization algorithms. The question is whether the solution to this relaxed problem coincide with the solution to the original binary-constrained problem.
To address this, the concept of \ac{BRIP} is introduced, which measures the ability of a matrix $\mat{\Phi}$ to approximately preserve the squared $\ell_2$ norm of sparse vectors with bounded entries. This property plays a critical role in analyzing whether the relaxed optimization problem yields a solution close to the original binary vector.

\begin{definition}
    A matrix $\mat{\Phi}$ satisfies the \ac{BRIP} of order $s$ with bounded restricted isometry constant $\delta_s^\mathrm{b}$ if, for all $s$-sparse vectors $\vec{x}$ satisfying $\|\vec{x}\|_\infty \leq 1$, the following inequality holds:
    \begin{equation}
        (1 - \delta_s^\mathrm{b}) \, \|\vec{x}\|_2^2 \leq \|\mat{\Phi}\vec{x}\|_2^2 \leq (1 + \delta_s^\mathrm{b}) \, \|\vec{x}\|_2^2 .
    \end{equation}
\end{definition}

\begin{theorem}  \label{thm:tight-convex}
    Suppose $\vec{y} = \mat{\Phi} \vec{x}_\mathrm{true} + \vec{n}$, where $\|\vec{n}\|_2 \leq \epsilon$. The true vector $\vec{x}_\mathrm{true}$ is binary and has $s$ nonzero entries. Let $\vec{x}^\star$ denote the solution to the optimization problem \eqref{eq:opt-unit-interval-constraint}.
    If $\delta_{2s}^\mathrm{b}(\mat{\Phi}) < \sqrt{2}-1$, then the solution $\vec{x}^\star$ satisfies
    \begin{equation}
        \|\vec{x}^\star - \vec{x}_\mathrm{true}\|_2 \leq C \epsilon ,
    \end{equation}
    where $C$ is a constant.
\end{theorem}
The proof of \cref{thm:tight-convex} is provided in \ref{app:tight-convex-proof}.
\cref{thm:tight-convex} states that the solution to the optimization problem \eqref{eq:opt-unit-interval-constraint} is close to the true binary vector $\vec{x}_\mathrm{true}$ in the $\ell_2$ norm, provided that the bounded restricted isometry constant $\delta_{2s}^\mathrm{b}$ of the measurement matrix $\mat{\Phi}$ is sufficiently small. In the case where $\epsilon$ is sufficiently small such that \mbox{$\|\vec{x}^\star - \vec{x}_\mathrm{true}\|_\infty \leq \|\vec{x}^\star - \vec{x}_\mathrm{true}\|_2 \leq C\epsilon < 0.5$}, thresholding the solution vector $\vec{x}^\star$ can exactly recover the binary vector $\vec{x}_\mathrm{true}$.

\begin{remark}
    The standard \ac{RIP} \cite{RN1200} with restricted isometry constant $\delta_s$ requires a matrix $\mat{\Phi}$ to satisfy
    \begin{equation*}
        (1 - \delta_s) \, \|\vec{x}\|_2^2 \leq \|\mat{\Phi}\vec{x}\|_2^2 \leq (1 + \delta_s) \, \|\vec{x}\|_2^2
    \end{equation*}
    for all $s$-sparse vectors, irrespective of their magnitudes. In contrast, the \ac{BRIP} imposes an additional constraint: it applies only to \mbox{$s$-sparse} vectors $\vec{x}$ satisfying $\|\vec{x}\|_\infty \leq 1$. This restriction narrows the set of vectors under consideration to a strict subset of all $s$-sparse vectors. By limiting the scope to vectors with bounded entries, the \ac{BRIP} condition is less stringent than the standard \ac{RIP}. This is because large entries in $\vec{x}$ can amplify distortions when the transformation $\mat{\Phi}$ is applied, potentially causing larger values for the restricted isometry constant $\delta_s$. Although verifying whether a specific matrix satisfies the \ac{RIP} or the \ac{BRIP} is NP-hard, the bounded nature of the \ac{BRIP} inherently increases the likelihood of a matrix $\mat{\Phi}$ satisfying the \ac{BRIP} compared to the \ac{RIP}.
\end{remark}

\subsection{Multiple measurement vectors with unknown variance}
In many real-world scenarios, measurement data is often composed of segments with varying noise characteristics, making \eqref{eq:single-measurement} insufficient. To better represent such cases, the aim is to find a binary vector $\vec{x}$ that satisfies
\begin{equation}  \label{eq:multiple-measurements}
    \vec{y}^{(l)} = \mat{\Phi}^{(l)} \vec{x} + \vec{n}^{(l)} \quad \forall l = 1,\dots,L ,
\end{equation}
where $\vec{y}^{(l)} \in \mathbb{R}^M$ are the measured vectors, $\mat{\Phi}^{(l)} \in \mathbb{R}^{M\times N}$ are the measurement matrices, and $\vec{n}^{(l)} \in \mathbb{R}^M$ are measurement noise vectors with different variances.
 Assume that the noise vectors are \ac{IID} Gaussian:
\begin{equation}
    \vec{n}^{(l)} \sim \mathcal{N}\bigl( \vec{0}, (\beta^{(l)})^{-1} \mat{I} \bigr) ,
\end{equation}
where $\beta^{(l)}$ represents the inverse-variance for each noise vector. The problem of recovering the binary vector $\vec{x}$ can be formulated as the following convex optimization problem:
\begin{equation}  \label{eq:opt-unit-interval-constraint-augmented}
    \begin{array}{ll}
        \minimize  & \|\vec{x}\|_1  \\
        \subjectto & \| \tilde{\mat{\Phi}} \vec{x} - \tilde{\vec{y}} \|_2 \leq \epsilon  \\
                   & \vec{0} \preceq \vec{x} \preceq \vec{1} ,
    \end{array}
\end{equation}
where
\begin{equation*}
    \tilde{\mat{\Phi}} = \begin{bmatrix} \sqrt{\beta^{(1)}} \mat{\Phi}^{(1)} \\ \vdots \\ \sqrt{\beta^{(L)}} \mat{\Phi}^{(L)} \end{bmatrix} , \quad
    \tilde{\vec{y}}    = \begin{bmatrix} \sqrt{\beta^{(1)}} \vec{y}^{(1)}    \\ \vdots \\ \sqrt{\beta^{(L)}} \vec{y}^{(L)}    \end{bmatrix} .
\end{equation*}
When $L = 1$, this optimization problem is analogous to the one in \eqref{eq:opt-unit-interval-constraint}. For $L > 1$, if a specific channel $l$ exhibits high noise, indicated by a small value of $\beta^{(l)}$, the impact of the corresponding measurement matrix $\mat{\Phi}^{(l)}$ and measurement vector $\vec{y}^{(l)}$ on the overall system is reduced due to the scaling applied in the augmented matrices and vectors.
In the optimization problem \eqref{eq:opt-unit-interval-constraint-augmented}, a suitable choice for $\epsilon$ is $c\sqrt{ML}$, where $c$ represents the number of standard deviations the noise exhibits.

However, what if the noise contents are unknown. One approach is to begin with an initial guess for all $\beta^{(l)}$ values and then iteratively refine them based on some heuristic. The process leads the iterative algorithm described in \cref{alg:convex-optimization}.\\
\begin{algorithm}[H]
    \caption{Convex optimization for binary vector recovery}
    \label{alg:convex-optimization}
    \SetKwInOut{Input}{Input}
    \SetKwInOut{Init}{Initialize}
    \SetKwInOut{Output}{Output}
    \SetKwRepeat{Do}{do}{while}

    \Input{$\mat{\Phi}^{(1)}, \dots, \mat{\Phi}^{(L)}, \vec{y}^{(1)}, \dots, \vec{y}^{(L)}$}
    \Init{$\beta^{(l)} = \beta_0 \ \ \forall l = 1,\dots,L$}

    \Do{$\vec{x}^\star$ not converged}{
        $\tilde{\mat{\Phi}} \gets \Bigl[ \begin{matrix}
            \sqrt{\beta^{(1)}} \bigl(\mat{\Phi}^{(1)}\bigr)^\top & \dots & \sqrt{\beta^{(L)}} \bigl(\mat{\Phi}^{(L)}\bigr)^\top
        \end{matrix} \Bigr]^\top$

        $\tilde{\vec{y}} \gets \Bigl[ \begin{matrix}
            \sqrt{\beta^{(1)}} \bigl(\vec{y}^{(1)}\bigr)^\top & \dots & \sqrt{\beta^{(L)}} \bigl(\vec{y}^{(L)}\bigr)^\top
        \end{matrix} \Bigr]^\top$

        $x^\star \gets \minimize \|\vec{x}\|_1 \quad \subjectto \|\tilde{\mat{\Phi}} \vec{x} - \tilde{\vec{y}}\|_2 \leq \epsilon,\ \vec{0}\preceq\vec{x}\preceq{1}$

        \ForEach{$l = 1, \dots, L$}{
            $\beta^{(l)} \gets M / \| \vec{y}^{(l)} - \mat{\Phi}^{(l)} \vec{x}^\star \|_2^2$
        }
    }
    \Output{$\vec{x}^\star$}
\end{algorithm}
\noindent In \cref{alg:convex-optimization}, the process begins by initializing all $\beta^{(l)}$ values to an initial guess $\beta_0$. The algorithm then alternates between solving the convex optimization problem \eqref{eq:opt-unit-interval-constraint-augmented} and updating each $\beta^{(l)}$ based on the squared error of the current estimate. While this $\beta^{(l)}$ update heuristic may appear arbitrary, the following claim provides a formal justification.
\begin{claim}  \label{thm:expectation-maximization-instance}
    \cref{alg:convex-optimization} is an instance of the \ac{EM} algorithm.
\end{claim}
\begin{proof}
    The \ac{EM} algorithm is an iterative method for finding estimates of parameters \cite{RN157}. Given the set of observed variables $\set{Y}$, unobserved variables $\set{X}$, and unknown parameters $\set{\Theta}$, the \ac{EM} algorithm  alternates between the E-step and the M-step to estimate the unknown parameters.
    \begin{itemize}[nosep, leftmargin=*, label=-]
        \item E-step:\hspace{4.6mm}$Q(\set{\Theta} | \set{\Theta}^{(\text{old})}) = \int p(\set{X} | \set{Y}, \set{\Theta}^{(\text{old})}) \ln p(\set{Y}, \set{X} | \set{\Theta}) \, d\set{X}$
        \item M-step:\quad$\set{\Theta}^{(\text{new})} = \argmax_{\set{\Theta}} Q(\set{\Theta} | \set{\Theta}^{(\text{old})})$
    \end{itemize}
    For our problem, let $\set{Y} = \{ \vec{y}^{(l)} \}_{l=1}^L$, $\set{X} = \{ \vec{x} \}$, and $\set{\Theta} = \{ \beta^{(l)} \}_{l=1}^L$. The \ac{EM} algorithm updates $\set{\Theta}$ as follows:
    \begin{equation*}
    \begin{split}
        \set{\Theta}^{(\text{new})}
        &= \argmax_{\set{\Theta}} \int p(\set{X} | \set{Y}, \set{\Theta}^{(\text{old})}) \ln p(\set{Y}, \set{X} | \set{\Theta}) \, d\set{X}  \\
        &= \argmax_{\set{\Theta}} \int p(\set{X} | \set{Y}, \set{\Theta}^{(\text{old})}) \left[ \ln p(\set{Y} | \set{X} , \set{\Theta}) + \ln p(\set{X} | \set{\Theta}) \right] \, d\set{X}  \\
        &\overset{3}{=} \argmax_{\set{\Theta}} \int p(\set{X} | \set{Y}, \set{\Theta}^{(\text{old})}) \left[ \ln p(\set{Y} | \set{X} , \set{\Theta}) + \ln p(\set{X}) \right] \, d\set{X}  \\
        &\overset{4}{=} \argmax_{\set{\Theta}} \int p(\set{X} | \set{Y}, \set{\Theta}^{(\text{old})}) \ln p(\set{Y} | \set{X} , \set{\Theta}) \, d\set{X} .
    \end{split}
    \end{equation*}
    Equality 3 arises because $\set{X}$ is independent of $\set{\Theta}$. Equality 4 holds because the second term does not dependent on $\set{\Theta}$. Substituting the specific variables into the expression results in
    \begin{equation*}
        (\beta^{(l)})^{(\text{new})}
        = \argmax_{\beta^{(l)}} \int p\Bigl( \vec{x} \big| \vec{y}^{(1)},\dots,\vec{y}^{(L)}, (\beta^{(1)})^{(\text{old})},\dots,(\beta^{(L)})^{(\text{old})} \Bigr) \ln p\Bigl( \vec{y}^{(l)} \big| \vec{x} , \beta^{(l)} \Bigr) \, d\vec{x} .
    \end{equation*}
    The posterior probability of $\vec{x}$ can be viewed as a probability mass centered at the convex optimization solution:
    \begin{equation*}
        p\Bigl( \vec{x} \big| \vec{y}^{(1)},\dots,\vec{y}^{(L)}, (\beta^{(1)})^{(\text{old})},\dots,(\beta^{(L)})^{(\text{old})} \Bigr)
        = \delta(\vec{x} - \vec{x}^\star) ,
    \end{equation*}
    where $\delta(\cdot)$ denotes the Dirac delta function. Additionally, the measurement likelihood is given by
    \begin{equation*}
        p\Bigl( \vec{y}^{(l)} \big| \vec{x} , \beta^{(l)} \Bigr)
        = \biggl( \frac{\beta^{(l)}}{2\pi} \biggr)^{M/2} \exp \biggl[ -\frac{1}{2} \beta^{(l)} \| \vec{y}^{(l)} - \mat{\Phi}^{(l)} \vec{x} \|_2^2 \biggr] .
    \end{equation*}
    Combining these leads to the following expression for the update of $\beta^{(l)}$:
    \begin{equation*}
        (\beta^{(l)})^{(\text{new})} = \frac{M}{\| \vec{y}^{(l)} - \mat{\Phi}^{(l)} \vec{x}^\star \|_2^2} .
    \end{equation*}
\end{proof}

%% file: contents/section3.tex
\section{Bayesian inference for binary vector recovery}
\label{sec:bayesian-inference}

In addition to convex optimization, Bayesian inference offers an alternative approach to recovering a binary vector $\vec{x}$ that satisfies
\begin{equation}  \label{eq:multiple-measurements2}
    \vec{y}^{(l)} = \mat{\Phi}^{(l)} \vec{x} + \vec{n}^{(l)} \quad \forall l = 1,\dots,L ,
\end{equation}
where $\vec{y}^{(l)} \in \mathbb{R}^M$, $\mat{\Phi}^{(l)} \in \mathbb{R}^{M \times N}$, and $\vec{n}^{(l)} \sim \mathcal{N}\bigl(\vec{0}, (\beta^{(l)})^{-1} \mat{I}\bigr)$.
The Bayesian network illustrated in \cref{fig:bayesian-network} represents the linear measurement model defined by \eqref{eq:multiple-measurements2}, where $x_j$ are binary variables and
\begin{equation}  \label{eq:measurement-likelihood}
    p(y_i^{(l)} | x_1, \dots, x_N, \beta^{(l)})
    = \sqrt{\frac{\beta^{(l)}}{2\pi}} \exp \biggl[ -\frac{1}{2} \beta^{(l)} \Bigl( y_i^{(l)} - \sum_{j=1}^N \phi_{i,j}^{(l)} x_j \Bigr)^2 \biggr] .
\end{equation}
\begin{figure}
    \centering
    \input{figures/bayesian-network.tex}
    \caption{A Bayesian network representing linear measurements of a binary vector. The vertices $x_j$ ($j=1,\dots,N$) are binary variables. The vertices $y_i^{(l)}$ ($i=1,\dots,M$) are Gaussian random variables with a mean given by a linear combination of $x_1, \dots, x_N$ and a variance of $1/\beta^{(l)}$.}
    \label{fig:bayesian-network}
\end{figure}
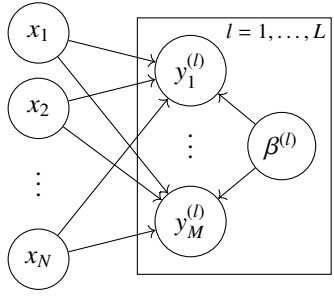
A Bayesian network allows for two types of inference: \ac{MAP} inference and conditional inference. However, performing exact inference of either type on the Bayesian network illustrated in \cref{fig:bayesian-network} is an NP-hard problem.

\begin{theorem}  \label{thm:MAP-inference-NP-hard}
    \Ac{MAP} inference on the Bayesian network illustrated in \cref{fig:bayesian-network} is NP-hard.
\end{theorem}
\begin{proof}
    It will be shown that the ability to perform \ac{MAP} inference enables answering the \ac{SSP} defined in \cref{def:subset-sum-problem}. Consider an instance of the Bayesian network where $M=1$ and $L=1$. Additionally, set $y_1^{(1)} = W$, $\phi_{1,j}^{(1)} = w_j$, $p(x_j) > 0$, and let $\beta^{(1)}$ be an arbitrarily large number. \ac{MAP} inference seeks the combination of binary value assignments that has the greatest posterior probability:
    \begin{equation*}
    \begin{split}
           \argmax_{x_1, \dots, x_N \in \{0,1\}} p(x_1, \dots, x_N | y_1^{(1)}, \beta^{(1)})
        &= \argmax_{x_1, \dots, x_N \in \{0,1\}} p(y_1^{(1)} | x_1, \dots, x_N, \beta^{(1)}) p(x_1) \dots p(x_N)  \\
        &= \argmax_{x_1, \dots, x_N \in \{0,1\}} \ln p(y_1^{(1)} | x_1, \dots, x_N, \beta^{(1)}) + \sum_{j=1}^N \ln p(x_j)  \\
        &= \argmax_{x_1, \dots, x_N \in \{0,1\}} -\frac{1}{2} \beta^{(1)} \Bigl( y_1^{(1)} - \sum_{j=1}^N \phi_{1,j}^{(1)} x_j \Bigr)^2 + \sum_{j=1}^N \ln p(x_j)  \\
        &\overset{4}{=} \argmax_{x_1, \dots, x_N \in \{0,1\}} -\frac{1}{2} \beta^{(1)} \Bigl( W - \sum_{j=1}^N w_j x_j \Bigr)^2
    \end{split}
    \end{equation*}
    Equality 4 holds because $\beta^{(1)}$ is arbitrarily large and $\ln p(x_j)$ is finite. The maximum occurs when the combination of $x_j$ values satisfy
    \begin{equation*}
        W - \sum_{j=1}^N w_j x_j = 0 .
    \end{equation*}
    The result of the \ac{MAP} inference, denoted as $x_1, \dots, x_N$, directly corresponds to selecting a subset in the \ac{SSP}: $x_j = 1$ means $w_j$ is included in the subset, while $x_j = 0$ means it is not.
    If the \ac{MAP} inference result satisfies $\sum_{j=1}^N w_j x_j = W$, then the answer to the \ac{SSP} is `yes'.
    Conversely, if the \ac{MAP} inference result does not satisfy $\sum_{j=1}^N w_j x_j = W$, the answer to the \ac{SSP} is `no'.
\end{proof}

\begin{theorem}  \label{thm:conditional-inference-NP-hard}
    Exact conditional inference on the Bayesian network illustrated in \cref{fig:bayesian-network} is NP-hard.
\end{theorem}
\begin{proof}
    Consider an instance of the Bayesian network where $M=1$ and $L=1$. Set $y_1^{(1)} = W$, $\phi_{1,j}^{(1)} = w_j$, $p(x_j) > 0$, and let $\beta^{(1)}$ be an arbitrarily large number. The posterior probability is given by:
    \begin{equation*}
    \begin{split}
        p(x_1, x_2 \dots, x_N | y_1^{(1)}, \beta^{(1)})
        &= \frac{p(y_1^{(1)} | x_1, x_2, \dots, x_N, \beta^{(1)}) p(x_1) p(x_2) \dots p(x_N)}{p(y_1^{(1)})}  \\
        &= \sqrt{\frac{\beta^{(1)}}{2\pi}} \exp \biggl[ -\frac{1}{2} \beta^{(1)} \Bigl( y_1^{(1)} - \sum_{j=1}^N \phi_{1,j}^{(1)} x_j \Bigr)^2 \biggr] \times \prod_{j=1}^N p(x_j) \times \frac{1}{p(y_1^{(1)})}  \\
        &= \sqrt{\frac{\beta^{(1)}}{2\pi}} \exp \biggl[ -\frac{1}{2} \beta^{(1)} \Bigl( W - \sum_{j=1}^N w_j x_j \Bigr)^2 \biggr] \times \prod_{j=1}^N p(x_j) \times \frac{1}{p(y_1^{(1)})}
    \end{split}
    \end{equation*}
    This probability is nonzero if and only if $W - \sum_{j=1}^N w_j x_j = 0$, because $\beta^{(1)}$ is arbitrarily large.
    And since
    \begin{equation*}
        p(x_1 | y_1^{(1)}, \beta^{(1)}) = \sum_{x_2\in\{0,1\}} \dots \sum_{x_N\in\{0,1\}} p(x_1, x_2, \dots, x_N | y_1^{(1)}, \beta^{(1)}) ,
    \end{equation*}
    the exact conditional inference query $\prob(x_1=1 | y_1^{(1)}, \beta^{(1)}) > 0$ if and only if there exists a subset of integers from $\{ w_1, w_2, \dots, w_N \}$ that includes $w_1$ and sums exactly to $W$.
    Determining whether such a subset exists is NP-hard because it is equivalent to a \ac{SSP} asking whether there is a subset of $\{ w_2, \dots, w_N \}$ that sums exactly to $W - w_1$.
\end{proof}

\cref{thm:MAP-inference-NP-hard} and \cref{thm:conditional-inference-NP-hard} shows that inference for the posterior probability of $\vec{x}$ on the Bayesian network is a hard problem. Nonetheless, approximate inference techniques could provide reasonable solutions in polynomial time.

\subsection{Information theory background}
\begin{definition}
    Let $p(\set{\Theta})$ and $q(\set{\Theta})$ represent two joint probability distributions over a set of variables $\set{\Theta}$.
    The \ac{KL}-divergence between $q(\set{\Theta})$ and $p(\set{\Theta})$ is defined as
    \begin{equation}  \label{eq:KL-divergence}
        \KL\bigl( q(\set{\Theta}) \Vert p(\set{\Theta}) \bigr) = \int q(\set{\Theta}) \ln \frac{q(\set{\Theta})}{p(\set{\Theta})} \, d\set{\Theta} .
    \end{equation}
\end{definition}
\begin{lemma}  \label{thm:KL-divergence-positive}
    The \ac{KL}-divergence satisfies $\KL\bigl( q(\set{\Theta}) \Vert p(\set{\Theta}) \bigr) \geq 0$, with equality if and only if $q(\set{\Theta}) = p(\set{\Theta})$.
\end{lemma}
\begin{proof}
    The proof can be found in \cite{RN1028}.
\end{proof}
\noindent \cref{thm:KL-divergence-positive} shows that $\KL\bigl( q(\set{\Theta}) \Vert p(\set{\Theta}) \bigr)$ can serve as a distance metric, and minimizing it amounts to finding a distribution $q(\set{\Theta})$ that is close to $p(\set{\Theta})$.
\begin{remark}
    The \ac{KL}-divergence is asymmetric, meaning that $\KL\bigl( q \Vert p \bigr) \neq \KL\bigl( p \Vert q \bigr)$.
    Minimizing $\KL\bigl( p \Vert q \bigr)$ results in a distribution $q$ that matches the expected sufficient statistics of the true distribution $p$ \cite{koller2009pgm}. However, this projection is often computationally intractable.
    When using the other projection, the approximation distribution $q$ is typically selected to have a simple form, allowing the minimization of $\KL\bigl( q \Vert p \bigr)$ to be computationally tractable.
\end{remark}

\begin{proposition}  \label{thm:mean-field-approx}
    Consider a set of unobserved variables $\set{\Theta}$ and observed variables $\set{Y}$.
    Suppose the set $\set{\Theta}$ is partitioned into disjoint subsets $\set{\Theta}_1, \set{\Theta}_2, \dots$. Assume $q(\set{\Theta})$ is chosen as the class of joint distributions that factorizes over these partitions, referred to as mean field distributions \cite{RN157}:
    \begin{equation}  \label{eq:mean-field}
        q(\set{\Theta}) = \prod_k q(\set{\Theta}_k) ,
    \end{equation}
    where $q(\set{\Theta}_k)$ is any arbitrary probability distribution over the variables $\set{\Theta}_k$. By keeping all of $\{ q(\set{\Theta}_k) \}_{k\neq\ell}$ fixed, the unique minimizer of $\KL\bigl( q(\set{\Theta}) \Vert p(\set{\Theta} | \set{Y}) \bigr)$ with respect to $q(\set{\Theta}_\ell)$ satisfies
    \begin{equation}  \label{eq:mean-field-update-rule}
    \begin{split}
        \ln q(\set{\Theta}_\ell)
        &= \int \ln p(\set{Y}, \set{\Theta}) \bigl( \prod_{k\neq\ell} q(\set{\Theta}_k) \, d\set{\Theta}_k \bigr) + \const  \\
        &= \expval_{\prod_{k\neq\ell} q(\set{\Theta}_k)} \bigl[ \ln p(\set{Y}, \set{\Theta}) \bigr] + \const .
    \end{split}
    \end{equation}
\end{proposition}
\begin{proof}
    Start by defining the energy functional $F\bigl( p(\set{Y}, \set{\Theta}), q(\set{\Theta}) \bigr)$ as
    \begin{equation*}
        F\bigl( p(\set{Y}, \set{\Theta}), q(\set{\Theta}) \bigr) = \int q(\set{\Theta}) \ln \frac{p(\set{Y}, \set{\Theta})}{q(\set{\Theta})} \, d\set{\Theta} .
    \end{equation*}
    The following identity always holds:
    \begin{equation*}
    \begin{split}
        \KL\bigl( q(\set{\Theta}) \Vert p(\set{\Theta} | \set{Y}) \bigr) + F\bigl( p(\set{Y}, \set{\Theta}), q(\set{\Theta}) \bigr)
        &= \int q(\set{\Theta}) \ln \frac{q(\set{\Theta})}{p(\set{\Theta} | \set{Y})} \, d\set{\Theta} + \int q(\set{\Theta}) \ln \frac{p(\set{Y}, \set{\Theta})}{q(\set{\Theta})} \, d\set{\Theta}  \\
        &= \int q(\set{\Theta}) \ln \frac{q(\set{\Theta})}{p(\set{\Theta} | \set{Y})} \frac{p(\set{Y}, \set{\Theta})}{q(\set{\Theta})} \, d\set{\Theta}  \\
        &= \int q(\set{\Theta}) \ln p(\set{Y}) \, d\set{\Theta}  \\
        &= \ln p(\set{Y}) .
    \end{split}
    \end{equation*}
    Since $\ln p(\set{Y})$ is independent of $q(\set{\Theta})$, minimizing $\KL\bigl( q(\set{\Theta}) \Vert p(\set{\Theta} | \set{Y}) \bigr)$ is equivalent to maximizing $F\bigl( p(\set{Y}, \set{\Theta}), q(\set{\Theta}) \bigr)$ with respect to $q(\set{\Theta})$. Substituting the factorized form of $q(\set{\Theta})$ from \eqref{eq:mean-field} into $F\bigl( p(\set{Y}, \set{\Theta}), q(\set{\Theta}) \bigr)$ yields
    \begin{equation*}
    \begin{split}
        F\bigl( p(\set{Y}, \set{\Theta}), q(\set{\Theta}) \bigr)
        &= \int \bigl( \prod_k q(\set{\Theta}_k)\bigr) \bigl( \ln p(\set{Y}, \set{\Theta}) - \sum_\ell \ln q(\set{\Theta}_\ell) \bigr) \bigl( \prod_k \, d\set{\Theta}_k \bigr)  \\
        &= \int \ln p(\set{Y}, \set{\Theta}) \bigl( \prod_k q(\set{\Theta}_k) \, d\set{\Theta}_k \bigr) - \sum_\ell \int \ln q(\set{\Theta}_\ell) \bigl( \prod_k q(\set{\Theta}_k) \, d\set{\Theta}_k \bigr)  \\
        &= \iint \Bigl( \ln p(\set{Y}, \set{\Theta}) \bigl( \prod_{k\neq\ell} q(\set{\Theta}_k) \, d\set{\Theta}_k \bigr) \Bigr) q(\set{\Theta}_\ell) \, d\set{\Theta}_\ell - \sum_\ell \int q(\set{\Theta}_\ell) \ln q(\set{\Theta}_\ell) \, d\set{\Theta}_\ell  \\
        &= \iint \Bigl( \ln p(\set{Y}, \set{\Theta}) \bigl( \prod_{k\neq\ell} q(\set{\Theta}_k) \, d\set{\Theta}_k \bigr) \Bigr) q(\set{\Theta}_\ell) \, d\set{\Theta}_\ell - \int q(\set{\Theta}_\ell) \ln q(\set{\Theta}_\ell) \, d\set{\Theta}_\ell - \sum_{k\neq\ell} \int q(\set{\Theta}_k) \ln q(\set{\Theta}_k) \, d\set{\Theta}_k  \\
        &= \int q(\set{\Theta}_\ell) \ln \frac{\tilde{p}(\set{\Theta}_\ell)}{q(\set{\Theta}_\ell)} \, d\set{\Theta}_\ell - \sum_{k\neq\ell} \int q(\set{\Theta}_k) \ln q(\set{\Theta}_k) \, d\set{\Theta}_k  \\
        &= -\KL \bigl( q(\set{\Theta}_\ell) \Vert \tilde{p}(\set{\Theta}_\ell) \bigr) - \sum_{k\neq\ell} \int q(\set{\Theta}_k) \ln q(\set{\Theta}_k) \, d\set{\Theta}_k ,
    \end{split}
    \end{equation*}
    where
    \begin{equation*}
        \ln \tilde{p}(\set{\Theta}_\ell) = \int \ln p(\set{Y}, \set{\Theta}) \bigl( \prod_{k\neq\ell} q(\set{\Theta}_k) \, d\set{\Theta}_k \bigr) + \const .
    \end{equation*}
    Obviously, $F\bigl( p(\set{Y}, \set{\Theta}), q(\set{\Theta}) \bigr)$ is maximized with respect to $q(\set{\Theta}_\ell)$ when $\KL \left( q(\set{\Theta}_\ell) \Vert \tilde{p}(\set{\Theta}_\ell) \right) = 0$. According to \cref{thm:KL-divergence-positive}, this implies that $q(\set{\Theta}_\ell) = \tilde{p}(\set{\Theta}_\ell)$.
\end{proof}
\cref{thm:mean-field-approx} provides a means for approximating the posterior probability $p(\set{\Theta} | \set{Y})$ by another probability density $q(\set{\Theta})$, which factorizes over the subsets of $\set{\Theta}$ as given in \eqref{eq:mean-field}.  While this mean field approximation cannot capture the correlation between variables in different subsets, the trade-off enables polynomial time algorithms.

\subsection{Inference using mean field approximation}
To apply \cref{thm:mean-field-approx} for inferencing on the Bayesian network illustrated in \cref{fig:bayesian-network}, it is necessary to first specify the prior probabilities.
Gamma priors are chosen for $\beta^{(l)}$:
\begin{equation}
    p(\beta^{(l)})
    = \mathrm{Gamma}(\beta | c,d)
    = \frac{d^c}{\Gamma(c)} (\beta^{(l)})^{c-1} e^{-d \beta^{(l)}} ,
\end{equation}
where $\Gamma(\cdot)$ is the gamma function. Parameters $c$ and $d$ are usually set to zero to obtain non-informative priors \cite{RN64}.
To enforce a binary solution vector, a Bernoulli prior is chosen for $x_j$:
\begin{equation}  \label{eq:distribution-2}
    p(x_j | \pi_j)
    = \mathrm{Bernoulli}(x_j | \pi_j)
    = \pi_j^{x_j} (1-\pi_j)^{1-x_j} .
\end{equation}
The hyperprior of $\pi_j$ is chosen as a beta distribution since it the conjugate prior of the Bernoulli distribution:
\begin{equation}
    p(\pi_j)
    = \mathrm{Beta}(\pi_j | a,b)
    = \frac{1}{\mathrm{B}(a,b)} \pi_j^{a-1} (1 - \pi_j)^{b-1} ,
\end{equation}
where $\mathrm{B}(\cdot,\cdot)$ is the beta function.
Furthermore, the measurement likelihood formula \eqref{eq:measurement-likelihood} written in vector form is
\begin{equation}  \label{eq:distribution-4}
    p(\vec{y}^{(l)} | \vec{x}, \beta^{(l)})
    = \biggl( \frac{\beta^{(l)}}{2\pi} \biggr)^{M/2} \exp \biggl[ -\frac{1}{2} \beta^{(l)} \| \vec{y}^{(l)} - \mat{\Phi}^{(l)} \vec{x} \|_2^2 \biggr] .
\end{equation}
The set of observed variables is $\set{Y} = \{ \vec{y}^{(l)} \}_{l=1}^L$ and the set of unobserved variables is $\set{\Theta} = \{ \beta^{(l)} \}_{l=1}^L \cup \{ x_j, \pi_j \}_{j=1}^N$. The joint distribution of $\set{Y}$ and $\set{\Theta}$ is given by
\begin{equation}
    p(\set{Y}, \set{\Theta}) =
    \prod_{l=1}^L p(\vec{y}^{(l)} | \vec{x}, \beta^{(l)}) p(\beta^{(l)})
    \prod_{j=1}^N p(x_j | \pi_j) p(\pi_j) .
\end{equation}
The mean field distribution is
\begin{equation}
    q(\set{\Theta}) = \prod_{l=1}^L q(\beta^{(l)}) \prod_{j=1}^N q(x_j) q(\pi_j) .
\end{equation}
With all the necessary variables and distributions defined, \cref{thm:mean-field-approx} can be applied for mean field approximation inference.

\paragraph*{Update of $q(\beta^{(l)})$}
According to the update rule \eqref{eq:mean-field-update-rule},
\begin{equation}  \label{eq:mean-field-derivation-first}
\begin{split}
    \ln q(\beta^{(l)})
    &= \expval_{\prod_{j=1}^N q(x_j)} \bigl[ \ln p(\vec{y}^{(l)} | \vec{x}, \beta^{(l)}) p(\beta^{(l)}) \bigr] + \const  \\
    &= \Bigl( \underbrace{c + \frac{M}{2}}_{\tilde{c}} - 1 \Bigr) \ln \beta^{(l)}  - \Bigl( \underbrace{d + \frac{1}{2} \expval_{{q(\vec{x})}} \bigl[ \| \vec{y}^{(l)} - \mat{\Phi}^{(l)} \vec{x} \|_2^2 \bigr]}_{\tilde{d}} \Bigr) \beta^{(l)} + \const  \\
    &= \ln \mathrm{Gamma} ( \beta^{(l)} | \tilde{c}, \tilde{d} ).
\end{split}
\end{equation}
The expected value of the quadratic form is given by
\begin{equation}
    \expval_{\prod_{j=1}^N q(x_j)} \bigl[ \| \vec{y}^{(l)} - \mat{\Phi}^{(l)} \vec{x} \|_2^2 \bigr]
    = \| \vec{y}^{(l)} - \mat{\Phi}^{(l)} \hat{\vec{x}} \|_2^2
    + \trace( \diag( \hat{\vec{x}} \odot (\vec{1} - \hat{\vec{x}}) ) \mat{\Phi}^{(l)^\top} \mat{\Phi}^{(l)} ) ,
\end{equation}
where $\hat{x}$ denotes the mean of $\vec{x}$ with respect to $\prod_{j=1}^N q(x_j)$.
The mean of $\beta^{(l)}$, denoted as $\hat{\beta}^{(l)}$, is given by
\begin{equation}
    \hat{\beta}^{(l)} = \expval_{q(\beta^{(l)})} \bigl[ \beta^{(l)} \bigr] = \frac{\tilde{c}}{\tilde{d}} .
\end{equation}

\paragraph*{Update of $q(x_j)$}
First define
\begin{equation}
    \vec{y}_{\sim j}^{(l)} = \vec{y}^{(l)} - \sum_{k \neq j} \vec{\phi}_{:,k}^{(l)} x_k ,
\end{equation}
where $\vec{\phi}_{:,k}$ denotes the $k$th column of $\mat{\Phi}$.
According to the update rule \eqref{eq:mean-field-update-rule},
\begin{equation}
\begin{split}
    \ln q(x_j)
    &= \expval_{q(\pi_j) \prod_{l=1}^L q(\beta^{(l)})} \bigl[ \ln \prod_{l=1}^L p(\vec{y}_{\sim j}^{(l)} | x_j, \beta^{(l)}) p(x_j | \pi_j) \bigr] + \const  \\
    &= \sum_{l=1}^L \Bigl( -\frac{1}{2} \expval_{q(\beta^{(l)})} \bigl[ \beta^{(l)} \bigr] \ \| \vec{y}_{\sim j}^{(l)} - \vec{\phi}_{:,j}^{(l)} x_j \|_2^2 \Bigr) +
       x_j \, \expval_{q(\pi_j)} \bigl[ \ln \pi_j \bigr] + (1-x_j) \, \expval_{q(\pi_j)} \bigl[ \ln (1-\pi_j) \bigr] + \const  \\
    &= \sum_{l=1}^L -\frac{1}{2} \hat{\beta}^{(l)} \left( x_j^2 \vec{\phi}_{:,j}^{(l)^\top} \vec{\phi}_{:,j}^{(l)} -2 x_j \vec{\phi}_{:,j}^{(l)^\top} \vec{y}_{\sim j}^{(l)} \right) +
       x_j \, \expval_{q(\pi_j)} \bigl[ \ln \pi_j \bigr] + (1-x_j) \, \expval_{q(\pi_j)} \bigl[ \ln (1-\pi_j) \bigr] + \const .
\end{split}
\end{equation}
The probabilities of $x_j$ being $1$ and $0$ are
\begin{subequations}
\begin{align}
    \prob(x_j=1) &\propto
    \exp \Bigl[ \expval_{q(\pi_j)} \bigl[ \ln \pi_j \bigr] \Bigr] \times
    \prod_{l=1}^L \exp\Bigl[ -\frac{1}{2} \hat{\beta}^{(l)} \Bigl( \vec{\phi}_{:,j}^{(l)^\top} \vec{\phi}_{:,j}^{(l)} -2 \vec{\phi}_{:,j}^{(l)^\top} \vec{y}_{\sim j}^{(l)} \Bigr) \Bigr] ,  \\
    \prob(x_j=0) &\propto \exp \Bigl[ \expval_{q(\pi_j)} \bigl[ \ln (1-\pi_j) \bigr] \Bigr] .
\end{align}
\end{subequations}
Then, the mean of $x_j$ with respect to $q(x_j)$, denoted as $\hat{x}_j$, is given by
\begin{equation}
    \hat{x}_j = \frac{\prob(x_j=1)}{\prob(x_j=1) + \prob(x_j=0)} .
\end{equation}

\paragraph*{Update of $q(\pi_j)$}
According to the update rule \eqref{eq:mean-field-update-rule},
\begin{equation}
\begin{split}
    \ln q(\pi_j)
    &= \expval_{q(x_j)} \bigl[ \ln p(x_j | \pi_j) p(\pi_j) \bigr] + \const  \\
    &= \expval_{q(x_j)} [x_j] \ln \pi_j + \expval_{q(x_j)} [1-x_j] \ln (1-\pi_j) +
       (a-1) \ln \pi_j + (b-1) \ln (1-\pi_j) + \const  \\
    &= ( \underbrace{\hat{x}_j + a}_{\tilde{a}_j} - 1 ) \ln \pi_j +
       ( \underbrace{(1-\hat{x}_j) + b}_{\tilde{b}_j} - 1 ) \ln (1-\pi_j) + \const  \\
    &= \ln \mathrm{Beta} ( \pi_j | \tilde{a}_j, \tilde{b}_j ) .
\end{split}
\end{equation}
The moments of the logarithm of $\pi_j$ and $1-\pi_j$ with respect to $q(\pi_j)$ are
\begin{subequations}  \label{eq:mean-field-derivation-last}
\begin{align}
    \expval_{q(\pi_j)} [\ln    \pi_j ] &= \psi(\tilde{a}_j) - \psi(\tilde{a}_j+\tilde{b}_j) ,  \\
    \expval_{q(\pi_j)} [\ln (1-\pi_j)] &= \psi(\tilde{b}_j) - \psi(\tilde{a}_j+\tilde{b}_j) ,
\end{align}
\end{subequations}
where $\psi(z) = \frac{d}{dz} \ln \Gamma(z)$ is the digamma function.

Following the derivations \eqref{eq:mean-field-derivation-first} through \eqref{eq:mean-field-derivation-last}, the mean field approximation procedure is summarized in \cref{alg:mean-field-approx}.\\
\begin{algorithm}[H]
    \caption{Mean field approximate inference for binary vector recovery}
    \label{alg:mean-field-approx}
    \SetKwInOut{Input}{Input}
    \SetKwInOut{Init}{Initialize}
    \SetKwInOut{Output}{Output}
    \SetKwRepeat{Do}{do}{while}

    \Input{$\mat{\Phi}^{(1)}, \dots, \mat{\Phi}^{(L)}, \vec{y}^{(1)}, \dots, \vec{y}^{(L)}$}
    \Init{$\hat{\vec{x}} = 0.5 \cdot \vec{1}$}
    \Do{$\hat{\vec{x}}$ not converged}{
        \ForEach{$l = 1, \dots, L$}{
            $\hat{\beta}^{(l)} \gets  M / \bigl( \| \vec{y}^{(l)} - \mat{\Phi}^{(l)} \hat{\vec{x}} \|_2^2 + \trace( \diag( \hat{\vec{x}} \odot (\vec{1} - \hat{\vec{x}}) ) \mat{\Phi}^{(l)^\top} \mat{\Phi}^{(l)} ) \bigr)$  \\
        }
        \ForEach{$j = 1, \dots, N$}{
            $\expval_{q(\pi_j)} [\ln    \pi_j ] \gets \psi(\hat{x}_j + a) - \psi(a+b+1)$  \\
            $\expval_{q(\pi_j)} [\ln (1-\pi_j)] \gets \psi(1-\hat{x}_j + b) - \psi(a+b+1)$  \\
            $\vec{y}_{\sim j}^{(l)} \gets \vec{y}^{(l)} - \sum_{k \neq j} \vec{\phi}_{:,k}^{(l)} \hat{x}_k$  \\
            $\xi_{j1} \gets \expval_{q(\pi_j)} [\ln \pi_j] + \sum_{l=1}^L -\frac{1}{2} \hat{\beta}^{(l)} \bigl( \vec{\phi}_{:,j}^{(l)^\top} \vec{\phi}_{:,j}^{(l)} -2 \vec{\phi}_{:,j}^{(l)^\top} \vec{y}_{\sim j}^{(l)} \bigr)$  \\
            $\xi_{j0} \gets \expval_{q(\pi_j)} [\ln (1-\pi_j)]$  \\
            $\hat{x}_j \gets 1 / ( 1 + \exp(\xi_{j0} - \xi_{j1}))$  \\
        }
    }
    \Output{$\hat{\vec{x}}$}
\end{algorithm}

\begin{claim}
    \cref{alg:mean-field-approx} converges globally.
\end{claim}
\begin{proof}
    From \cref{thm:mean-field-approx}, sequentially applying the update rule in \eqref{eq:mean-field-update-rule} ensures that the \ac{KL}-divergence $\KL\bigl( q(\set{\Theta}) \Vert p(\set{\Theta} | \set{Y}) \bigr)$ is monotonically nonincreasing. Since the \ac{KL}-divergence is bounded below by zero, as stated in \cref{thm:KL-divergence-positive}, the iterations must converge. Furthermore, \cref{thm:mean-field-approx} guarantees the uniqueness of the minimizers, so the convergence point must be a fixed point.
\end{proof}

\subsection{Inference using approximate message passing}
Inspired by expectation propagation \cite{minka2001}, another approximate inference algorithm can be developed using message passing on cluster graphs \cite{koller2009pgm}. For convenience, the prior and conditional probabilities from \eqref{eq:distribution-2} through \eqref{eq:distribution-4} are repeated here:
\begin{equation*}
    p(x_j | \pi_j)
    = \mathrm{Bernoulli}(x_j | \pi_j)
    = \pi_j^{x_j} (1-\pi_j)^{1-x_j} ,
\end{equation*}
\begin{equation*}
    p(\pi_j)
    = \mathrm{Beta}(\pi_j | a,b)
    = \frac{1}{\mathrm{B}(a,b)} \pi_j^{a-1} (1 - \pi_j)^{b-1} ,
\end{equation*}
\begin{equation*}
    p(\vec{y}^{(l)} | \vec{x})
    = \biggl( \frac{\beta^{(l)}}{2\pi} \biggr)^{M/2} \exp \biggl[ -\frac{1}{2} \beta^{(l)} \| \vec{y}^{(l)} - \mat{\Phi}^{(l)} \vec{x} \|_2^2 \biggr] .
\end{equation*}
Here, $\beta^{(l)}$ are treated as parameters instead of random variables.
\cref{fig:cluster-graph} illustrates a cluster graph corresponding to these probability distributions. The leftmost cluster has an initial potential of $\prod_{j=1}^N p(\pi_j)$. The central cluster has an initial potential of $\prod_{j=1}^N p(x_j | \pi_j)$. The rightmost cluster has an initial potential of $\prod_{l=1}^L p(\vec{y}^{(l)} | \vec{x})$.
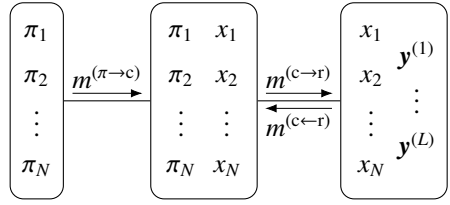
\begin{figure}
    \centering
    \input{figures/cluster-graph.tex}
    \caption{A cluster graph where each vertex corresponds to a subset of variables. Priors and conditionals are associated with specific clusters. The arrows indicate the direction of message passing.}
    \label{fig:cluster-graph}
\end{figure}

The message passing algorithm operates on cluster graphs by facilitating communication between clusters. In general, message passing operates in the following steps:
\begin{enumerate}[nosep, leftmargin=*]
    \item Each cluster receives messages from its neighboring clusters.
    \item The cluster multiplies its initial potential with all incoming messages, forming the cluster's belief.
    \item The cluster optionally marginalizes the belief or projects the belief onto another distribution.
    \item Finally, the cluster sends a message to each neighboring cluster. The message is calculated by dividing the cluster's belief by the incoming message from that neighbor.
\end{enumerate}

All distributions and messages over the variables $x_1, \dots, x_N$ are represented in the exponential family as follows:
\begin{equation}  \label{eq:exponential-family}
    \underbrace{\prod_{j=1}^N (1-\hat{x}_j)}_{\frac{1}{Z}} \times \exp \biggl(
        \underbrace{\begin{bmatrix} x_1 \\ \vdots \\ x_N \end{bmatrix}^\top}_{\vec{x}^\top}
        \underbrace{\begin{bmatrix} \ln\frac{\hat{x}_1}{1-\hat{x}_1} \\ \vdots \\ \ln\frac{\hat{x}_N}{1-\hat{x}_N} \end{bmatrix}}_{\vec{\eta}} \biggr) ,
\end{equation}
where $\vec{x}$ is the sufficient statistics, $\vec{\eta}$ is the natural parameter, $Z$ is the normalization constant, and $\hat{x}_j$ is the mean of the Bernoulli-distributed variable $x_j$.
The natural parameter for the messages $m^{(\text{c} \to \text{r})}$ and $m^{(\text{c} \leftarrow \text{r})}$ are defined as $\vec{\eta}^{(\text{c} \to \text{r})}$ and $\vec{\eta}^{(\text{c} \leftarrow \text{r})}$, resulting in
\begin{equation*}
    m^{(\text{c} \to        \text{r})} \propto \exp(\vec{x}^\top \vec{\eta}^{(\text{c} \to        \text{r})}) , \quad
    m^{(\text{c} \leftarrow \text{r})} \propto \exp(\vec{x}^\top \vec{\eta}^{(\text{c} \leftarrow \text{r})}) .
\end{equation*}

In the cluster graph illustrated in \cref{fig:cluster-graph}, the leftmost cluster's outgoing message, denoted as $m^{(\pi \to \text{c})}$, is given by
\begin{equation}  \label{eq:message-passing-derivation-first}
    m^{(\pi \to \text{c})} = \prod_{j=1}^N p(\pi_j) .
\end{equation}
The central cluster receives messages $m^{(\pi \to \text{c})}, m^{(\text{c} \leftarrow \text{r})}$ and multiplies them with its initial potential, resulting in the belief
\begin{equation}
    m^{(\pi \to \text{c})} m^{(\text{c} \leftarrow \text{r})} \times \prod_{j=1}^N p(x_j | \pi_j)
    = m^{(\text{c} \leftarrow \text{r})} \times \prod_{j=1}^N p(x_j | \pi_j) p(\pi_j) .
\end{equation}
This belief, when marginalized over $\pi_j$, results in
\begin{equation}
\begin{split}
    m^{(\text{c} \leftarrow \text{r})} \times \prod_{j=1}^N \int_0^1 p(x_j | \pi_j) p(\pi_j) \, d\pi_j
    &= m^{(\text{c} \leftarrow \text{r})} \times \prod_{j=1}^N \frac{1}{\mathrm{B}(a,b)} \int_0^1 \pi_j^{a+x_j-1} (1-\pi_j)^{b+(1-x_j)-1} \, d\pi_j  \\
    &= m^{(\text{c} \leftarrow \text{r})} \times \prod_{j=1}^N \frac{\mathrm{B}(a+x_j, b+1-x_j)}{\mathrm{B}(a,b)} .
\end{split}
\end{equation}
This marginal is then projected onto the exponential family defined in \eqref{eq:exponential-family}, and is given by
\begin{equation}
    \argmin_{\vec{\eta}}\ \KL \biggl( m^{(\text{c} \leftarrow \text{r})} \times \prod_{j=1}^N \frac{\mathrm{B}(a+x_j, b+1-x_j)}{\mathrm{B}(a,b)} \,\Big\Vert\, \frac{1}{Z} \exp(\vec{x}^\top \vec{\eta}) \biggr)
    = \vec{\eta}^{(\text{c} \leftarrow \text{r})} + \ln \frac{\mathrm{B}(a+1,b)}{\mathrm{B}(a,b+1)} \cdot \vec{1}
\end{equation}
Finally, the message passed from the central cluster to the rightmost cluster, denoted as $m^{(\text{c} \to \text{r})}$, is computed as follows:
\begin{equation}
\begin{split}
    m^{(\text{c} \to \text{r})}
    &\propto \frac{\exp\Bigl( \vec{x}^\top \Bigl( \vec{\eta}^{(\text{c} \leftarrow \text{r})} + \ln \frac{\mathrm{B}(a+1,b)}{\mathrm{B}(a,b+1)} \cdot \vec{1} \Bigr) \Bigr)}{m^{(\text{c} \leftarrow \text{r})}}  \\
    &= \exp\Bigl( \vec{x}^\top \Bigl( \ln \frac{\mathrm{B}(a+1,b)}{\mathrm{B}(a,b+1)} \cdot \vec{1} \Bigr) \Bigr) .
\end{split}
\end{equation}
The message $m^{(\text{c} \to \text{r})}$ is multiplied with the rightmost cluster's initial potential and projected onto the exponential family defined in \eqref{eq:exponential-family}, resulting in
\begin{equation}  \label{eq:approx-message-passing-M-project}
    \vec{\eta}^{(\text{r})}
    = \argmin_{\vec{\eta}}\ \KL \biggl( \frac{1}{Z} \exp(\vec{x}^\top \vec{\eta}) \,\Big\Vert\, \exp(\vec{x}^\top \vec{\eta}^{(\text{c} \to \text{r})}) \prod_{l=1}^L p(\vec{y}^{(l)} | \vec{x}) \biggr) .
\end{equation}
This projection is chosen because the other projection
\begin{equation*}
    \argmin_{\vec{\eta}}\ \KL \biggl( \exp(\vec{x}^\top \vec{\eta}^{(\text{c} \to \text{r})}) \prod_{l=1}^L p(\vec{y}^{(l)} | \vec{x}) \,\Big\Vert\, \frac{1}{Z} \exp(\vec{x}^\top \vec{\eta}) \biggr)
\end{equation*}
is computationally intractable.
The minimization \eqref{eq:approx-message-passing-M-project} can be achieved by employing \cref{thm:mean-field-approx}. According to the update rule \eqref{eq:mean-field-update-rule},
\begin{equation}
\begin{split}
    \ln q(x_j)
    &= \expval_{\prod_{k \neq j} q(x_k)} \biggl[ \ln\, \Bigl( \exp(\vec{x}^\top \vec{\eta}^{(\text{c} \to \text{r})}) \prod_{l=1}^L p(\vec{y}^{(l)} | \vec{x}) \Bigr) \biggr] + \const  \\
    &= \expval_{\prod_{k \neq j} q(x_k)} \biggl[ \vec{x}^\top \vec{\eta}^{(\text{c} \to \text{r})} + \sum_{l=1}^L -\frac{1}{2} \beta^{(l)} \| \vec{y}^{(l)} - \mat{\Phi}^{(l)} \vec{x} \|_2^2 \biggr] + \const  \\
    &= \expval_{\prod_{k \neq j} q(x_k)} \biggl[ \sum_{l=1}^L -\frac{1}{2} \beta^{(l)} \| \vec{y}^{(l)} - \vec{\phi}_{:,j}^{(l)} x_j - \mat{\Phi}_{:,\sim j}^{(l)} \vec{x}_{\sim j} \|_2^2 + \vec{x}_{\sim j}^\top \vec{\eta}_{\sim j}^{(\text{c} \to \text{r})} + x_j \eta_j^{(\text{c} \to \text{r})}  \biggr] + \const  \\
    &= \sum_{l=1}^L -\frac{1}{2} \beta^{(l)} \Bigl[ \| \vec{y}^{(l)} - \vec{\phi}_{:,j}^{(l)} x_j - \mat{\Phi}_{:,\sim j}^{(l)} \hat{\vec{x}}_{\sim j} \|_2^2 + \trace( \diag( \hat{\vec{x}}_{\sim j} \odot (\vec{1} - \hat{\vec{x}}_{\sim j}) ) \mat{\Phi}_{:,\sim j}^{(l)^\top} \mat{\Phi}_{:,\sim j}^{(l)} ) \Bigr] + \hat{\vec{x}}_{\sim j}^\top \vec{\eta}_{\sim j}^{(\text{c} \to \text{r})} + x_j \eta_j^{(\text{c} \to \text{r})} + \const ,
\end{split}
\end{equation}
where
\begin{itemize}[nosep, leftmargin=*, label={}]
    \item $\vec{\phi}_{:,j}^{(l)}$ denotes the $j$th column of $\mat{\Phi}^{(l)}$,
    \item $\mat{\Phi}_{:,\sim j}^{(l)}$ is the matrix $\mat{\Phi}^{(l)}$ excluding the $j$th column,
    \item $\vec{x}_{\sim j}$ is the vector $\vec{x}$ without the $j$th entry,
    \item $\vec{\eta}_{\sim j}^{(\text{c} \to \text{r})}$ is the vector $\vec{\eta}^{(\text{c} \to \text{r})}$ without the $j$th entry, and
    \item $\hat{\vec{x}}_{\sim j}$ represents the expected value of $\vec{x}_{\sim j}$ with respect to $\prod_{k \neq j} q(x_k)$.
\end{itemize}
The probabilities of $x_j$ being $1$ and $0$ are
\begin{subequations}
\begin{align}
    \prob(x_j=1) &\propto \exp \biggl[ \sum_{l=1}^L -\frac{1}{2} \beta^{(l)} \Bigl[ \| \vec{y}^{(l)} - \vec{\phi}_{:,j}^{(l)} - \mat{\Phi}_{:,\sim j}^{(l)} \hat{\vec{x}}_{\sim j} \|_2^2 + \trace( \diag( \hat{\vec{x}}_{\sim j} \odot (\vec{1} - \hat{\vec{x}}_{\sim j}) ) \mat{\Phi}_{:,\sim j}^{(l)^\top} \mat{\Phi}_{:,\sim j}^{(l)} ) \Bigr] + \hat{\vec{x}}_{\sim j}^\top \vec{\eta}_{\sim j}^{(\text{c} \to \text{r})} + \eta_j^{(\text{c} \to \text{r})} \biggr] ,  \\
    \prob(x_j=0) &\propto \exp \biggl[ \sum_{l=1}^L -\frac{1}{2} \beta^{(l)} \Bigl[ \| \vec{y}^{(l)}   \hspace{2.4em}         - \mat{\Phi}_{:,\sim j}^{(l)} \hat{\vec{x}}_{\sim j} \|_2^2 + \trace( \diag( \hat{\vec{x}}_{\sim j} \odot (\vec{1} - \hat{\vec{x}}_{\sim j}) ) \mat{\Phi}_{:,\sim j}^{(l)^\top} \mat{\Phi}_{:,\sim j}^{(l)} ) \Bigr] + \hat{\vec{x}}_{\sim j}^\top \vec{\eta}_{\sim j}^{(\text{c} \to \text{r})}                      \biggr] .
\end{align}
\end{subequations}
Then, the mean of $x_j$ with respect to $q(x_j)$, denoted as $\hat{x}_j$, is given by
\begin{equation}
    \hat{x}_j = \frac{\prob(x_j=1)}{\prob(x_j=1) + \prob(x_j=0)} .
\end{equation}
After the belief of the rightmost cluster, denoted as $\frac{1}{Z} \exp(\vec{x}^\top \vec{\eta}^{(\text{r})})$, is computed, the parameters $\beta^{(l)}$ are updated using the \ac{EM} algorithm. The derivation for this update is similar to that presented in \cref{thm:expectation-maximization-instance}.
\begin{equation}  \label{eq:message-passing-derivation-last}
\begin{split}
    (\beta^{(l)})^{(\text{new})}
    &= \argmax_{\beta^{(l)}} \int p\Bigl( \vec{x} | \vec{y}^{(1)}, \dots, \vec{y}^{(L)}, (\beta^{(1)})^{(\text{old})}, \dots, (\beta^{(L)})^{(\text{old})} \Bigr) \ln p\Bigl( \vec{y}^{(l)}, \vec{x} | \beta^{(l)} \Bigr) \, d\vec{x}  \\
    &= \argmax_{\beta^{(l)}} \int p\Bigl( \vec{x} | \vec{y}^{(1)}, \dots, \vec{y}^{(L)}, (\beta^{(1)})^{(\text{old})}, \dots, (\beta^{(L)})^{(\text{old})} \Bigr) \ln p\Bigl( \vec{y}^{(l)} | \vec{x}, \beta^{(l)} \Bigr) \, d\vec{x}  \\
    &= \argmax_{\beta^{(l)}} \int \exp(\vec{x}^\top \vec{\eta}^{(\text{r})}) \ln \left( \biggl( \frac{\beta^{(l)}}{2\pi} \biggr)^{M/2} \exp \biggl[ -\frac{1}{2} \beta^{(l)} \| \vec{y}^{(l)} - \mat{\Phi}^{(l)} \vec{x} \|_2^2 \biggr] \right) d\vec{x}  \\
    &= \frac{M}{\| \vec{y}^{(l)} - \mat{\Phi}^{(l)} \hat{\vec{x}} \|_2^2 + \trace( \diag( \hat{\vec{x}} \odot (\vec{1} - \hat{\vec{x}}) ) \mat{\Phi}^{(l)^\top} \mat{\Phi}^{(l)} )} ,
\end{split}
\end{equation}
where $\hat{\vec{x}}$ denotes the expected value of $\frac{1}{Z} \exp(\vec{x}^\top \vec{\eta}^{(\text{r})})$.
The derivations from \eqref{eq:message-passing-derivation-first} through \eqref{eq:message-passing-derivation-last} are summarized in \cref{alg:approx-message-passing}.\\
\begin{algorithm}[H]
    \caption{Approximate message passing for binary vector recovery}
    \label{alg:approx-message-passing}
    \SetKwInOut{Input}{Input}
    \SetKwInOut{Init}{Initialize}
    \SetKwInOut{Output}{Output}
    \SetKwRepeat{Do}{do}{while}

    \Input{$\mat{\Phi}^{(1)}, \dots, \mat{\Phi}^{(L)}, \vec{y}^{(1)}, \dots, \vec{y}^{(L)}$}
    \Init{$\beta^{(l)} = \beta_0 \ \ \forall l=1,\dots,L$}
    \Do{$\hat{\vec{x}}$ not converged}{
        $\vec{\eta}^{(\text{c} \to \text{r})} = \ln \frac{\mathrm{B}(a+1,b)}{\mathrm{B}(a,b+1)} \cdot \vec{1}$  \\
        $\hat{\vec{x}} \gets \vec{1} / ( \vec{1} + \exp( -\vec{\eta}^{(\text{c} \to \text{r})} ) )$  \\
        \Do{$\hat{\vec{x}}$ not converged}{
            \ForEach{$j = 1, \dots, N$}{
                $\xi_{j0} \gets \sum_{l=1}^L -\frac{1}{2} \beta^{(l)} \| \vec{y}^{(l)}   \hspace{2.4em}         - \mat{\Phi}_{:,\sim j}^{(l)} \hat{\vec{x}}_{\sim j} \|_2^2                     $  \\
                $\xi_{j1} \gets \sum_{l=1}^L -\frac{1}{2} \beta^{(l)} \| \vec{y}^{(l)} - \vec{\phi}_{:,j}^{(l)} - \mat{\Phi}_{:,\sim j}^{(l)} \hat{\vec{x}}_{\sim j} \|_2^2 + \eta_j^{(\text{c} \to \text{r})}$  \\
                $\hat{x}_j \gets 1 / ( 1 + \exp(\xi_{j0} - \xi_{j1}))$  \\
            }
        }
        \ForEach{$l = 1, \dots, L$}{
            $\beta^{(l)} \gets  M / \bigl( \| \vec{y}^{(l)} - \mat{\Phi}^{(l)} \hat{\vec{x}} \|_2^2 + \trace( \diag( \hat{\vec{x}} \odot (\vec{1} - \hat{\vec{x}}) ) \mat{\Phi}^{(l)^\top} \mat{\Phi}^{(l)} ) \bigr)$  \\
        }
    }
    \Output{$\hat{\vec{x}}$}
\end{algorithm}
\noindent
In \cref{alg:convex-optimization}, the process begins by initializing all $\beta^{(l)}$ values to an initial guess $\beta_0$.
Furthermore, lines 3 through 10 computes to projection of $\exp(\vec{x}^\top \vec{\eta}^{(\text{c} \to \text{r})}) \prod_{l=1}^L p(\vec{y}^{(l)} | \vec{x})$ onto the exponential family defined in \eqref{eq:exponential-family}.

%% file: figures/bayesian-network.tex
\begin{tikzpicture}
    \tikzstyle{var} = [draw, circle, minimum size=8mm]
    \tikzstyle{dep} = [->]

    \node[var] (x1) at (0, 0) {$x_1$};
    \node[var] (x2) at (0,-1) {$x_2$};
    \node           at (0,-1.9) {$\vdots$};
    \node[var] (xN) at (0,-3) {$x_N$};

    \node[var] (y1) at (2,-0.5) {$y_1^{(l)}$};
    \node           at (2,-1.4) {$\vdots$};
    \node[var] (yM) at (2,-2.5) {$y_M^{(l)}$};

    \node[var] (beta) at (3.2,-1.5) {$\beta^{(l)}$};

    \draw[dep] (x1) -- (y1);   \draw[dep] (x1) -- (yM);
    \draw[dep] (x2) -- (y1);   \draw[dep] (x2) -- (yM);
    \draw[dep] (xN) -- (y1);   \draw[dep] (xN) -- (yM);
    \draw[dep] (beta) -- (y1); \draw[dep] (beta) -- (yM);

    \draw (1.3,0.2) rectangle (3.85,-3.2);
    \node[anchor=north east, inner sep=2pt] at (3.85,0.2) {\footnotesize $l = 1, \dots, L$};
\end{tikzpicture}

%% file: figures/cluster-graph.tex
\begin{tikzpicture}
    \tikzstyle{bg} = [draw, fill=white, fill opacity=0.6, rounded corners=5pt, minimum width=0.7cm]
    \tikzstyle{msg} = [midway,inner sep=2pt]

    \begin{scope}[shift={(-0.3,0)}]
        \node[bg, minimum height=2.6cm] (cluster1) at (0,0) {};
        \node at (0, 0.9) {$\pi_1$};
        \node at (0, 0.3) {$\pi_2$};
        \node at (0,-0.2) {$\vdots$};
        \node at (0,-0.9) {$\pi_N$};
    \end{scope}

    \begin{scope}[shift={(1.6,0)}]
        \node[bg, minimum width=1.4cm, minimum height=2.6cm] (cluster2) at (0.3,0) {};
        \node at (0, 0.9) {$\pi_1$};
        \node at (0, 0.3) {$\pi_2$};
        \node at (0,-0.2) {$\vdots$};
        \node at (0,-0.9) {$\pi_N$};
        \node at (0.6, 0.9) {$x_1$};
        \node at (0.6, 0.3) {$x_2$};
        \node at (0.6,-0.2) {$\vdots$};
        \node at (0.6,-0.9) {$x_N$};
    \end{scope}
    \draw (cluster1.east) -- (cluster2.west);
    \draw[-Latex,very thin] ($0.9*(cluster1.east)+0.1*(cluster2.west)+(0,0.1)$) -- ($0.1*(cluster1.east)+0.9*(cluster2.west)+(0,0.1)$) node[msg,above] {$m^{(\pi \rightarrow \text{c})}$};

    \begin{scope}[shift={(4.1,0)}]
        \node[bg, minimum width=1.4cm, minimum height=2.6cm] (cluster3) at (0.3,0) {};
        \node at (0, 0.9) {$x_1$};
        \node at (0, 0.3) {$x_2$};
        \node at (0,-0.2) {$\vdots$};
        \node at (0,-0.9) {$x_N$};
        \node at (0.6, 0.6) {$\vec{y}^{(1)}$};
        \node at (0.6, 0.1) {$\vdots$};
        \node at (0.6,-0.6) {$\vec{y}^{(L)}$};
    \end{scope}
    \draw (cluster2.east) -- (cluster3.west);
    \draw[-Latex,very thin] ($0.9*(cluster2.east)+0.1*(cluster3.west)+(0,0.1)$) -- ($0.1*(cluster2.east)+0.9*(cluster3.west)+(0,0.1)$) node[msg,above] {$m^{(\text{c} \rightarrow \text{r})}$};
    \draw[Latex-,very thin] ($0.9*(cluster2.east)+0.1*(cluster3.west)-(0,0.1)$) -- ($0.1*(cluster2.east)+0.9*(cluster3.west)-(0,0.1)$) node[msg,below] {$m^{(\text{c} \leftarrow  \text{r})}$};

\end{tikzpicture}

%% file: contents/section4.tex
\section{Validation of the algorithms}
\label{sec:algorithmic-experiment}

\begin{table}
    \caption{List of algorithms.}
    \label{tbl:algorithms}
    \centering
    \begin{tabular}{l|l|l}
        \thickhline
        & Description & Computational \\
        &             & complexity    \\
        \semithickhline
        \cref{alg:convex-optimization}    & {\small Convex optimization for binary vector recovery        } & $O(N^3)$   \\
        \cref{alg:mean-field-approx}      & {\small Mean field approximation for binary vector recovery   } & $O(NML)$   \\
        \cref{alg:approx-message-passing} & {\small Approximate message passing for binary vector recovery} & $O(N^2ML)$ \\
        Split Bregman \cite{RN93}         & {\small Convex optimization for sparse vector recovery        } & $O(N^2)$   \\
        \ac{SBL} \cite{RN211}             & {\small Bayesian inference for sparse vector recovery         } & $O(NM^2)$  \\
        \thickhline
    \end{tabular}
\end{table}

\cref{tbl:algorithms} lists the binary vector recovery algorithms developed in this article, alongside selected sparse vector recovery algorithms from the literature. The split Bregman method \cite{RN93} uses convex optimization to solve the $\ell_1$ minimization problem and is an instance of the \ac{ADMM} \cite{RN1311}. Among the various variants of the \ac{SBL} algorithm, the one included here, based on \cite{RN211}, does not impose any prior correlation between the entries. \cref{tbl:algorithms} also includes the per-iteration computational complexity of each algorithm.

Each algorithm is tested on its ability to recover binary vectors when the measurement matrices are Gaussian random. Gaussian random matrices are chosen because, as per the Johnson-Lindenstrauss lemma, they satisfy the \ac{RIP}. In each experiment with parameters $N$, $M$, $L$, and $s$:
\begin{enumerate}[nosep, leftmargin=*]
    \item Generate matrices $\mat{\Phi}^{(l)} \in \mathbb{R}^{M \times N}$ for $l=1,\dots,L$ with entries drawn from $\gaussdist(0,1)$.
    \item Create a random binary vector $\vec{x}_\text{true} \in \{0,1\}^N$ containing exactly $s$ ones.
    \item Compute $\vec{y}^{(l)} = \mat{\Phi}^{(l)} \vec{x}_\text{true} + \vec{n}^{(l)}$ for all $l=1,\dots,L$, where $\vec{n}^{(l)}$ is an \ac{IID} Gaussian noise with a specified \ac{SNR}.
\end{enumerate}
The objective is to recover $\vec{x}_\text{true}$ using each algorithm. Performance is evaluated using the \ac{IoU} metric, defined as:
\begin{equation*}  \label{eq:IoU}
    \text{\ac{IoU}} = \frac{\left| \support{\vec{x}} \cap \support{\vec{x}_\text{true}} \right|}{\left| \support{\vec{x}} \cup \support{\vec{x}_\text{true}} \right|} ,
\end{equation*}
where $\support{\vec{x}}$ denotes the set of indices of nonzero entries in $\vec{x}$, and $\left|\cdot\right|$ denotes the cardinality of a set.
The \ac{IoU} metric is well-suited for comparing binary vectors, where a value of one indicates exact recovery, and a value of zero indicates no similarity.

The first set of experiments examine the simple case where $L=1$ and no noise is added to the measurement vector. Each combination of parameter values is tested over 1000 trials, with the average \ac{IoU} recorded. Results are presented in \cref{fig:IoU-single-measure}.

\cref{fig:IoU-single-measure-vary-M} shows the results when $s$ is fixed and $M$ varies. The \ac{IoU} increases as the number of measurements grows. \cref{alg:convex-optimization} consistently achieves the highest \ac{IoU}. At very low $M$, \cref{alg:mean-field-approx} outperforms both the split Bregman method and \ac{SBL}. As $M$ increases slightly, the split Bregman method and \ac{SBL} surpass \cref{alg:mean-field-approx} in \ac{IoU}, with all algorithms eventually reaching an \ac{IoU} of one. In contrast, \cref{alg:approx-message-passing} shows the lowest performance across all $M$ values.

\cref{fig:IoU-single-measure-vary-k} shows the results when $M$ is fixed and $s$ varies. As $s$ increases and making the vector less sparse, the \ac{IoU} decreases. Both \cref{alg:convex-optimization} and \cref{alg:mean-field-approx} outperform the sparse vector recovery algorithms. Although \cref{alg:approx-message-passing} has a lower \ac{IoU} at small $s$, it eventually surpasses the sparse vector recovery algorithms as $s$ increases. This suggests that, when binary constraints are applicable, binary vector recovery algorithms should be preferred over sparse recovery methods for improved performance.

\begin{figure}
    \centering
    \subcaptionbox{\label{fig:IoU-single-measure-vary-M}}{
        \input{figures/IoU-varying-M-single-measure.tex}
    } \hfill
    \subcaptionbox{\label{fig:IoU-single-measure-vary-k}}{
        \input{figures/IoU-varying-k-single-measure.tex}
    }
    \caption{Algorithm performance for $L=1$ and $N=500$. (a) $s=25$. (b) $M=250$.}
    \label{fig:IoU-single-measure}
\end{figure}

The second set of experiments examine the case where there are multiple measurement vectors and unknown noise variance. Specifically, the tests are conducted with $L=3$, and noise levels corresponding to SNR values of 29~dB, 30~dB, and 31~dB are added to the measurement vectors. As before, 1000 trials are conducted for each combination of parameter values, and the average \ac{IoU} is recorded. The results are shown in \cref{fig:IoU-multiple-measure}. Notably, the split Bregman method and \ac{SBL} are excluded from these experiments because they do not support multiple measurement vectors.

\cref{fig:IoU-multiple-measure-vary-M} illustrates the results for varying $M$ with a fixed $s$, the \ac{IoU} increases as the number of measurements grows. \cref{fig:IoU-multiple-measure-vary-k} shows the results for varying $s$ with a fixed $M$, the \ac{IoU} decreases as the vector become less sparse.
Among the three algorithms tested, \cref{alg:approx-message-passing} exhibits the lowest performance. This is likely due to fundamental differences in the approaches. Both \cref{alg:convex-optimization} and \cref{alg:mean-field-approx} incorporate optimization: \cref{alg:convex-optimization} optimizes over the convex hull of the binary constraint, while \cref{alg:mean-field-approx} minimizes the \ac{KL}-divergence over a class of distributions. In contrast, \cref{alg:approx-message-passing} employs approximate message passing, which do not explicitly optimize any objective function but instead operates by propagating messages.
Notably, \cref{alg:convex-optimization} consistently outperforms \cref{alg:mean-field-approx}. This is likely because \cref{alg:mean-field-approx} uses mean field approximation that neglects the correlation between the posterior probabilities of entries, resulting in biased solutions. In contrast, \cref{alg:convex-optimization} benefits from the recovery guarantees provided in \cref{thm:tight-convex}, assuming the vector is sufficiently sparse. However, \cref{alg:convex-optimization} has a much higher computational complexity compared to \cref{alg:mean-field-approx}.

\begin{figure}
    \centering
    \subcaptionbox{\label{fig:IoU-multiple-measure-vary-M}}{
        \input{figures/IoU-varying-M-multiple-measure.tex}
    } \hfill
    \subcaptionbox{\label{fig:IoU-multiple-measure-vary-k}}{
        \input{figures/IoU-varying-k-multiple-measure.tex}
    }
    \caption{Algorithm performance for $L=3$ and $N=500$. (a) $s=25$. (b) $M=83$.}
    \label{fig:IoU-multiple-measure}
\end{figure}

%% file: figures/IoU-varying-M-single-measure.tex
\begin{tikzpicture}

\begin{axis}[%
width=0.49\columnwidth,
height=0.35\columnwidth,
xmin=0.1,
xmax=0.9,
xlabel={$M/N$},
ymin=0,
ymax=1,
ylabel={\hyperref[eq:IoU]{IoU}},
xmajorgrids,
ymajorgrids,
tick label style={
  /pgf/number format/fixed
},
legend style={
  anchor=south east,
  at={(0.99,0.01)},
  fill opacity=0.7,
  text opacity=1
}
]

\addplot+ [mark=x]
  table[row sep=crcr]{%
0.1	0.205572961943701\\
0.12	0.284339275180032\\
0.14	0.415427031008568\\
0.16	0.705262889071891\\
0.18	0.925913055054119\\
0.2	0.992211496217703\\
0.22	1\\
0.24	1\\
0.26	1\\
0.28	1\\
0.3	1\\
0.32	1\\
0.34	1\\
0.36	1\\
0.38	1\\
0.4	1\\
0.5	1\\
0.6	1\\
0.7	1\\
0.8	1\\
0.9	1\\
};
\addlegendentry{\cref{alg:convex-optimization}}

\addplot+ [mark=x]
table[row sep=crcr]{%
0.1	0.182594828916021\\
0.12	0.232865980764085\\
0.14	0.293513743220171\\
0.16	0.390363780773841\\
0.18	0.556293754459154\\
0.2	0.730189580662104\\
0.22	0.858194614072164\\
0.24	0.935107849674678\\
0.26	0.975885678597335\\
0.28	0.986530616627866\\
0.3	0.992501615037545\\
0.32	0.997641469458274\\
0.34	1\\
0.36	1\\
0.38	1\\
0.4	1\\
0.5	1\\
0.6	1\\
0.7	1\\
0.8	1\\
0.9	1\\
};
\addlegendentry{\cref{alg:mean-field-approx}}

\addplot+ [mark=x]
  table[row sep=crcr]{%
0.1	0.0616754887424489\\
0.12	0.0658761477866182\\
0.14	0.0696620648344616\\
0.16	0.0730777214719684\\
0.18	0.080968831310227\\
0.2	0.0869703964238566\\
0.22	0.100117254971148\\
0.24	0.131841336315998\\
0.26	0.160110220326674\\
0.28	0.220014961562648\\
0.3	0.330434230821849\\
0.32	0.457399177371662\\
0.34	0.5827658778974\\
0.36	0.72409983649313\\
0.38	0.806253979317499\\
0.4	0.879920696009061\\
0.5	0.995145627545325\\
0.6	1\\
0.7	1\\
0.8	1\\
0.9	1\\
};
\addlegendentry{\cref{alg:approx-message-passing}}

\addplot+ [mark=x]
  table[row sep=crcr]{%
0.1	0.138492777774972\\
0.12	0.178978163620388\\
0.14	0.223621884305264\\
0.16	0.30728943367978\\
0.18	0.441254414117481\\
0.2	0.703642654670341\\
0.22	0.923440857648099\\
0.24	0.9912150997151\\
0.26	0.99964\\
0.28	0.99972\\
0.3	1\\
0.32	1\\
0.34	1\\
0.36	1\\
0.38	1\\
0.4	1\\
0.5	1\\
0.6	1\\
0.7	1\\
0.8	1\\
0.9	1\\
};
\addlegendentry{Split Bregman}

\addplot+ [mark=x]
  table[row sep=crcr]{%
0.1	0.126157308732372\\
0.12	0.17376112477887\\
0.14	0.231475078709347\\
0.16	0.31626089709293\\
0.18	0.536000144619171\\
0.2	0.821552314957711\\
0.22	0.97419530830394\\
0.24	0.998826210826211\\
0.26	1\\
0.28	1\\
0.3	1\\
0.32	1\\
0.34	1\\
0.36	1\\
0.38	1\\
0.4	1\\
0.5	1\\
0.6	1\\
0.7	1\\
0.8	1\\
0.9	1\\
};
\addlegendentry{\ac{SBL}}

\end{axis}

\end{tikzpicture}%

%% file: figures/IoU-varying-k-single-measure.tex
\begin{tikzpicture}

\begin{axis}[%
width=0.49\columnwidth,
height=0.35\columnwidth,
xmin=0.0,
xmax=0.4,
xlabel={$s/N$},
ymin=0,
ymax=1,
ylabel={\hyperref[eq:IoU]{IoU}},
xmajorgrids,
ymajorgrids,
tick label style={
  /pgf/number format/fixed
},
legend style={
  anchor=south west,
  at={(0.01,0.01)},
  fill opacity=0.7,
  text opacity=1
}
]

\addplot+ [mark=x]
  table[row sep=crcr]{%
0.002	1\\
0.022	1\\
0.042	1\\
0.062	1\\
0.082	1\\
0.102	1\\
0.122	1\\
0.142	1\\
0.162	1\\
0.182	1\\
0.202	1\\
0.222	1\\
0.242	1\\
0.262	1\\
0.282	0.9997668997669\\
0.302	0.998792724963969\\
0.322	0.99419051634343\\
0.342	0.985236040836652\\
0.362	0.969766199149501\\
0.382	0.956463392128645\\
0.402	0.926277024454523\\
};
\addlegendentry{\cref{alg:convex-optimization}}

\addplot+ [mark=x]
  table[row sep=crcr]{%
0.002	1\\
0.022	1\\
0.042	1\\
0.062	1\\
0.082	1\\
0.102	1\\
0.122	0.999571856714714\\
0.142	0.998664714626818\\
0.162	0.994829939564512\\
0.182	0.981668051969071\\
0.202	0.96077389520089\\
0.222	0.931317275056367\\
0.242	0.892518076816223\\
0.262	0.829942526251149\\
0.282	0.779426900110177\\
0.302	0.747989147651833\\
0.322	0.70856918644156\\
0.342	0.685674196142733\\
0.362	0.664122492207269\\
0.382	0.654360935942796\\
0.402	0.642077527199933\\
};
\addlegendentry{\cref{alg:mean-field-approx}}

\addplot+ [mark=x]
  table[row sep=crcr]{%
0.002	0.999030381383323\\
0.022	0.998178632961242\\
0.042	0.990871218366034\\
0.062	0.986024404122795\\
0.082	0.968186420935023\\
0.102	0.943762159168885\\
0.122	0.924031131091067\\
0.142	0.855293154779878\\
0.162	0.785948710253338\\
0.182	0.723950215799792\\
0.202	0.6544446408804\\
0.222	0.598831425491042\\
0.242	0.55815674698141\\
0.262	0.551760575095592\\
0.282	0.524071190608154\\
0.302	0.51456647569499\\
0.322	0.511200643959119\\
0.342	0.510552832911586\\
0.362	0.497102598926513\\
0.382	0.49504568996778\\
0.402	0.484078362970572\\
};
\addlegendentry{\cref{alg:approx-message-passing}}

\addplot+ [mark=x]
  table[row sep=crcr]{%
0.002	1\\
0.022	1\\
0.042	1\\
0.062	1\\
0.082	1\\
0.102	1\\
0.122	1\\
0.142	1\\
0.162	0.996635926127014\\
0.182	0.92603394808854\\
0.202	0.734494272210229\\
0.222	0.573088474307083\\
0.242	0.509637049491155\\
0.262	0.472464568721457\\
0.282	0.444985359895592\\
0.302	0.425815383267195\\
0.322	0.410136832640816\\
0.342	0.397864597842434\\
0.362	0.387174460953712\\
0.382	0.38046687371664\\
0.402	0.373578431073521\\
};
\addlegendentry{Split Bregman}

\addplot+ [mark=x]
  table[row sep=crcr]{%
0.002	1\\
0.022	1\\
0.042	1\\
0.062	1\\
0.082	1\\
0.102	1\\
0.122	1\\
0.142	1\\
0.162	0.998537098359936\\
0.182	0.962229572059696\\
0.202	0.751233895705961\\
0.222	0.531660870561393\\
0.242	0.467801251639504\\
0.262	0.433588127697741\\
0.282	0.40833815984365\\
0.302	0.38964053344654\\
0.322	0.376062473434354\\
0.342	0.362731639202968\\
0.362	0.352708566574119\\
0.382	0.342654879026933\\
0.402	0.335762675272416\\
};
\addlegendentry{\ac{SBL}}

\end{axis}

\end{tikzpicture}%

%% file: figures/IoU-varying-M-multiple-measure.tex
\begin{tikzpicture}

\begin{axis}[%
width=0.49\columnwidth,
height=0.35\columnwidth,
xmin=0.034,
xmax=0.3,
xlabel={$M/N$},
ymin=0,
ymax=1,
ylabel={\hyperref[eq:IoU]{IoU}},
xmajorgrids,
ymajorgrids,
tick label style={
  /pgf/number format/fixed
},
legend style={
  anchor=south east,
  at={(0.99,0.01)},
  fill opacity=0.7,
  text opacity=1
}
]

\addplot+ [mark=x]
table[row sep=crcr]{%
0.034	0.227016840757282\\
0.04	0.290173218448218\\
0.046	0.389086513486513\\
0.054	0.641382421744922\\
0.06	0.763620891608392\\
0.066	0.871894605394605\\
0.074	0.972654545454545\\
0.08	0.988527272727273\\
0.086	0.993354545454545\\
0.094	0.9993\\
0.1	  1\\
0.106	0.9998\\
0.114	1\\
0.12	1\\
0.126	0.9999\\
0.134	1\\
0.166	1\\
0.2 	1\\
0.234	1\\
0.266	1\\
0.3  	1\\
};
\addlegendentry{\cref{alg:convex-optimization}}

\addplot+ [mark=x]
  table[row sep=crcr]{%
0.034	0.182568759205403\\
0.04	0.230424108003783\\
0.046	0.282681269347825\\
0.054	0.362541992156796\\
0.06	0.457172804951066\\
0.066	0.588198512339428\\
0.074	0.744784102302565\\
0.08	0.830834887901835\\
0.086	0.898592226508983\\
0.094	0.944143519676009\\
0.1	  0.971332457886328\\
0.106	0.984186204777018\\
0.114	0.988134280482951\\
0.12	0.993961383862274\\
0.126	0.997170023404779\\
0.134	0.998708360197722\\
0.166	0.999578947368421\\
0.2 	1\\
0.234	1\\
0.266	1\\
0.3  	1\\
};
\addlegendentry{\cref{alg:mean-field-approx}}

\addplot+ [mark=x]
  table[row sep=crcr]{%
0.034 0.0575690340711466\\
0.04  0.0605930303911313\\
0.046 0.0629941691519551\\
0.054 0.0661214371382199\\
0.06  0.0669388154184892\\
0.066 0.0725477911153037\\
0.074 0.0808696799480666\\
0.08  0.0842100216061126\\
0.086 0.0917951684213043\\
0.094 0.107469675687401\\
0.1	  0.120484775962916\\
0.106 0.149263088727582\\
0.114 0.222760263245951\\
0.12  0.270609217508441\\
0.126 0.282134752319574\\
0.134 0.430920668533086\\
0.166 0.759666638615489\\
0.2   0.952672409746974\\
0.234 0.925198601548825\\
0.266 0.863491914189405\\
0.3   0.903106392401794\\
};
\addlegendentry{\cref{alg:approx-message-passing}}

\end{axis}

\end{tikzpicture}%

%% file: figures/IoU-varying-k-multiple-measure.tex
\begin{tikzpicture}

\begin{axis}[%
width=0.49\columnwidth,
height=0.35\columnwidth,
xmin=0.0,
xmax=0.4,
xlabel={$s/N$},
ymin=0,
ymax=1,
ylabel={\hyperref[eq:IoU]{IoU}},
xmajorgrids,
ymajorgrids,
tick label style={
  /pgf/number format/fixed
},
legend style={
  anchor=south west,
  at={(0.01,0.01)},
  fill opacity=0.7,
  text opacity=1
}
]

\addplot+ [mark=x]
table[row sep=crcr]{%
0.002	1\\
0.022	1\\
0.042	1\\
0.062	1\\
0.082	0.999412352353529\\
0.102	0.997478711764426\\
0.122	0.988654422500576\\
0.142	0.970765163646142\\
0.162	0.932131576707589\\
0.182	0.883661426509325\\
0.202	0.821456717151912\\
0.222	0.773548628835479\\
0.242	0.711957237161148\\
0.262	0.673594736427562\\
0.282	0.632897189602857\\
0.302	0.609481669476311\\
0.322	0.584555863711491\\
0.342	0.560285233576178\\
0.362	0.551532622163253\\
0.382	0.539576367998435\\
0.402	0.529415895719613\\
};
\addlegendentry{\cref{alg:convex-optimization}}

\addplot+ [mark=x]
  table[row sep=crcr]{%
0.002	1\\
0.022	1\\
0.042	1\\
0.062	0.999631947000368\\
0.082	0.994192208491199\\
0.102	0.948825266245989\\
0.122	0.843537562267124\\
0.142	0.735768262553648\\
0.162	0.66273899646125\\
0.182	0.61711549227838\\
0.202	0.598774828553239\\
0.222	0.585308596247557\\
0.242	0.577224559532369\\
0.262	0.567364744739791\\
0.282	0.566479853629206\\
0.302	0.564338962962388\\
0.322	0.562396415555376\\
0.342	0.562376095190142\\
0.362	0.561783265755956\\
0.382	0.561595796236551\\
0.402	0.561436097113004\\
};
\addlegendentry{\cref{alg:mean-field-approx}}

\addplot+ [mark=x]
  table[row sep=crcr]{%
0.002	0.918410777187615\\
0.022	0.85527908198254\\
0.042	0.756571316948114\\
0.062	0.730285061924257\\
0.082	0.617351457510212\\
0.102	0.576344951389011\\
0.122	0.541095569063621\\
0.142	0.484383235639026\\
0.162	0.467431708814725\\
0.182	0.4538823331009\\
0.202	0.4353424414341\\
0.222	0.434788647865595\\
0.242	0.422030914131262\\
0.262	0.417057769756764\\
0.282	0.411429348090782\\
0.302	0.408022996863993\\
0.322	0.401443086230148\\
0.342	0.395093484647475\\
0.362	0.382542979572083\\
0.382	0.381766419478042\\
0.402	0.371681543276872\\
};
\addlegendentry{\cref{alg:approx-message-passing}}

\end{axis}

\end{tikzpicture}%

%% file: contents/section5.tex
\section{Application to model based defect imaging}
\label{sec:application}

A common sensing method for detecting structural defects is eddy current sensing. This section demonstrates how binary vector recovery can be applied to these imaging problems.

\subsection{Formulating the defect imaging problem as a binary vector recovery problem}
An eddy current system, where the fields oscillate at an angular frequency $\omega$, is described by the following Maxwell's equations:
\begin{subequations}  \label{eq:electromagnetic-system1}
\begin{align}
    \nabla \times \vec{H} &= \sigma \vec{E} + \vec{J}_s ,  \\
    \nabla \times \vec{E} &= -j \omega \mu \vec{H} ,
\end{align}
\end{subequations}
where $\vec{E}$ is the electric field, $\vec{H}$ the magnetic field intensity, $\vec{J}_s$ the source electric current density, $\sigma$ is the electrical conductivity field, and $\mu$ is the magnetic permeability field.
The \ac{MFD} component measured by a magnetic sensor is given by
\begin{equation}  \label{eq:MFD-measurement}
    B = \int_V \mu_0 \vec{H} \cdot \hat{\vec{n}}_\text{sensor} \, \delta(\vec{r} - \vec{r}_\text{sensor}) \, dV,
\end{equation}
where $\mu_0$ is the vacuum magnetic permeability, $\hat{\vec{n}}_\text{sensor}$ is the sensor's sensing axis, $\vec{r}_\text{sensor}$ is the sensor's position, and $\delta(\cdot)$ is the Dirac delta function.

\begin{theorem}  \label{thm:eddy-current-linearize}
    The perturbation in the \ac{MFD} component measurement, denoted by $\delta B$, is related to the perturbations in the electrical conductivity field, $\delta \sigma$, and the magnetic permeability field, $\delta \mu$, as follows:
    \begin{equation}  \label{eq:measurement-perturb}
        \delta B = \int_V \bigl( (\vec{E} \cdot \acute{\vec{E}}) \, \delta\sigma + (-j \omega \vec{H} \cdot \acute{\vec{H}}) \, \delta\mu \bigr) \, dV,
    \end{equation}
    where $(\vec{E}, \vec{H})$ represents the electromagnetic field governed by \eqref{eq:electromagnetic-system1}, and $(\acute{\vec{E}}, \acute{\vec{H}})$ represents the electromagnetic field governed by
    \begin{subequations}  \label{eq:electromagnetic-system2}
    \begin{align}
        \nabla \times \acute{\vec{H}} &= \sigma \acute{\vec{E}},  \\
        \nabla \times \acute{\vec{E}} &= -j \omega \mu \acute{\vec{H}} + \mu_0 \hat{\vec{n}}_\mathrm{sensor} \, \delta(\vec{r} - \vec{r}_\mathrm{sensor}).
    \end{align}
    \end{subequations}
\end{theorem}
\begin{proof}
    The proof is provided in \ref{app:linearization-proof}.
\end{proof}

To apply \cref{thm:eddy-current-linearize}, it is necessary to solve the electromagnetic fields governed by \eqref{eq:electromagnetic-system1} and \eqref{eq:electromagnetic-system2}. In several configurations, semi-analytical solutions are available. For example, for a metal plate with an excitation source above it, solutions can be found in \cite{RN1414}. For a metal pipe with an outer excitation source, solutions are provided in \cite{RN1284}. Additionally, for a metal pipe with an inner excitation source, semi-analytical solutions are derived in \ref{app:electromagnetic-field-pipe-semianalytical}.

Suppose the inspected metal structure is discretized into $N$ voxels, \eqref{eq:measurement-perturb} can be written as
\begin{equation}  \label{eq:measurement-perturb-discretize}
    \delta B = \sum_{j=1}^N \biggl( \Bigl( \int_{V_{\text{voxel}_j}} \vec{E} \cdot \acute{\vec{E}} \, dV \Bigr) \, \delta \sigma_j + \Bigl( \int_{V_{\text{voxel}_j}} -j \omega \vec{H} \cdot \acute{\vec{H}} \, dV \Bigr) \, \delta \mu_j \biggr) ,
\end{equation}
where $\delta \sigma_j$ and $\delta \mu_j$ represent voxel $j$'s change in electrical conductivity and magnetic permeability, respectively.
Consider the case where there are $M$ sensors, \eqref{eq:measurement-perturb-discretize} can be written in vector-matrix form:
\begin{equation}  \label{eq:relation-linearize}
    \Delta\vec{B} = \mat{S}_\sigma \Delta\vec{\sigma} + \mat{S}_\mu \Delta\vec{\mu} ,
\end{equation}
where $\Delta\vec{B}$ is a vector representing the change in \ac{MFD} component measurements from the $M$ sensors.

A structural defect can be viewed as a region within the structure where the electrical conductivity and magnetic permeability deviate from their nominal values to those of air. Thus, \eqref{eq:relation-linearize} is normalized as follows:
\begin{equation}
    \Delta\vec{B} = \mat{S}_\sigma \diag(-\vec{\sigma}) \diag(-\vec{\sigma})^{-1} \Delta\vec{\sigma} + \mat{S}_\mu \diag(\mu_0 \vec{1} - \vec{\mu}) \diag(\mu_0 \vec{1} - \vec{\mu})^{-1} \Delta\vec{\mu} ,
\end{equation}
where $\vec{\sigma}$ and $\vec{\mu}$ denote the nominal electrical conductivity and magnetic permeability of the voxels, respectively.
Both $\bigl( \diag(-\vec{\sigma})^{-1} \Delta\vec{\sigma} \bigr)$ and $\bigl( \diag(\mu_0 \vec{1} - \vec{\mu})^{-1} \Delta\vec{\mu} \bigr)$ form binary vectors that indicate voxel vacancies, and are denoted here as $\vec{x}$. This yields
\begin{equation}  \label{eq:relation-linearize-normalize}
    \underbrace{\Delta\vec{B}}_{\vec{y}}
    = \bigl( \underbrace{\mat{S}_\sigma \diag(-\vec{\sigma}) + \mat{S}_\mu \diag(\mu_0 \vec{1} - \vec{\mu})}_{\mat{\Phi}} \bigr) \vec{x} .
\end{equation}
In \eqref{eq:relation-linearize-normalize}, $\vec{y}$ and $\mat{\Phi}$ have complex entries, which can be separated into their real and imaginary parts. When considering multiple frequencies and the noise associated with the measurements, the following relation applies:
\begin{equation}  \label{eq:relation-linearize-normalize-seperate}
    \vec{y}^{(l)} = \mat{\Phi}^{(l)} \vec{x} + \vec{n}^{(l)} \quad \forall l = 1,\dots,L ,
\end{equation}
where $\vec{y}^{(l)}$ and $\mat{\Phi}^{(l)}$ represents the real and imaginary parts of $\vec{y}$ and $\mat{\Phi}$ across all frequencies, respectively. Hence, $L$ equals twice the number of frequencies.
By utilizing binary vector recovery algorithms to solve for a binary vector $\vec{x}$ that satisfies \eqref{eq:relation-linearize-normalize-seperate}, a computed tomography image of the defects is obtained.

\subsection{Imaging defects in metal plates}
The first example concerns the imaging of defects in a metal plate. Traditional methods typically involve scanning an eddy current probe across the plate surface, recording the response at each point to form an image, and then using image processing or neural network techniques to identify any defects. However, the resolution of the recorded image is constrained by the Nyquist sampling theorem. By using a model-based approach along with binary vector recovery algorithms, it is possible to significantly reduce the number of samples needed.

\cref{fig:plate-sensing-schematic} shows the example setup for imaging defects in an aluminum plate using an 8$\times$8 array of magnetic sensors spaced 4~mm apart. These sensors are positioned below the central region of a double-D-shaped coil with their sensing axis aligned  parallel to the current direction. This orientation minimizes the \ac{MFD} component they sense when the plate has no defects. When a defect is present, however, the altered distribution of eddy currents causes nearby sensors to sense an \ac{MFD} component. This arrangement enhances the system's sensitivity to defects. The imaging region covers a 30$\times$30 mm$^2$ area below the magnetic sensors and spans the full 2~mm thickness of the aluminum plate. This region is discretized into 60$\times$60$\times$4 voxels, as illustrated in \cref{fig:plate-sensing-schematic-zoomin}. Since the voxels span four layers, excitation frequencies of 1000, 2000, 4000, and 16000~Hz are selected by approximately matching the eddy current penetration depth with the depth of the layers. The linearized sensitivities for all voxels across all frequencies are precomputed as described in \cref{thm:eddy-current-linearize} and \eqref{eq:measurement-perturb-discretize}.

\begin{figure}
    \centering
    \subcaptionbox{\label{fig:plate-sensing-schematic}}{
        \begin{tikzpicture}
            \tikzstyle{indicator} = [color=yellow, -{Circle[length=3pt]}, thick]
            \tikzstyle{label} = [anchor=west, rounded corners, fill=yellow, fill opacity=0.4, text=black, text opacity=1, align=left]

            \node[inner sep=0] at (0,0) {\includegraphics[width=11cm]{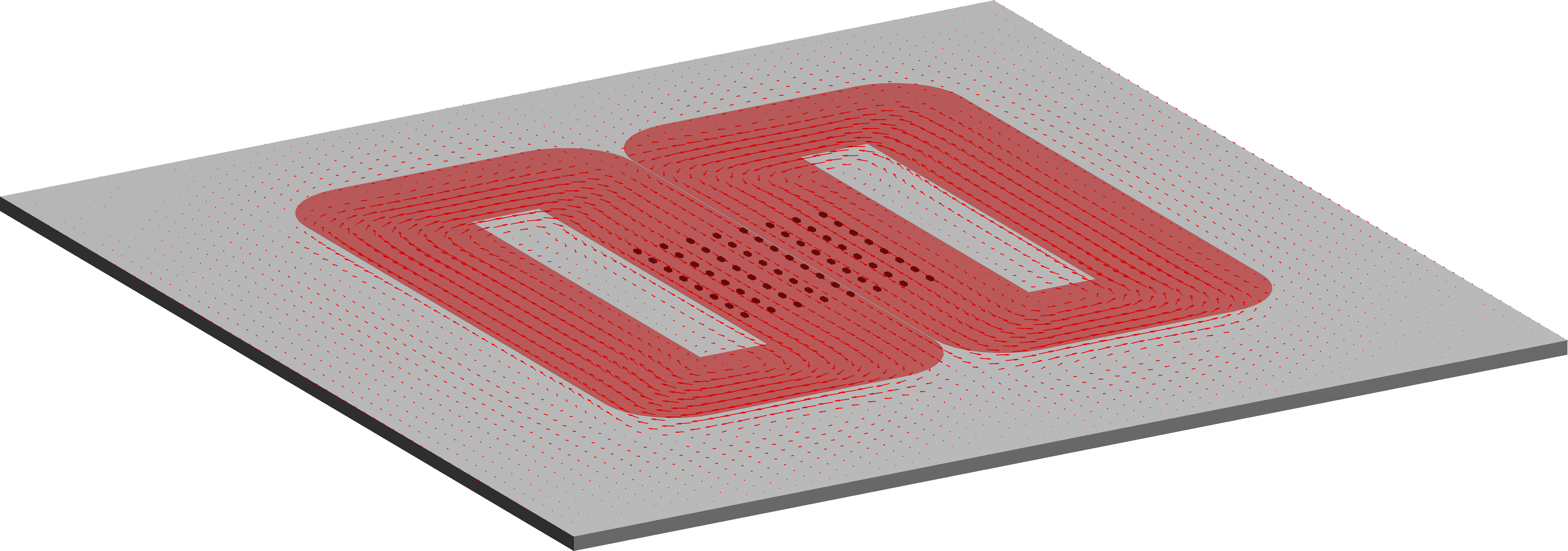}};

            \node[label] (label1) at (-4, 1.2) {Coil};
            \node[label] (label2) at (-4, 0.0) {Magnetic\\sensors};
            \node[label] (label3) at (-4,-1.2) {Metal plate};

            \draw[indicator] (label1.east) -- (-2.0, 0.7);
            \draw[indicator] (label2.east) -- (-0.5, 0.1);
            \draw[indicator] (label3.east) -- (-0.6,-1.7);
        \end{tikzpicture}
    } \hfill
    \subcaptionbox{\label{fig:plate-sensing-schematic-zoomin}}{
        \includegraphics[width=0.34\columnwidth]{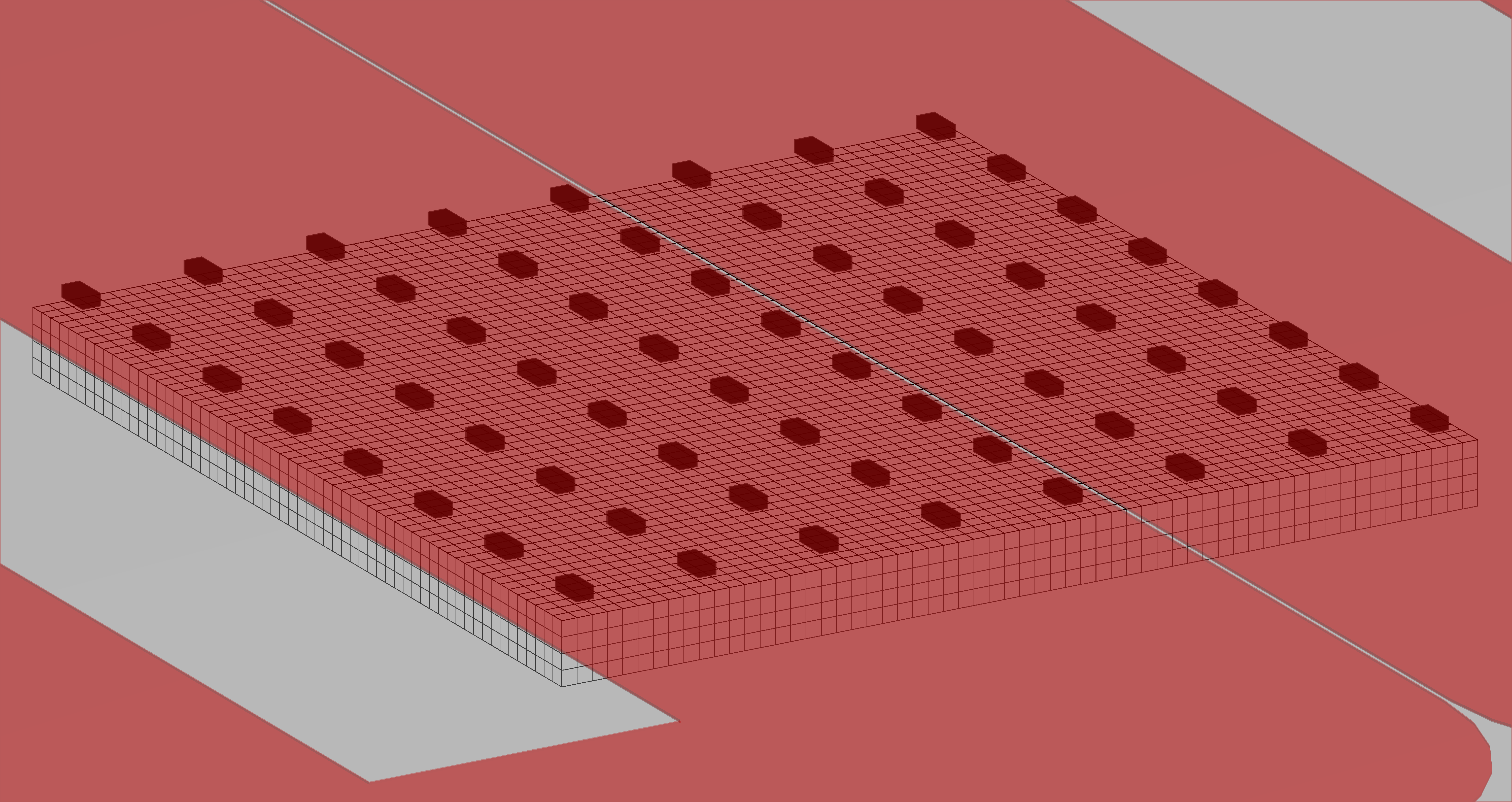}
    } \\
    \subcaptionbox{\label{fig:plate-sensing-simulation}}{
        \adjincludegraphics[width=11cm, trim={0 {0.23\height} 0 {0.23\height}}, clip]{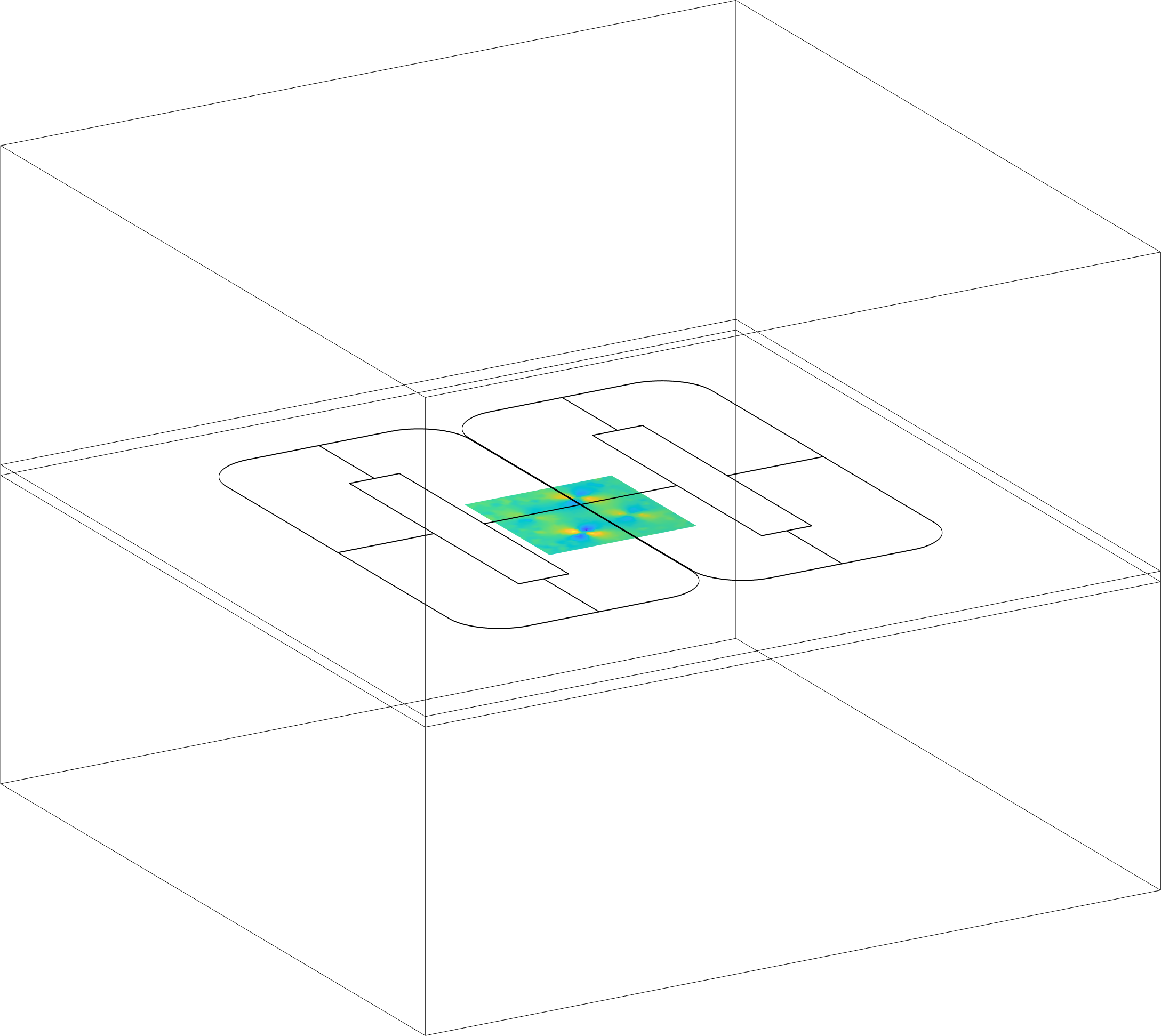}
    } \hfill
    \subcaptionbox{\label{fig:plate-sensing-MFD-diff}}{
        \begin{tikzpicture}
            \node[inner sep=0] at (0,0) {\includegraphics[height=4cm]{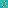}};
            \node[inner sep=0,draw,thick] at (2.3,0) {\includegraphics[height=4cm, width=2mm]{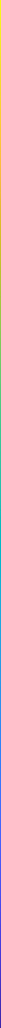}};
            \node[inner sep=0,anchor=east] at (2.9, 2.3) {\footnotesize ($\mu$T)};
            \node[inner sep=0,anchor=east] at (2.9, 2) {\footnotesize 2.5};
            \node[inner sep=0,anchor=east] at (2.9, 0) {\footnotesize 0\hspace*{2mm}};
            \node[inner sep=0,anchor=east] at (2.9,-2) {\footnotesize -2.5};
        \end{tikzpicture}
    }
    \caption{Eddy current sensing system for imaging defects in metal plates.
             (a) Full view with red arrows indicating the eddy currents.
             (b) Close-up view showing the voxels to be reconstructed.
             (c) Finite element simulation setup and simulated \ac{MFD} component difference.
             (d) \ac{MFD} component difference sampled at the 64 sensor positions.}
\end{figure}
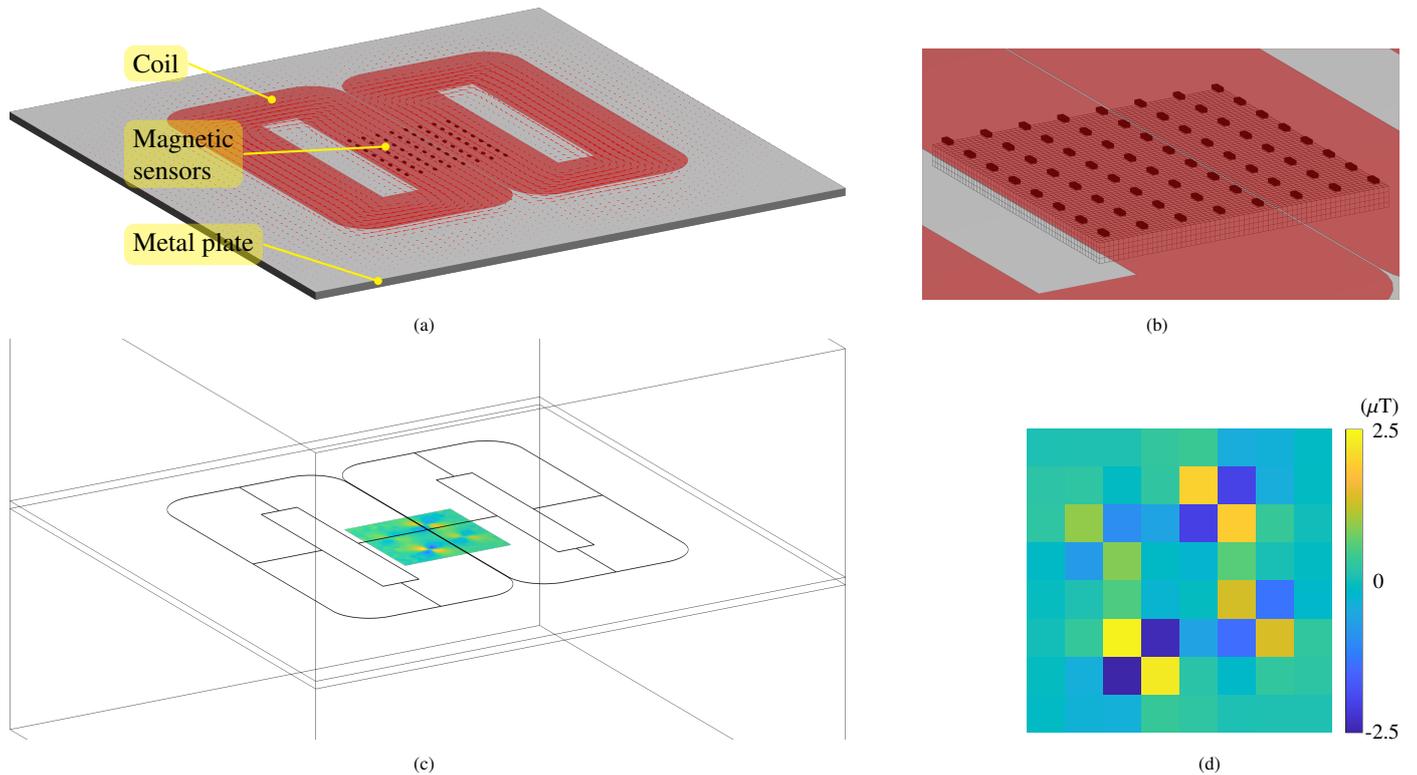

In this study, the \ac{MFD} component measured by the magnetic sensors is simulated using commercial finite element software. The simulation setup is illustrated in \cref{fig:plate-sensing-simulation}. The simulation domain also includes the air regions above and below the metal plate. The simulated \ac{MFD} component along the sensor's sensing axis is plotted in \cref{fig:plate-sensing-simulation}. In traditional dense sampling methods, this color plot would be the end result, with image processing techniques applied to identify any defects. In this example, however, the response is sampled more sparsely, at only 64 sensor positions, as shown in \cref{fig:plate-sensing-MFD-diff}. From this spatially sparse sampling, it is possible to identify four defects, but accurately determining their exact locations, sizes, and depths is more challenging.

To clarify, \cref{fig:plate-sensing-MFD-diff} shows only a single $\vec{y}^{(\ell)}$ vector, which has a length of 64, and there are seven other $\vec{y}^{(\ell)}$ vectors that are not plotted. In total, there are eight $\vec{y}^{(\ell)}$ measurement vectors, corresponding to four frequencies, each with real and imaginary components. Using these eight measurement vectors and applying the binary vector recovery algorithms, the vacancies of the 14400 voxels can be imaged. The results are presented in \cref{tbl:plate-result}.
\begin{table}
    \caption{Metal plate imaging results.}
    \label{tbl:plate-result}

    \setlength{\tabcolsep}{4pt}
    \def\w{24mm}

    \centering
    \begin{tabular}{|cc|c|c|c|c|c|c|}
        \thickhline
        && \multicolumn{3}{c|}{} & \multicolumn{3}{c|}{} \\[-9pt]
        \multicolumn{2}{|c|}{\multirow{-5.6}{*}{\rotatebox{90}{Ground truth}}} &
        \multicolumn{3}{c|}{
            \begin{tikzpicture}
                \node[inner sep=0, anchor=west] at (-0.58*\w,0) {\includeinkscape[width=\w]{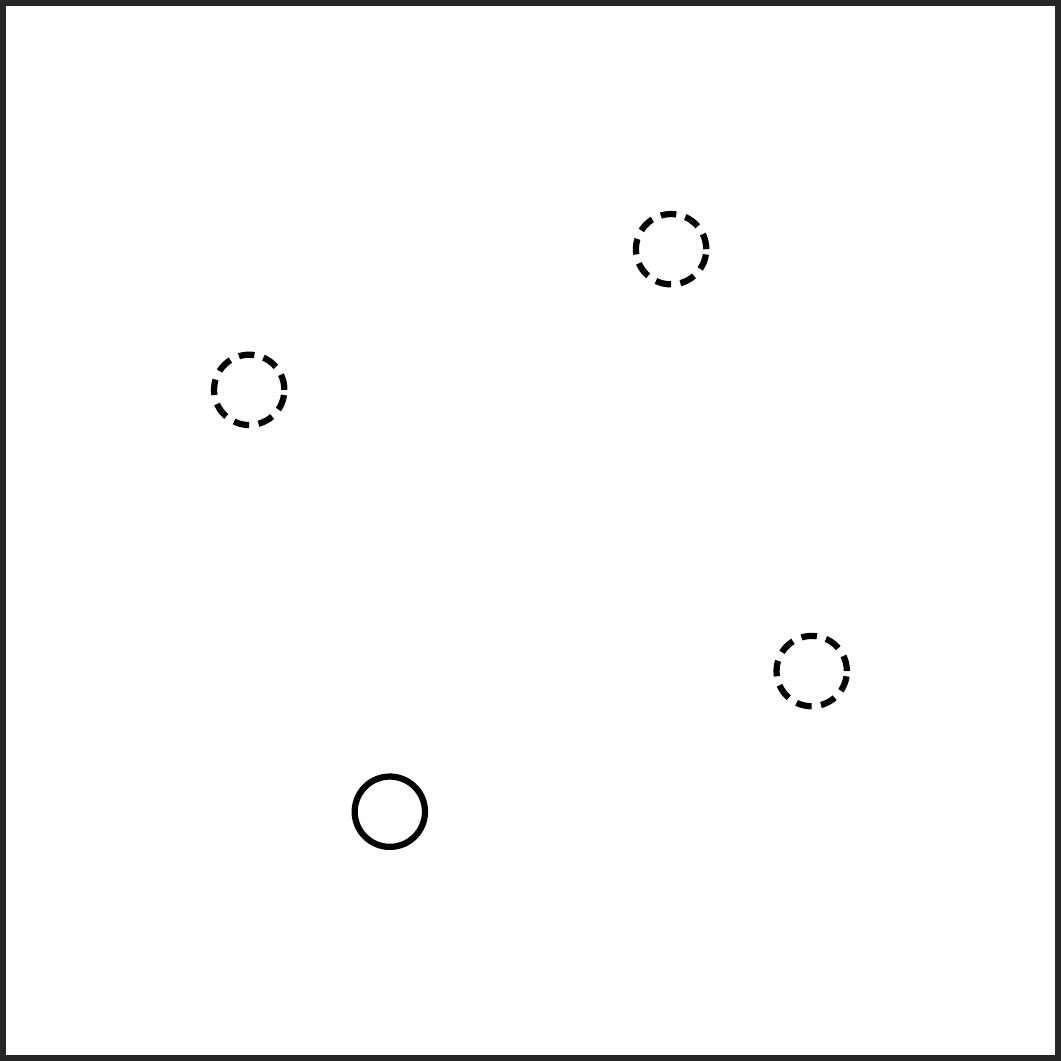}};
                \node[inner sep=0, anchor=east] at ( 0.58*\w,0) {\rotatebox{90}{\includeinkscape[width=\w]{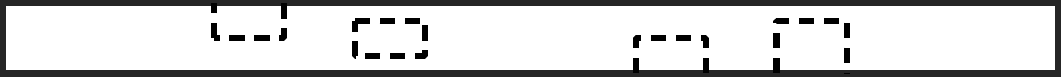}}};
                \tikzstyle{label} = [rounded corners, fill=black, fill opacity=0.1, text opacity=1, inner sep=1.5pt]
                \node[label] at (0,-0.42*\w) {\footnotesize 2~mm diameter defects};
            \end{tikzpicture}
        } &
        \multicolumn{3}{c|}{
            \begin{tikzpicture}
                \node[inner sep=0, anchor=west] at (-0.58*\w,0) {\includeinkscape[width=\w]{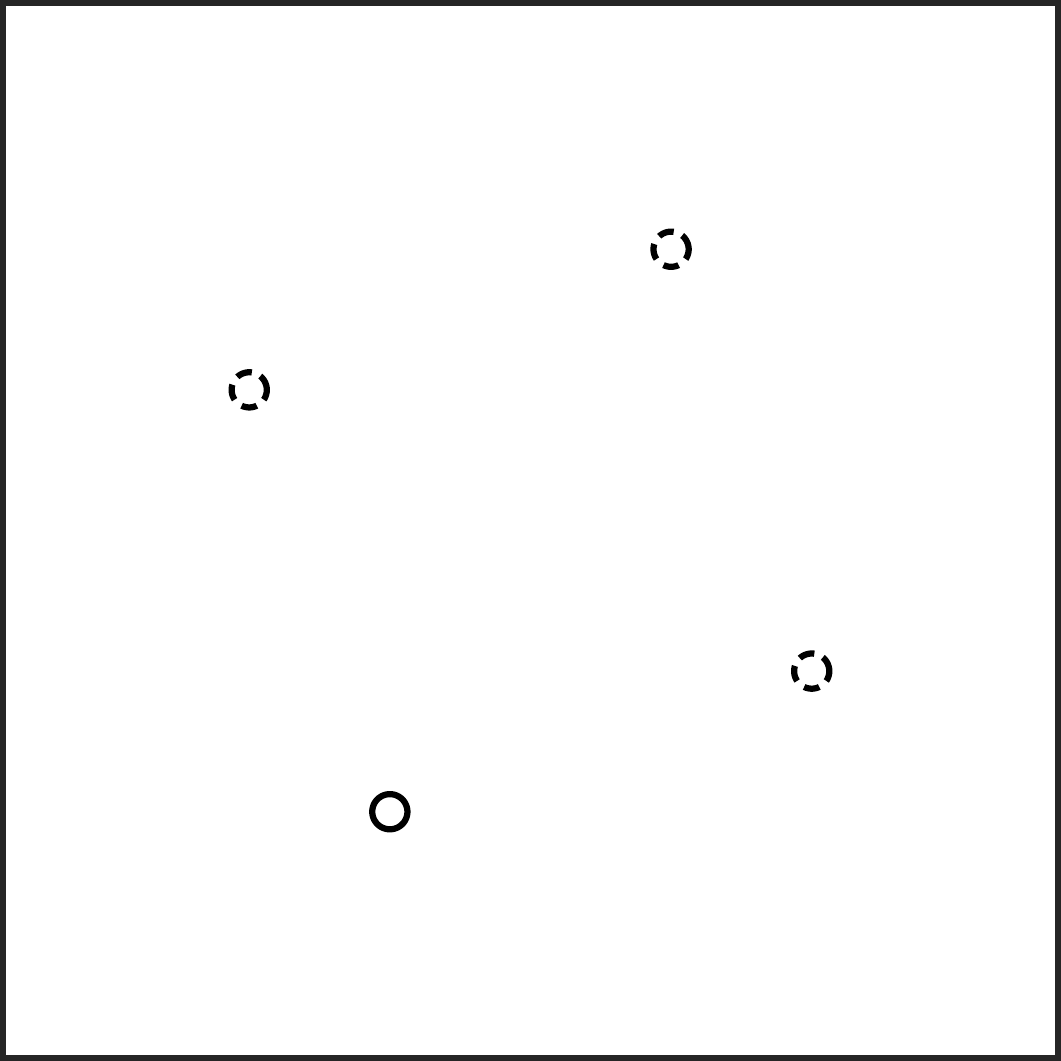}};
                \node[inner sep=0, anchor=east] at ( 0.58*\w,0) {\rotatebox{90}{\includeinkscape[width=\w]{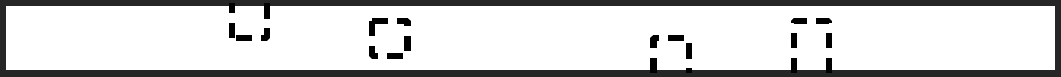}}};
                \tikzstyle{label} = [rounded corners, fill=black, fill opacity=0.1, text opacity=1, inner sep=1.5pt]
                \node[label] at (0,-0.42*\w) {\footnotesize 1~mm diameter defects};
            \end{tikzpicture}
        } \\
        \semithickhline
        & &
        \cref{alg:convex-optimization}    &
        \cref{alg:mean-field-approx}      &
        \cref{alg:approx-message-passing} &
        \cref{alg:convex-optimization}    &
        \cref{alg:mean-field-approx}      &
        \cref{alg:approx-message-passing} \\
        \semithickhline
        \multirow{24.4}{*}{\hspace*{-1mm}\rotatebox{90}{Imaging result across depth}\hspace*{-2mm}} &
        {\footnotesize 0.0~mm} &&&&&& \\[-9pt]
        & \multirow{1.6}{*}{\footnotesize 0.5~mm} &
        \includeinkscape[width=\w]{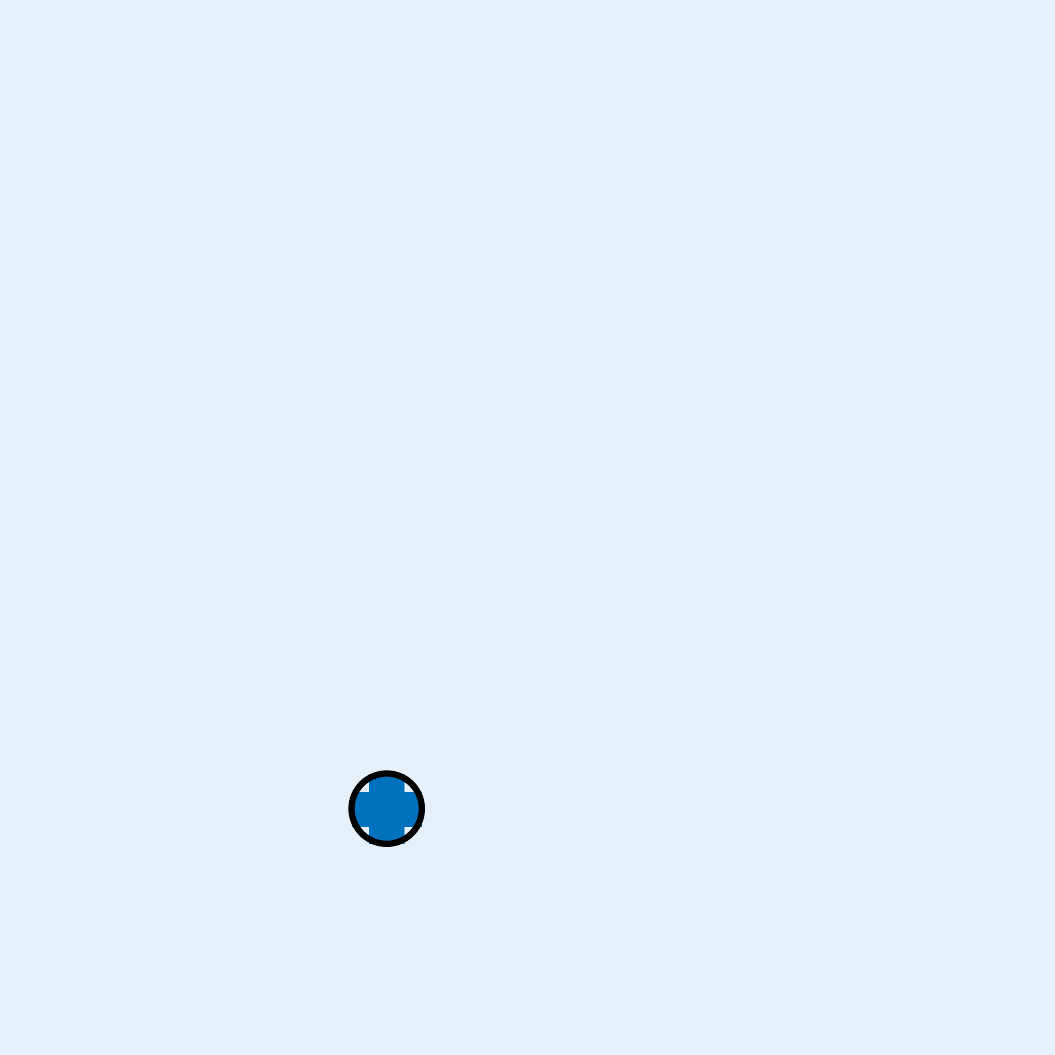} &
        \includeinkscape[width=\w]{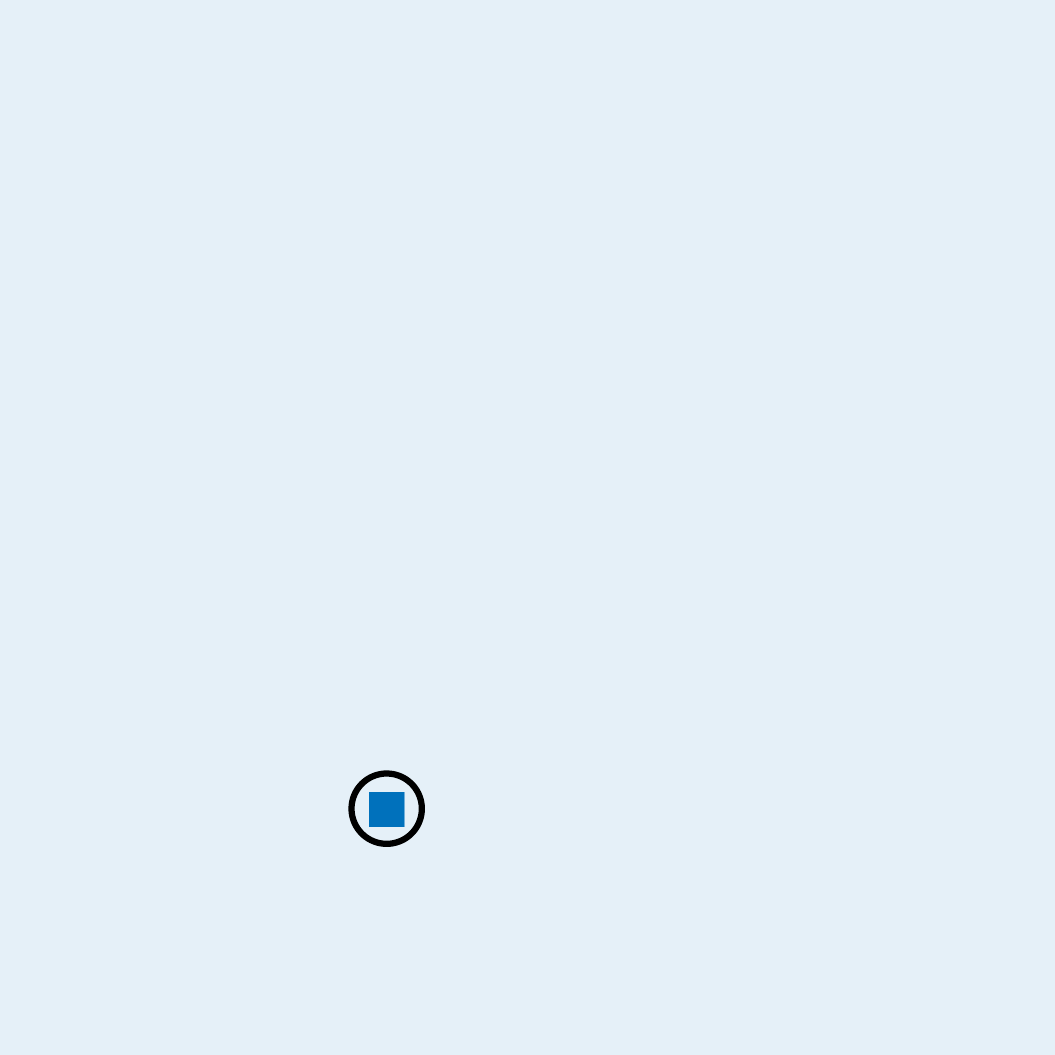}  &
        \includeinkscape[width=\w]{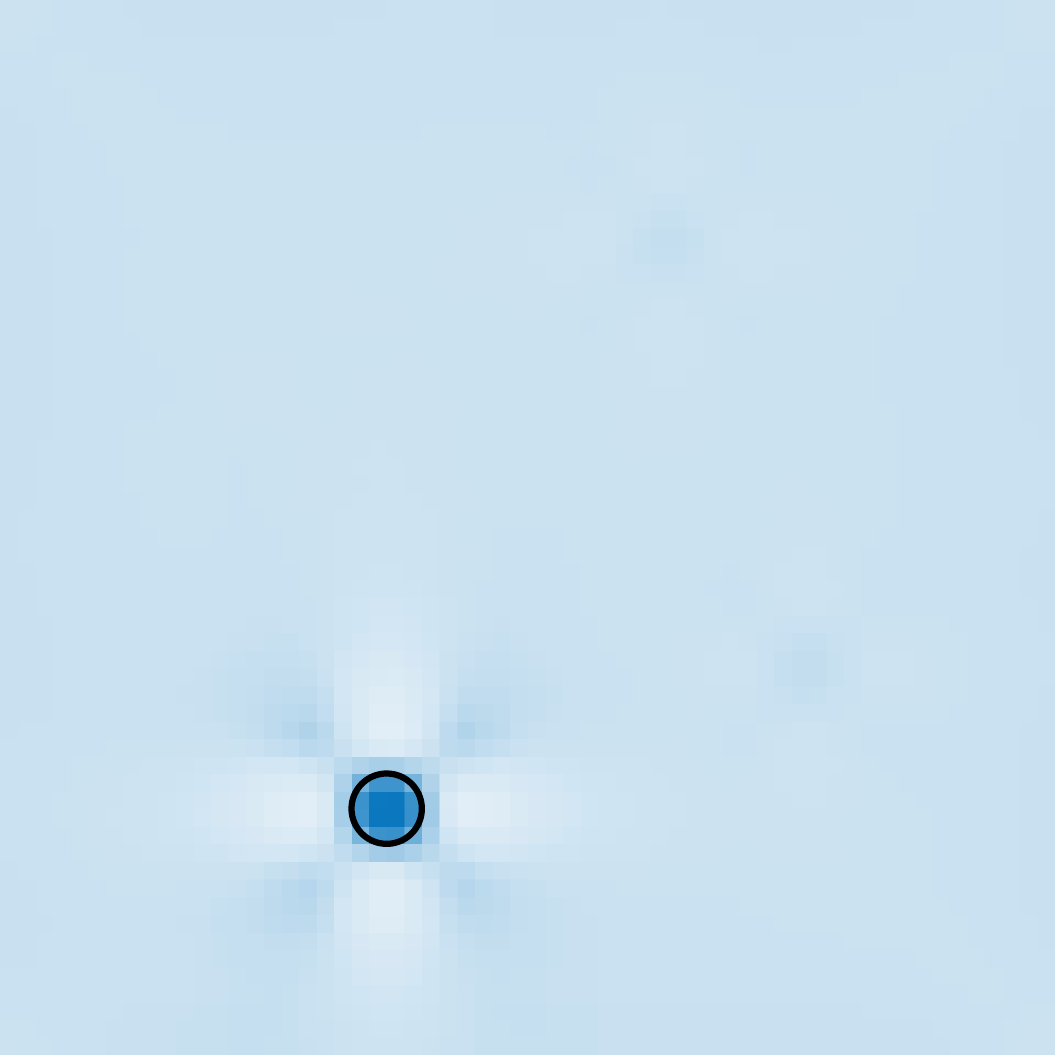} &
        \includeinkscape[width=\w]{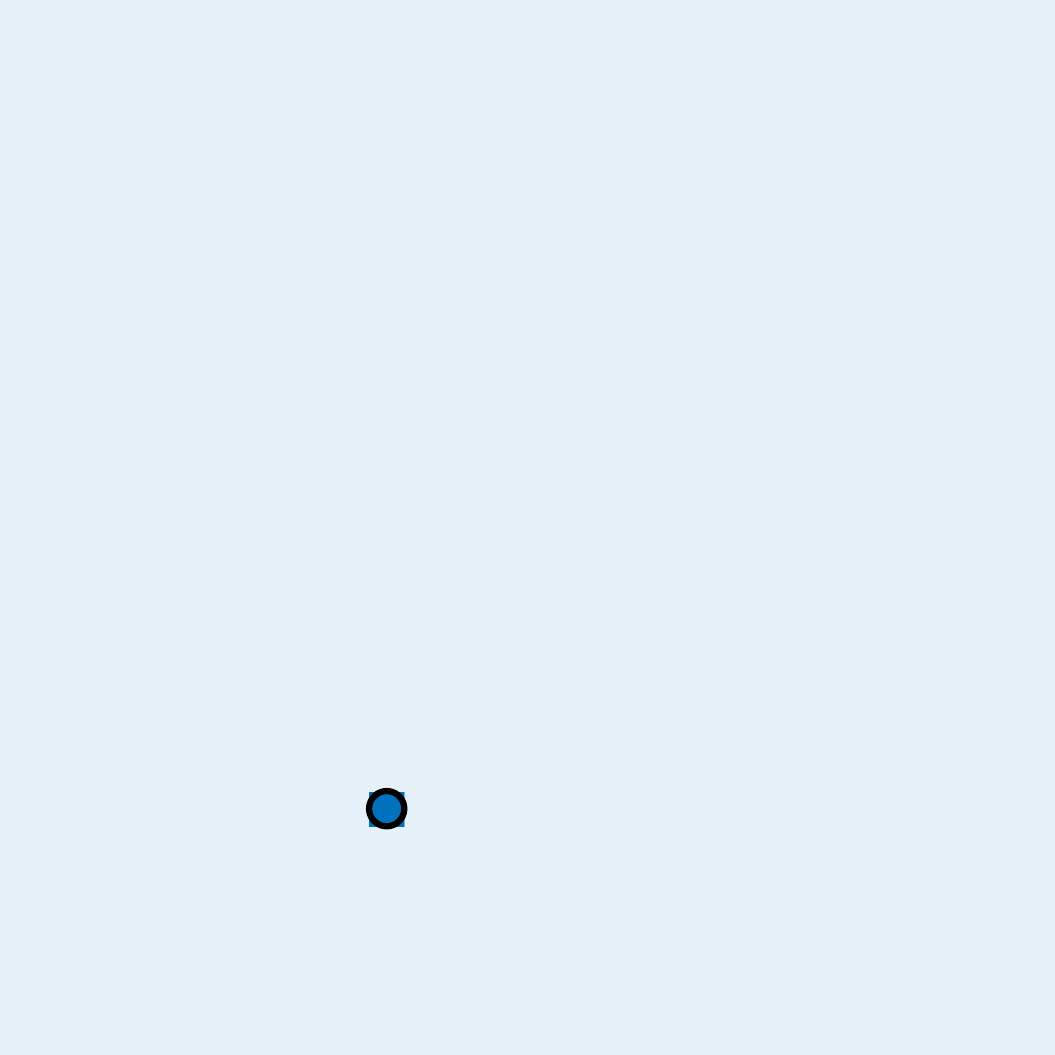} &
        \includeinkscape[width=\w]{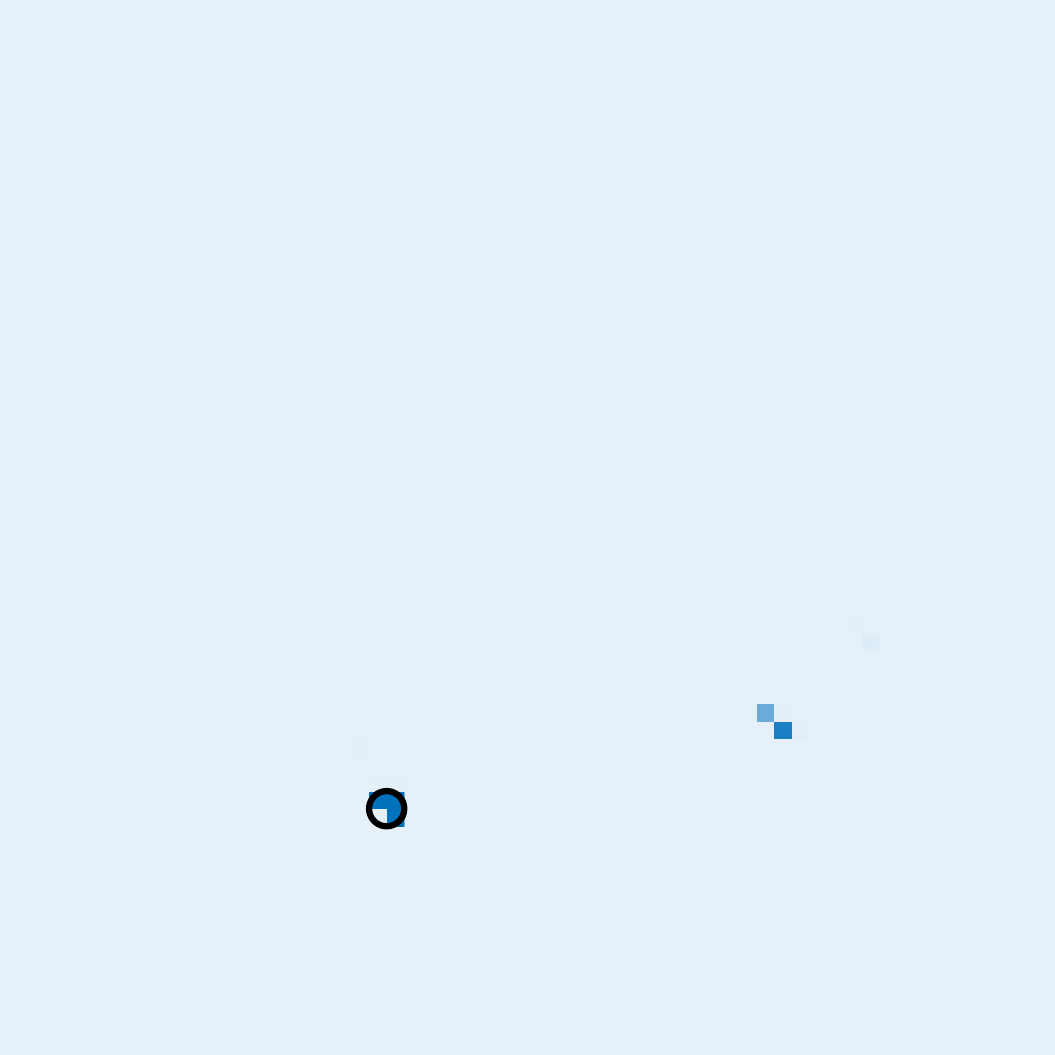}  &
        \includeinkscape[width=\w]{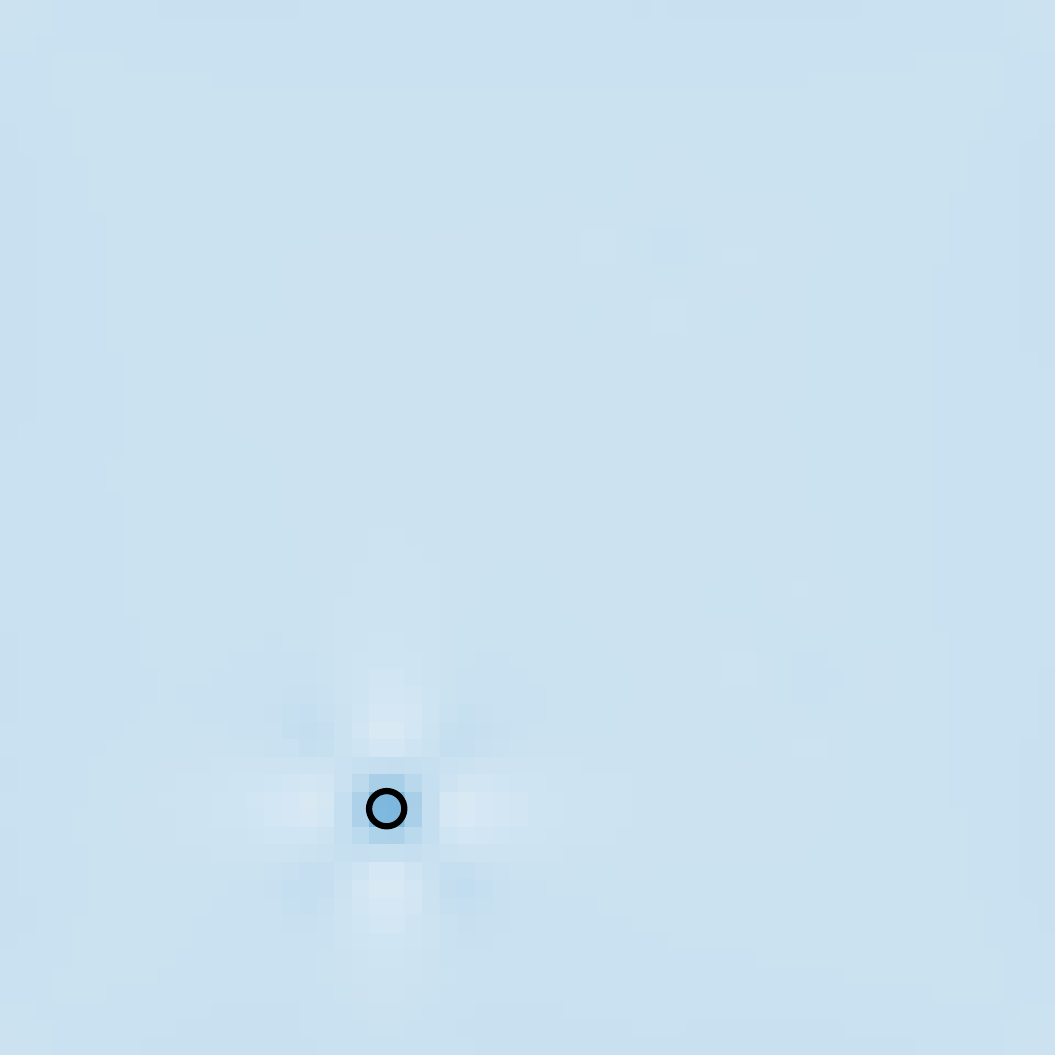} \\
        & \multirow{1.6}{*}{\footnotesize 1.0~mm} &
        \includeinkscape[width=\w]{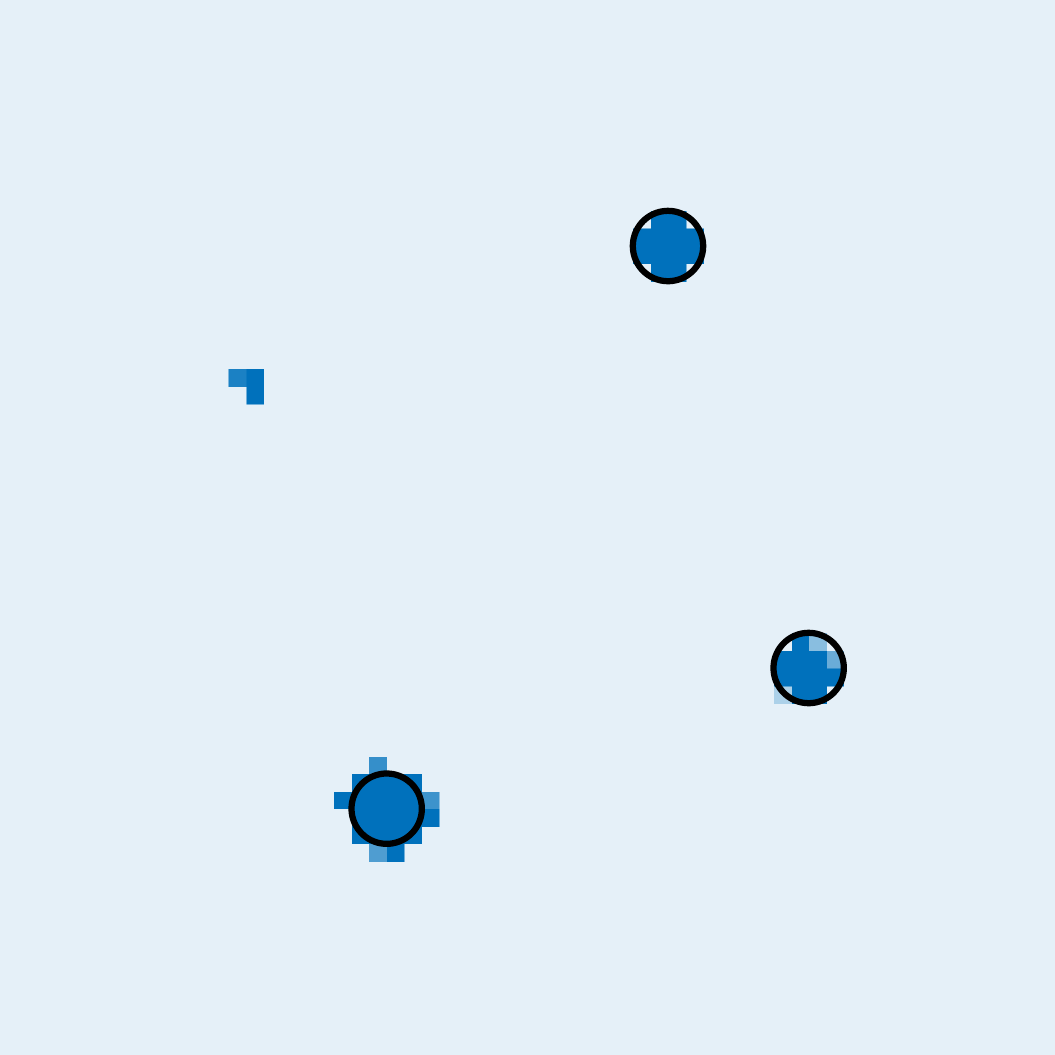} &
        \includeinkscape[width=\w]{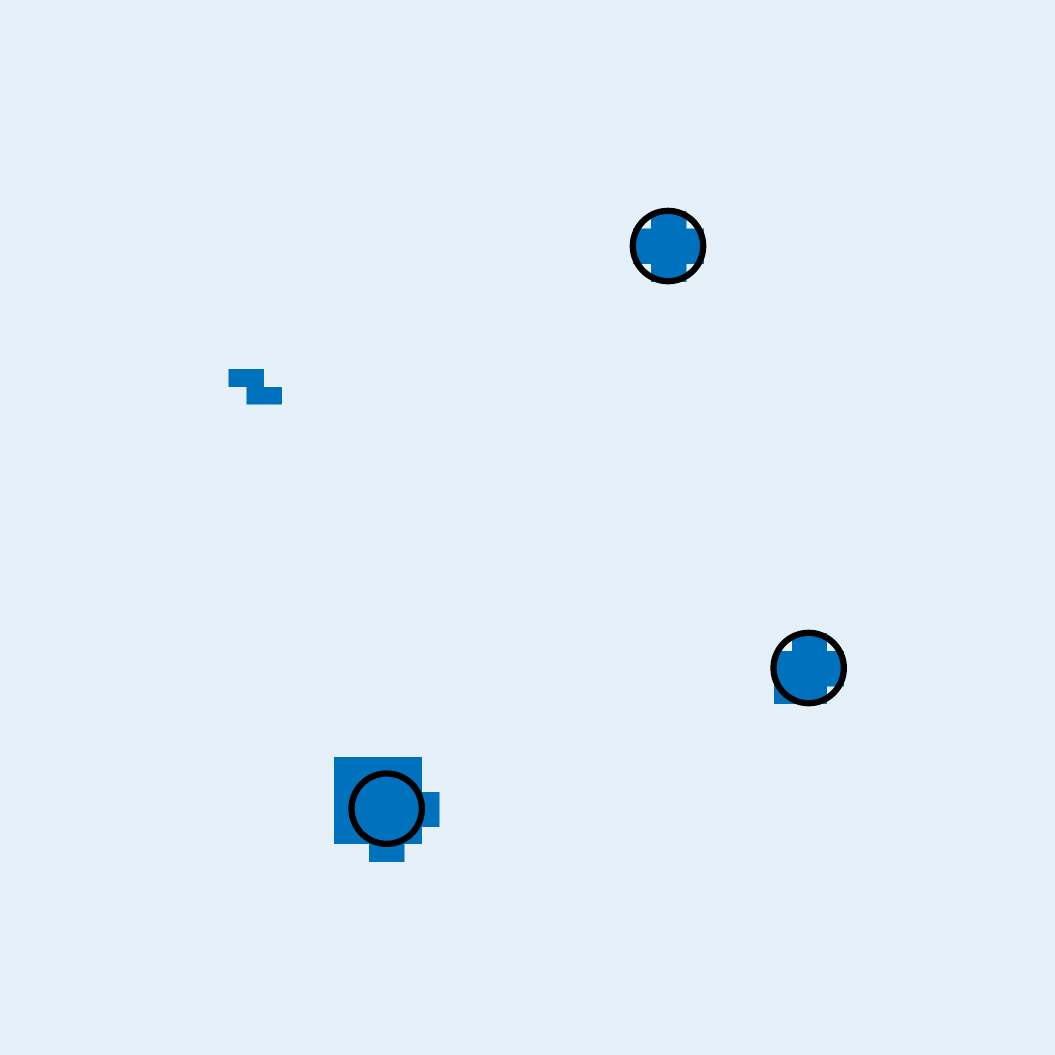}  &
        \includeinkscape[width=\w]{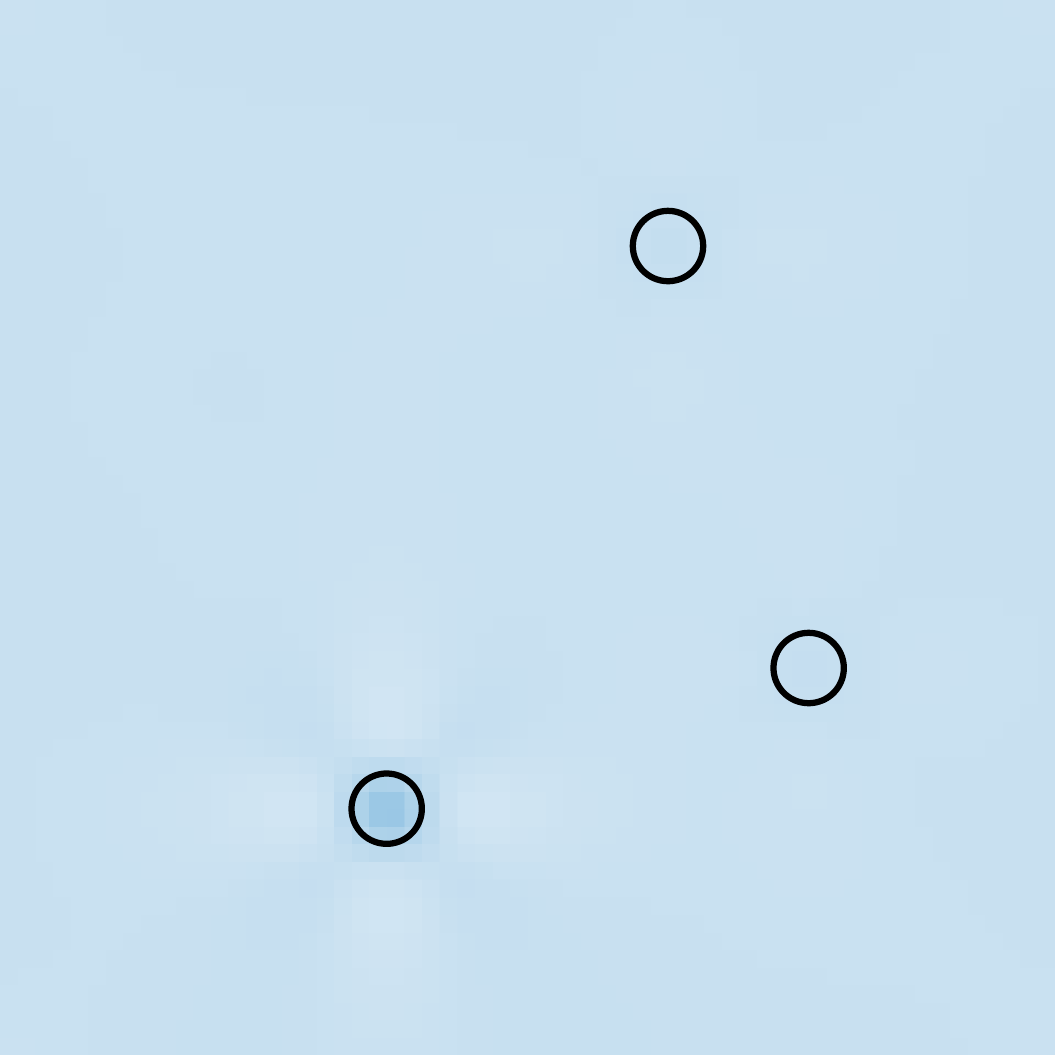} &
        \includeinkscape[width=\w]{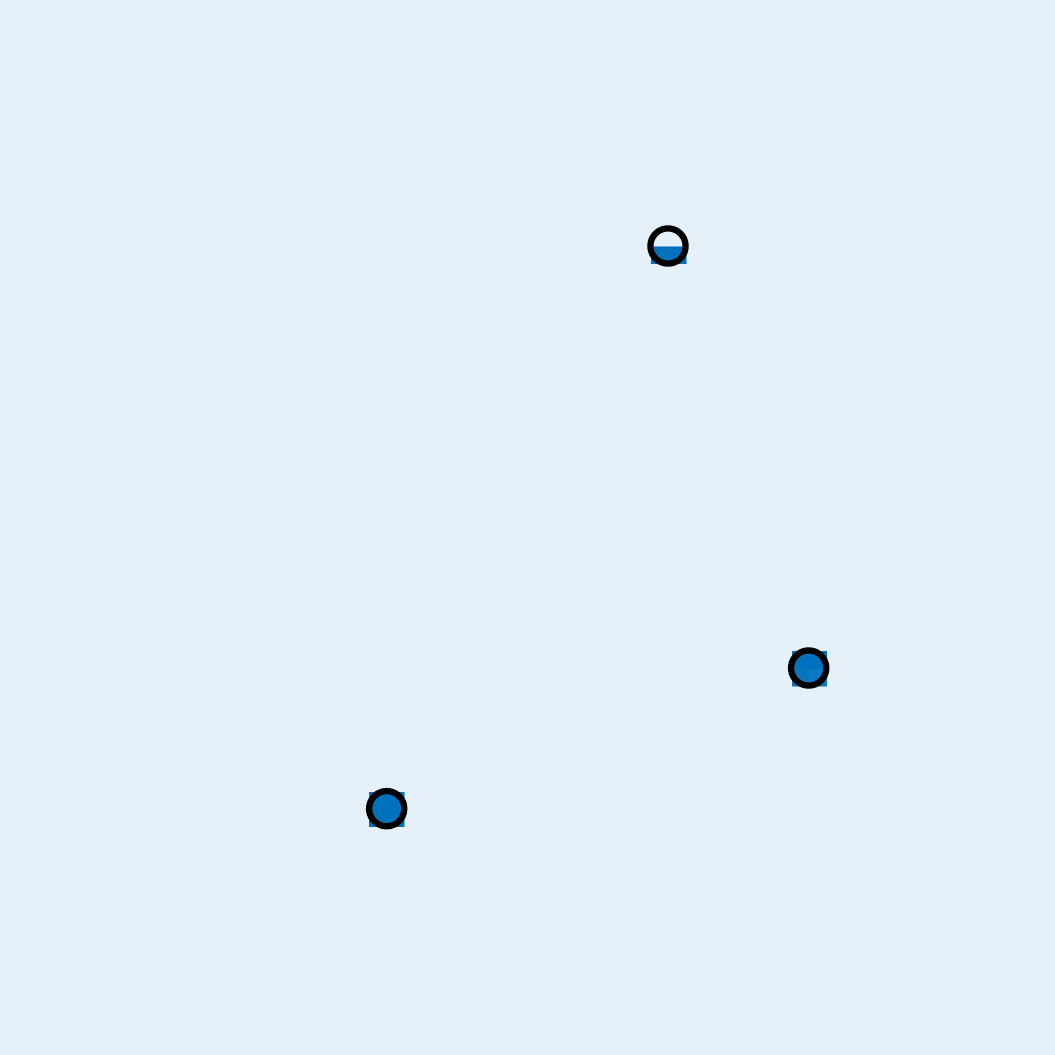} &
        \includeinkscape[width=\w]{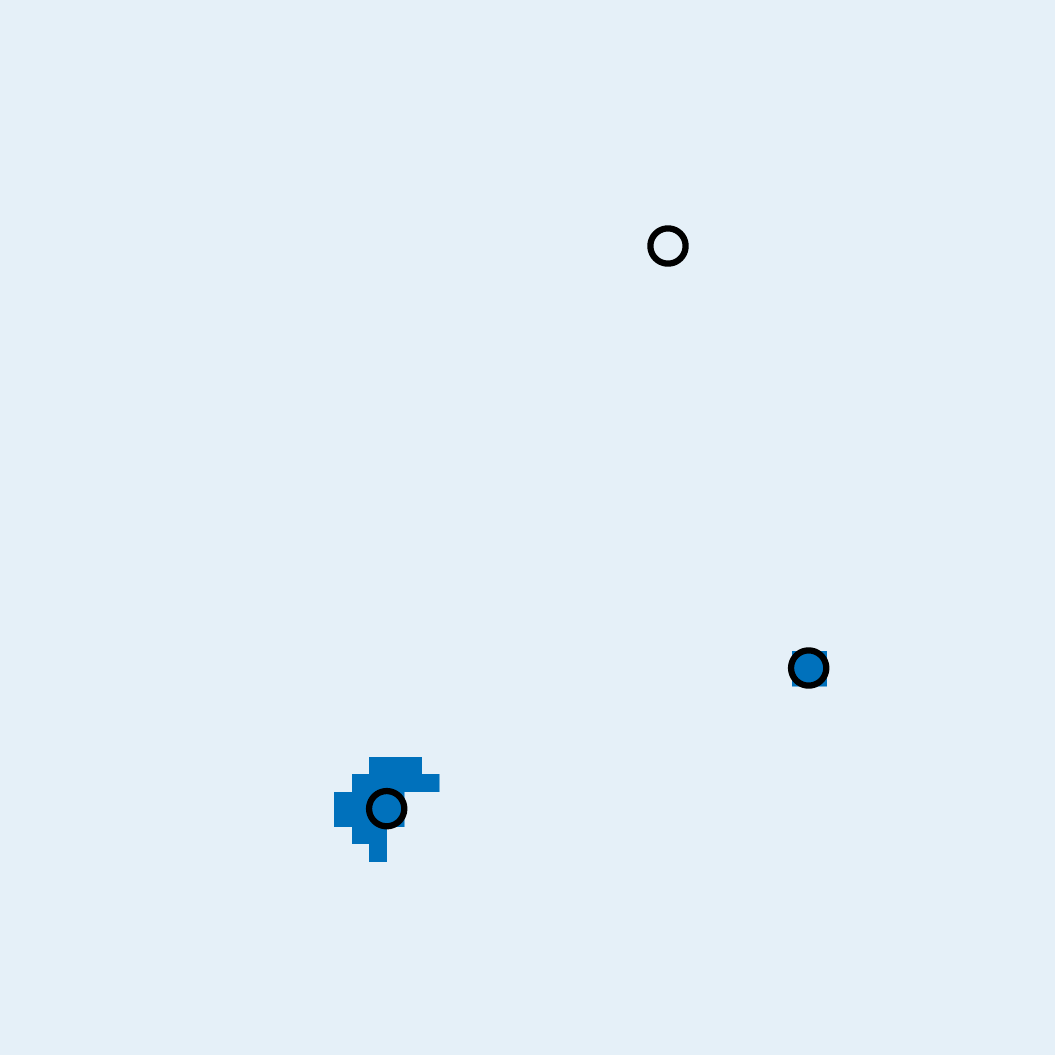}  &
        \includeinkscape[width=\w]{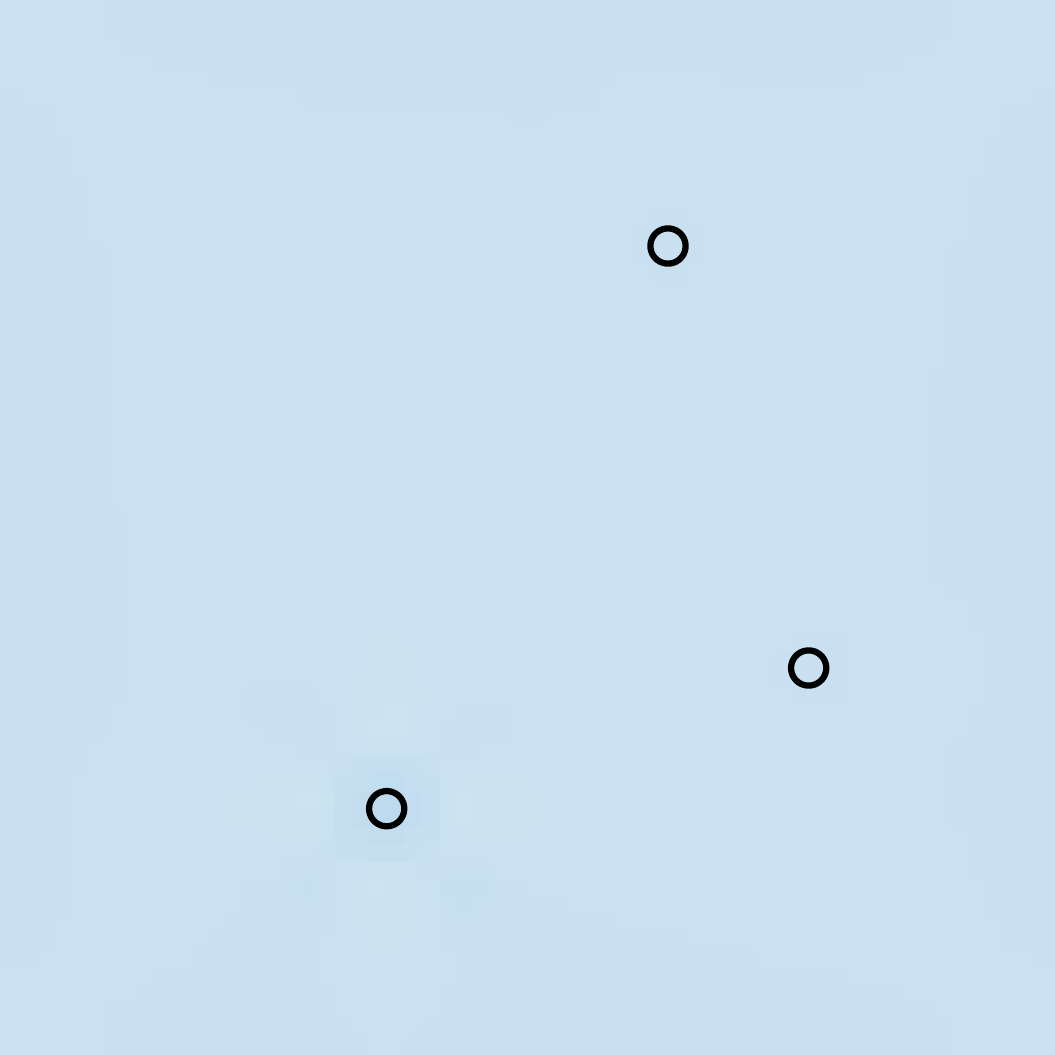}\\
        & \multirow{1.6}{*}{\footnotesize 1.5~mm} &
        \includeinkscape[width=\w]{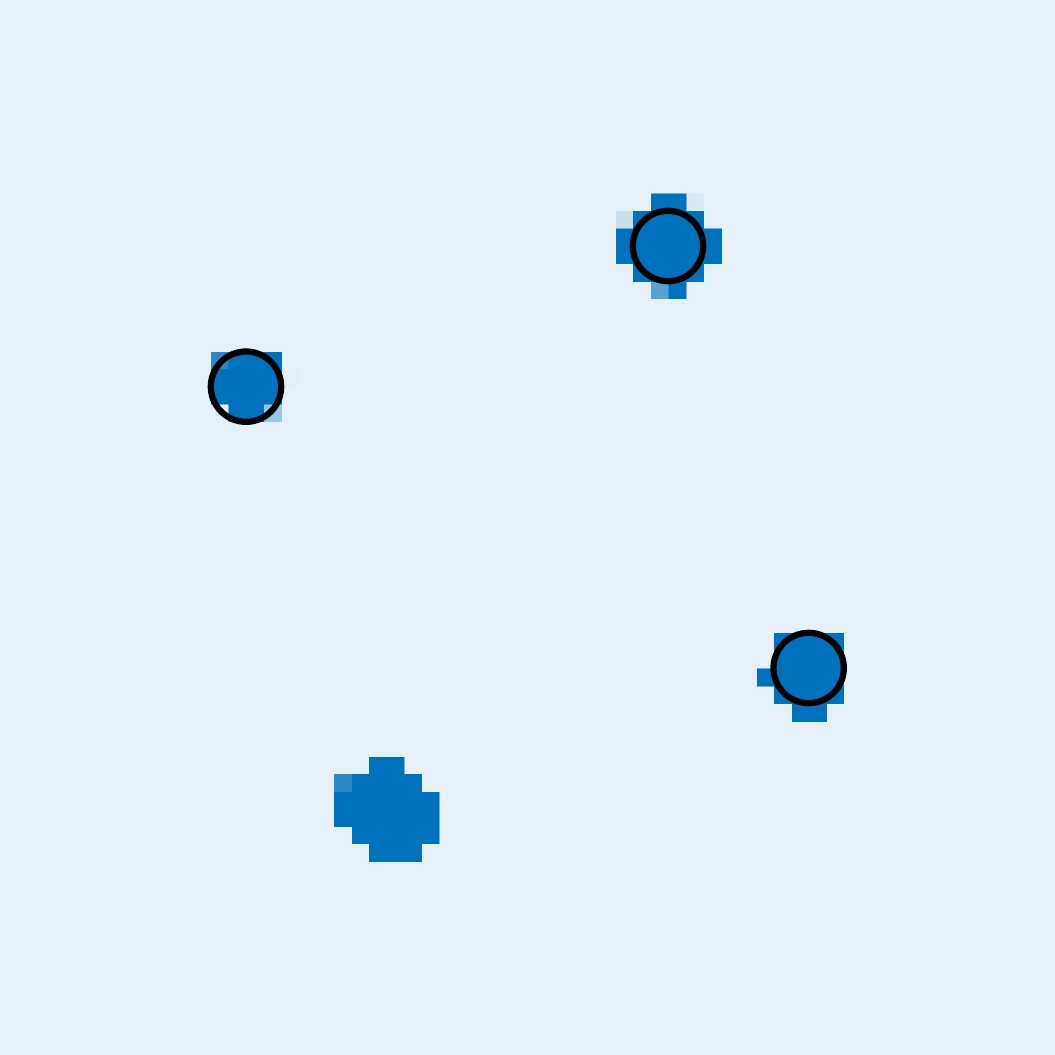} &
        \includeinkscape[width=\w]{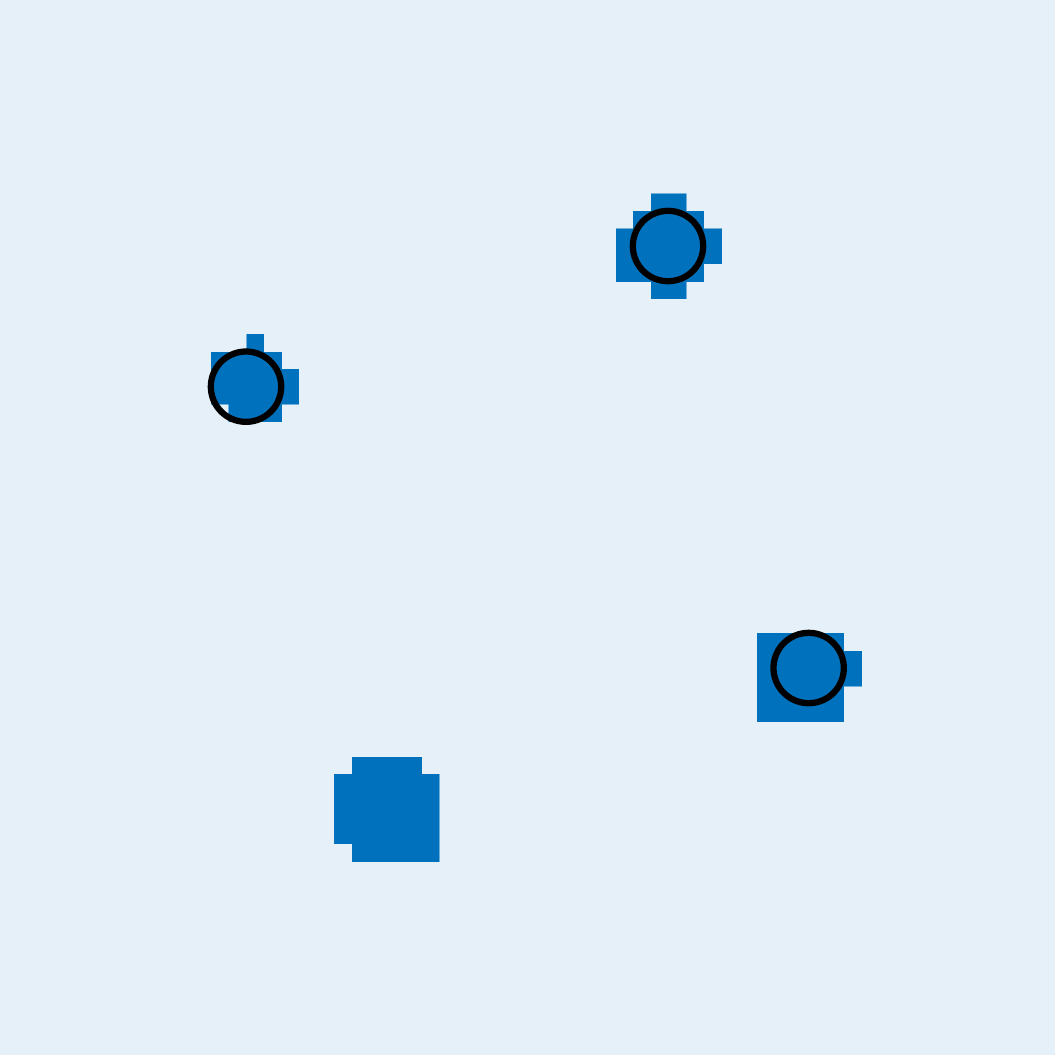}  &
        \includeinkscape[width=\w]{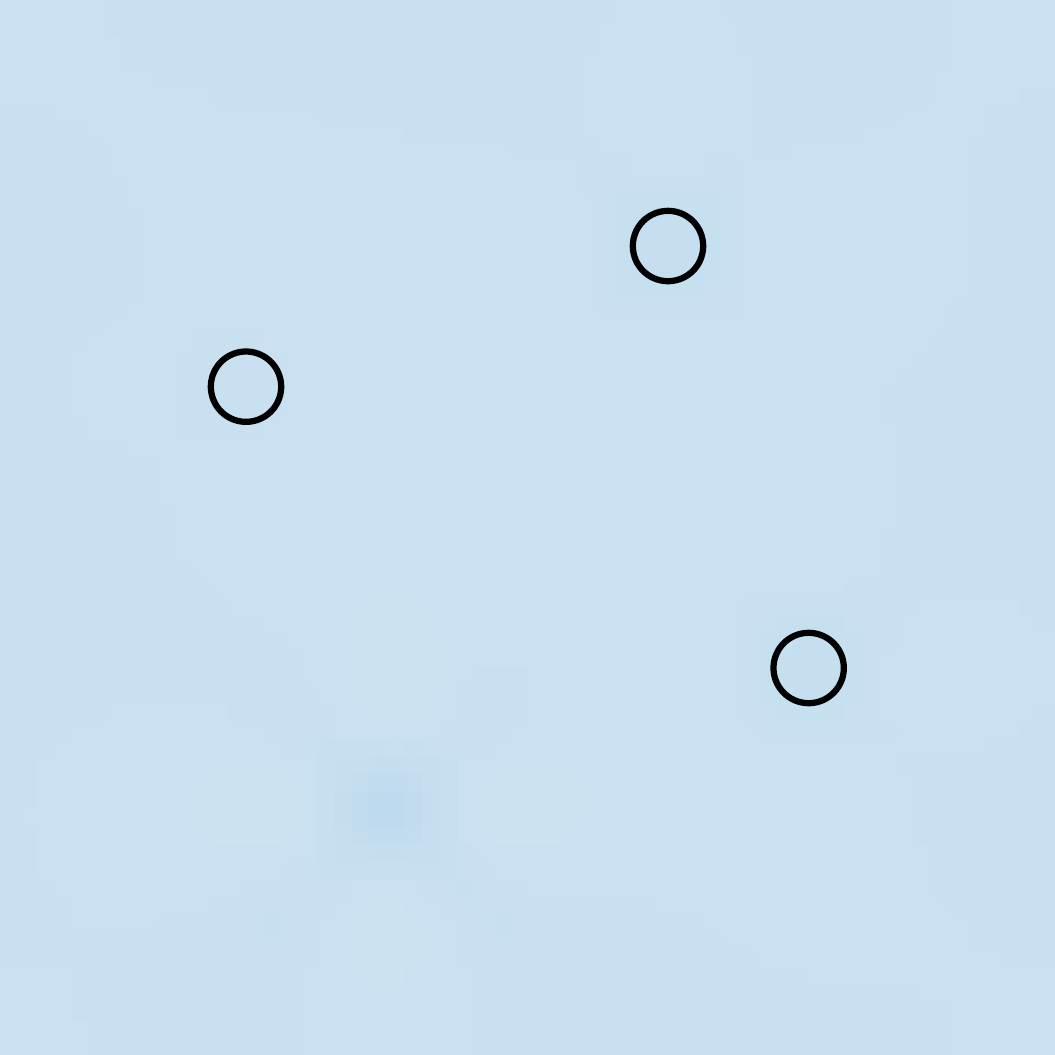} &
        \includeinkscape[width=\w]{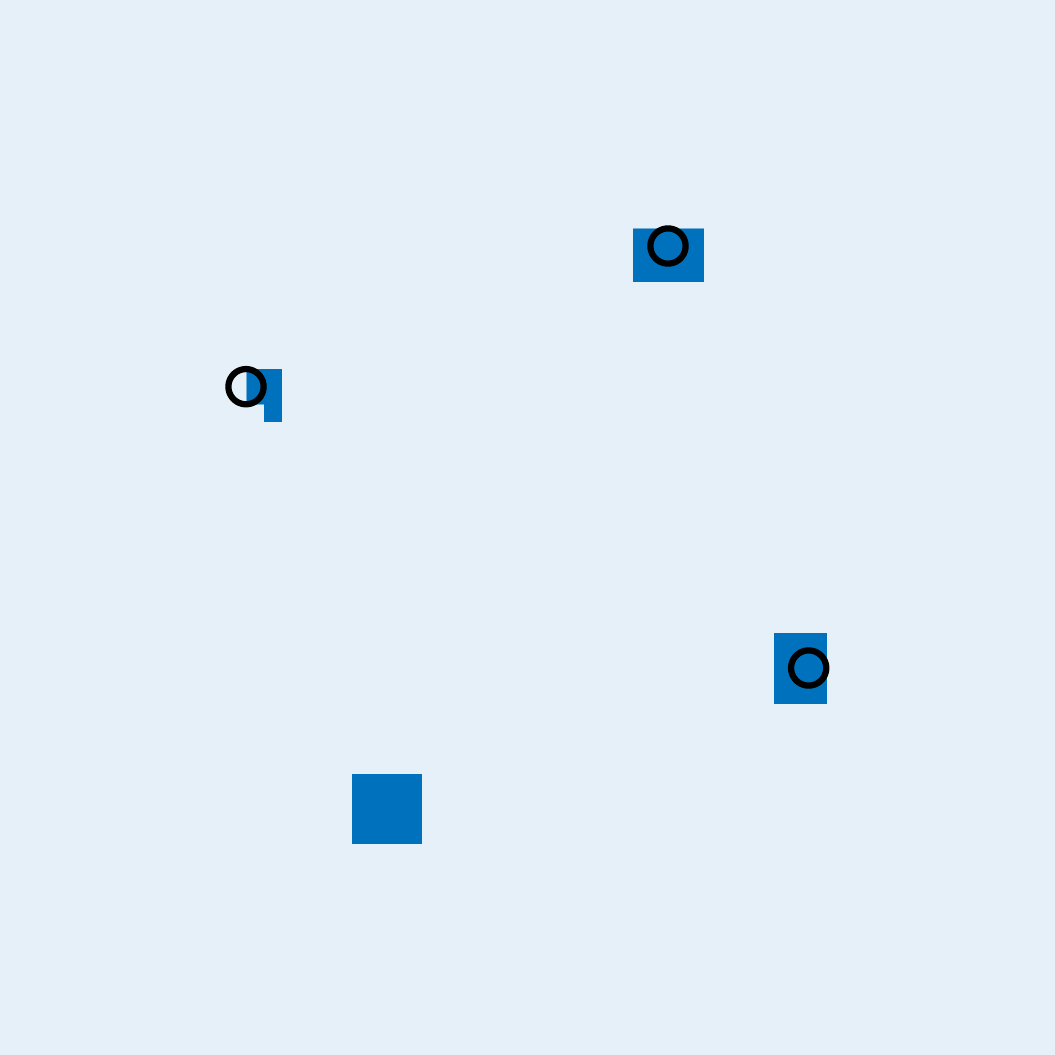} &
        \includeinkscape[width=\w]{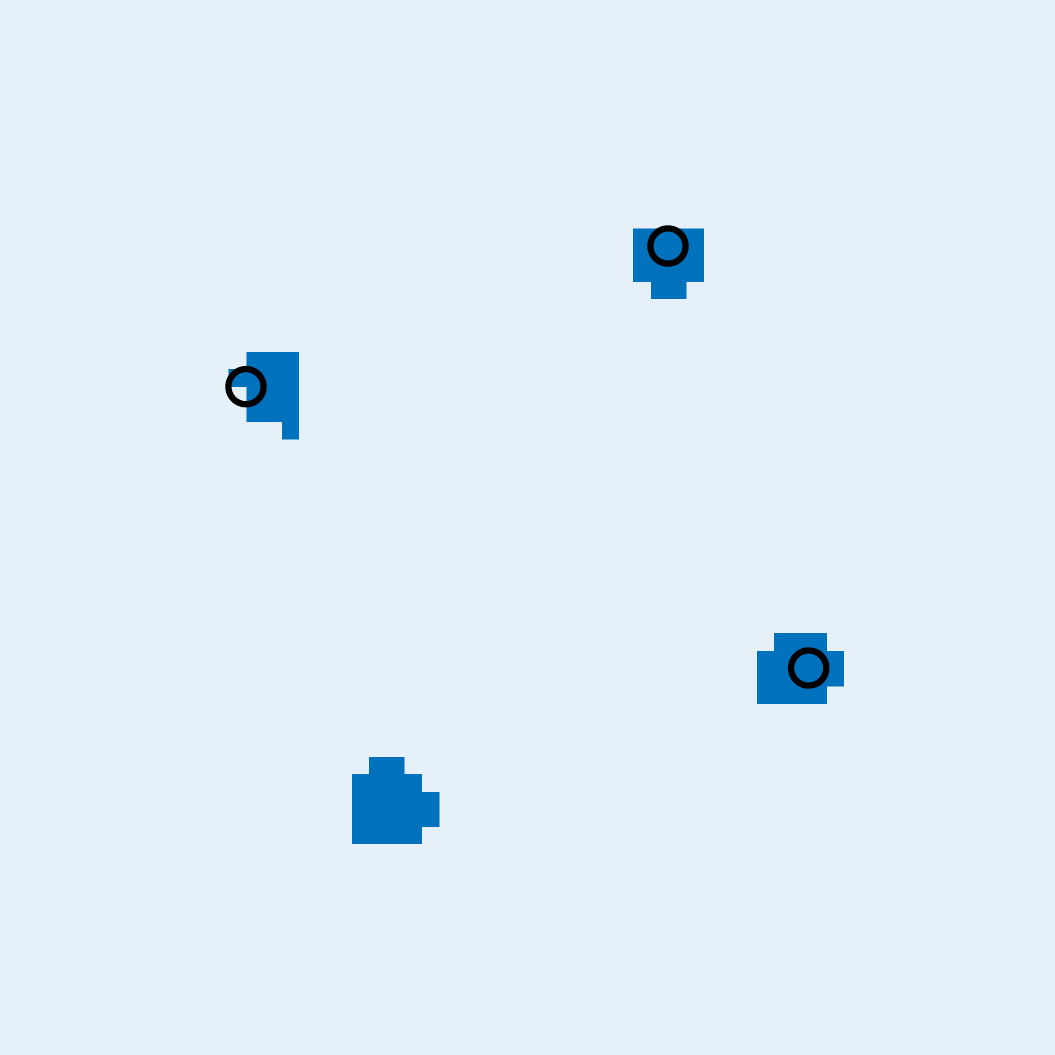}  &
        \includeinkscape[width=\w]{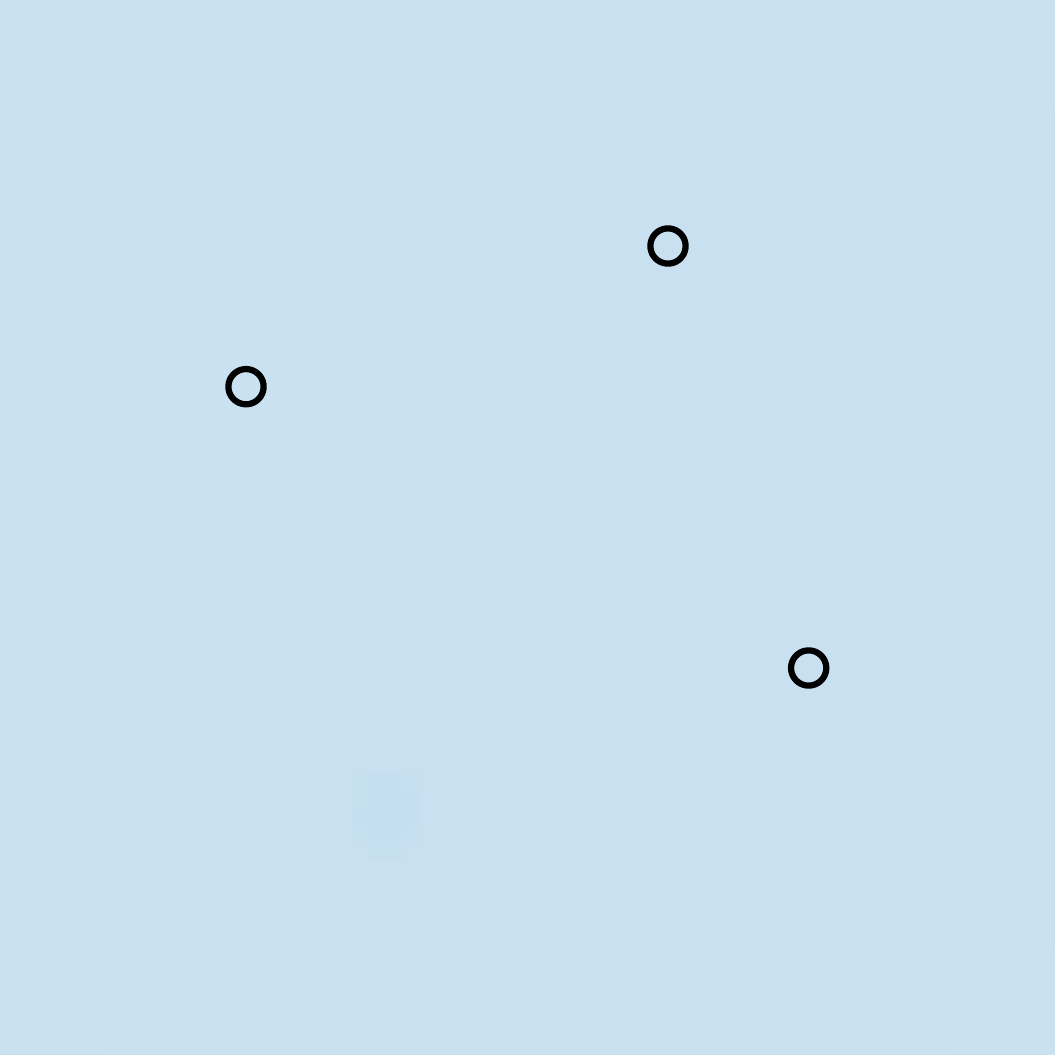} \\
        & \multirow{1.2}{*}{\footnotesize 2.0~mm} &
        \includeinkscape[width=\w]{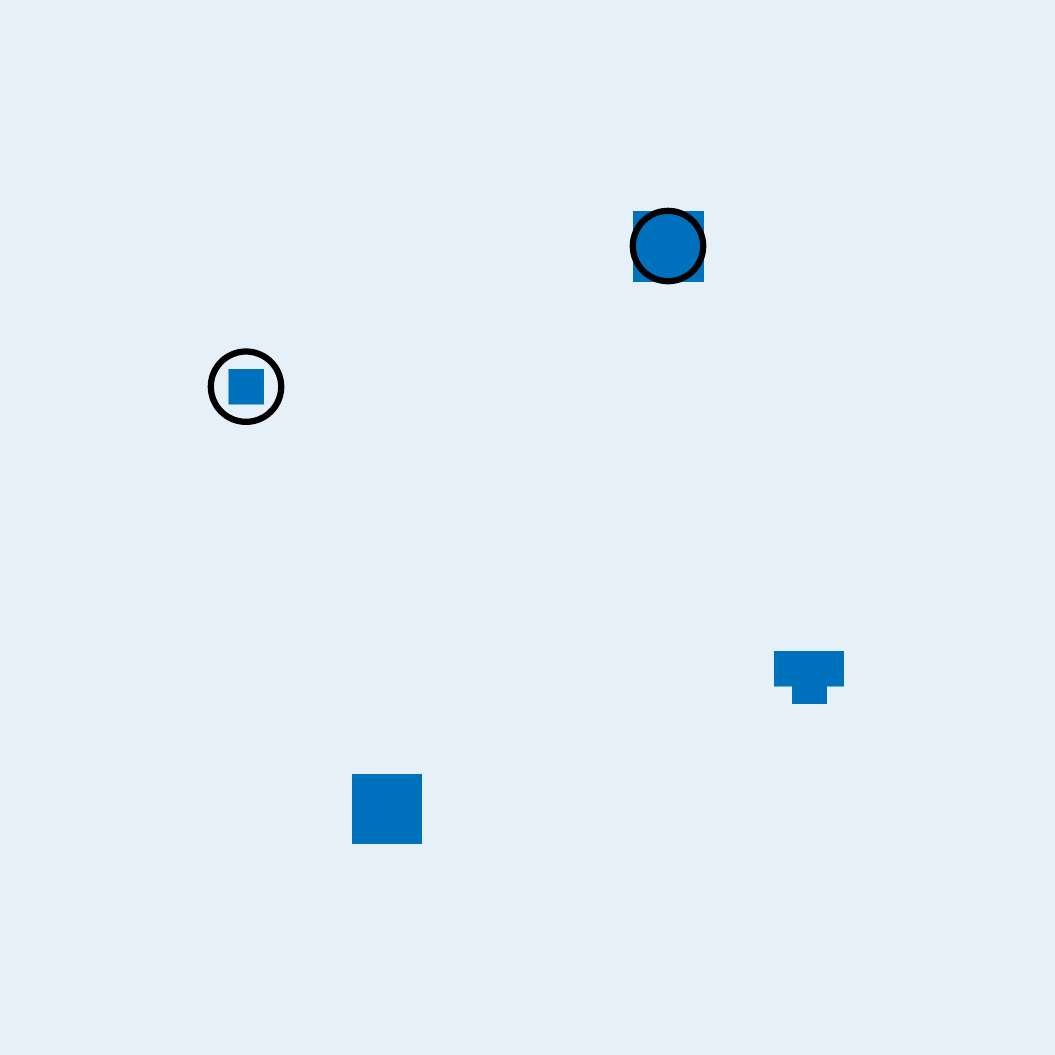} &
        \includeinkscape[width=\w]{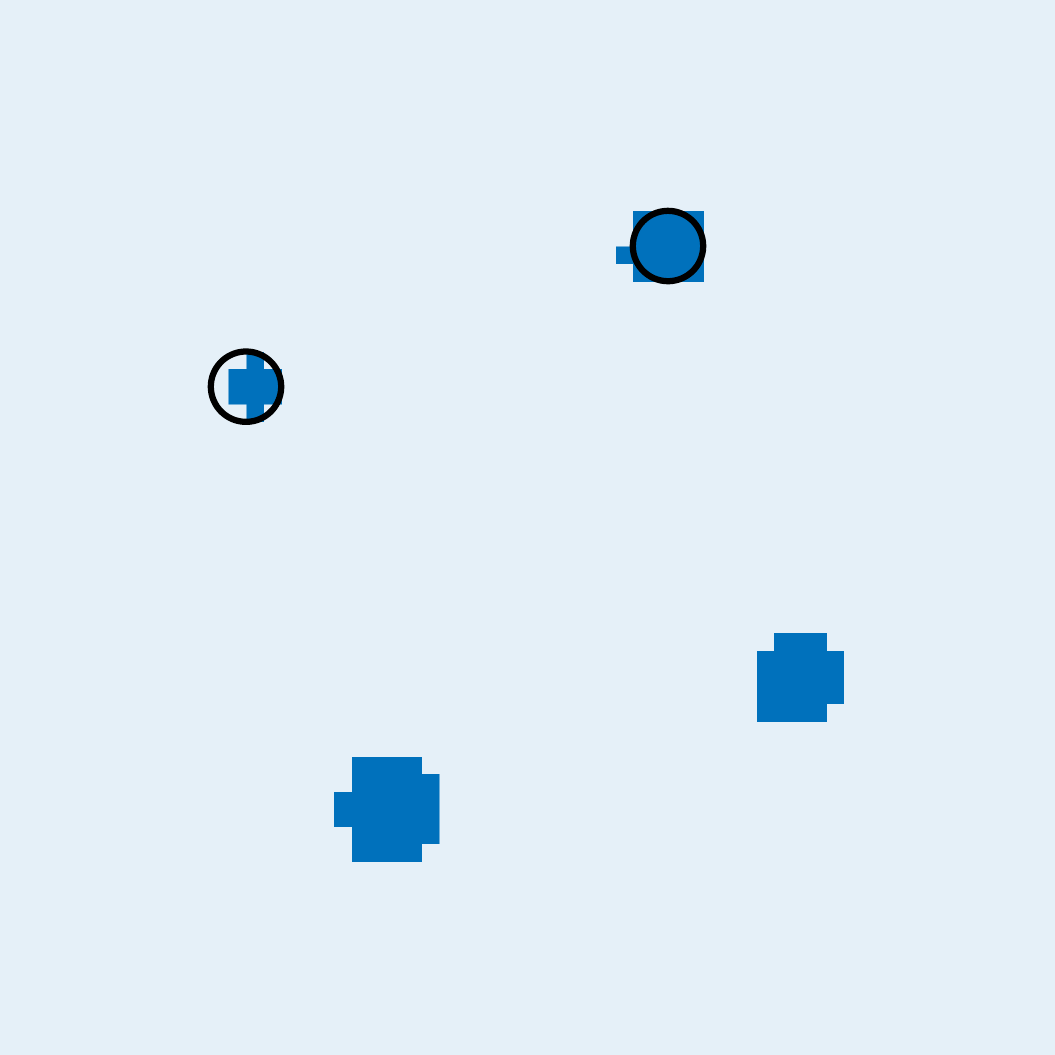}  &
        \includeinkscape[width=\w]{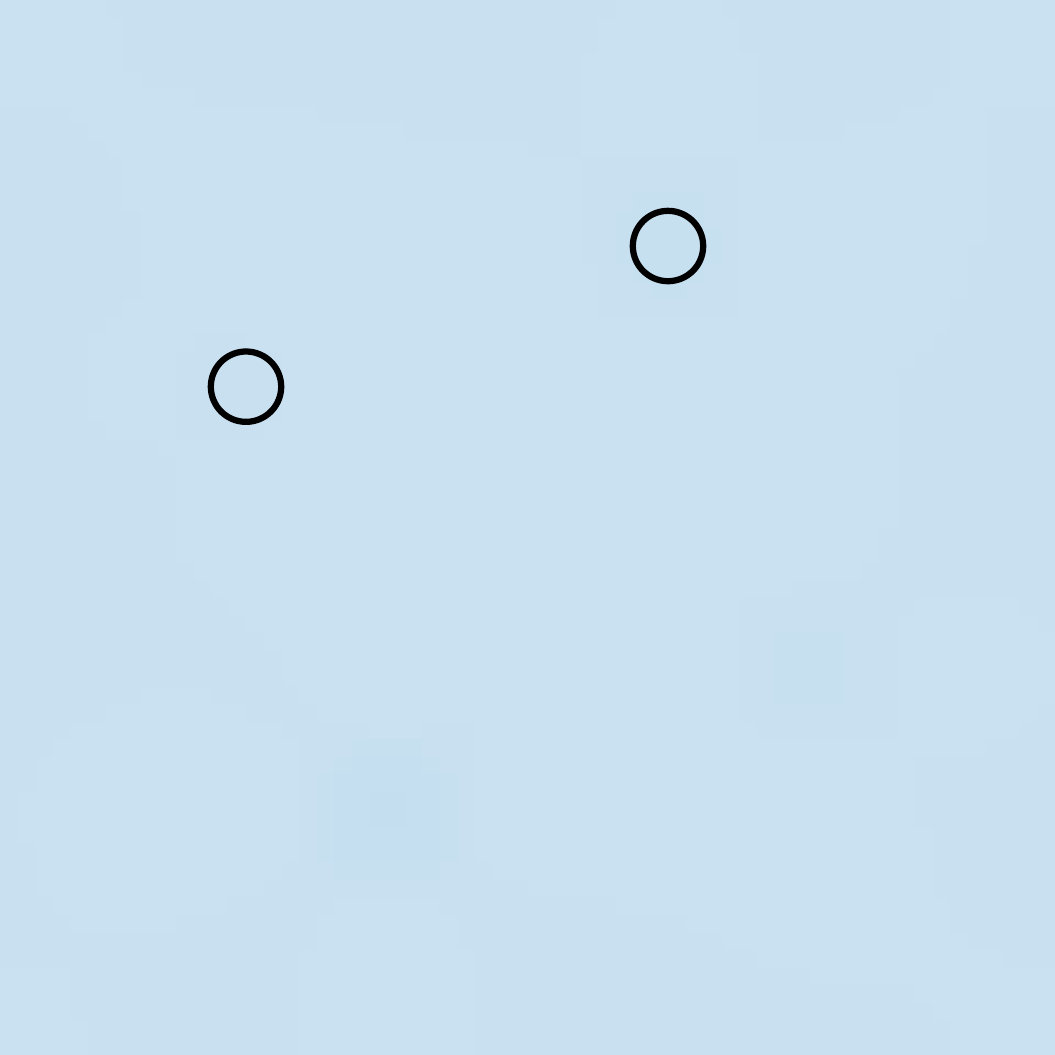} &
        \includeinkscape[width=\w]{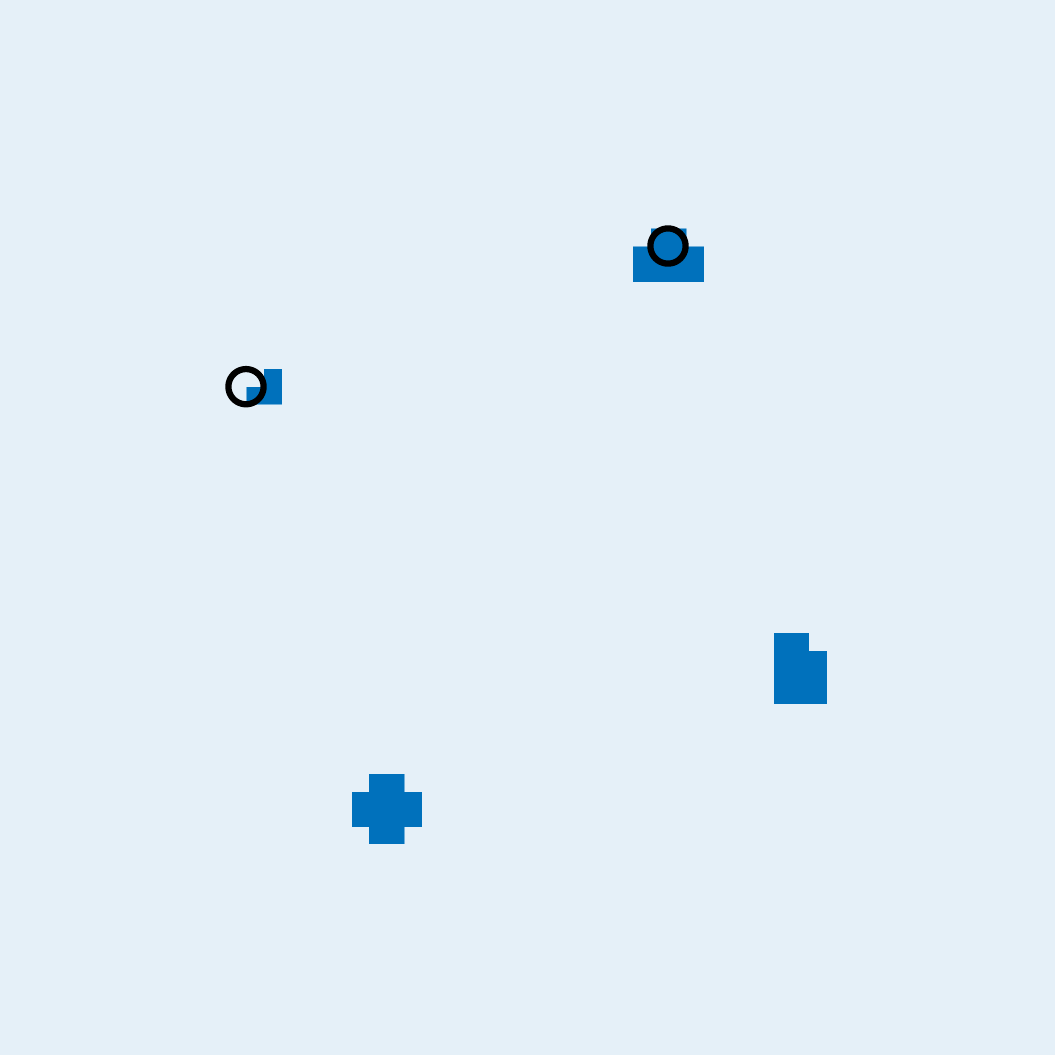} &
        \includeinkscape[width=\w]{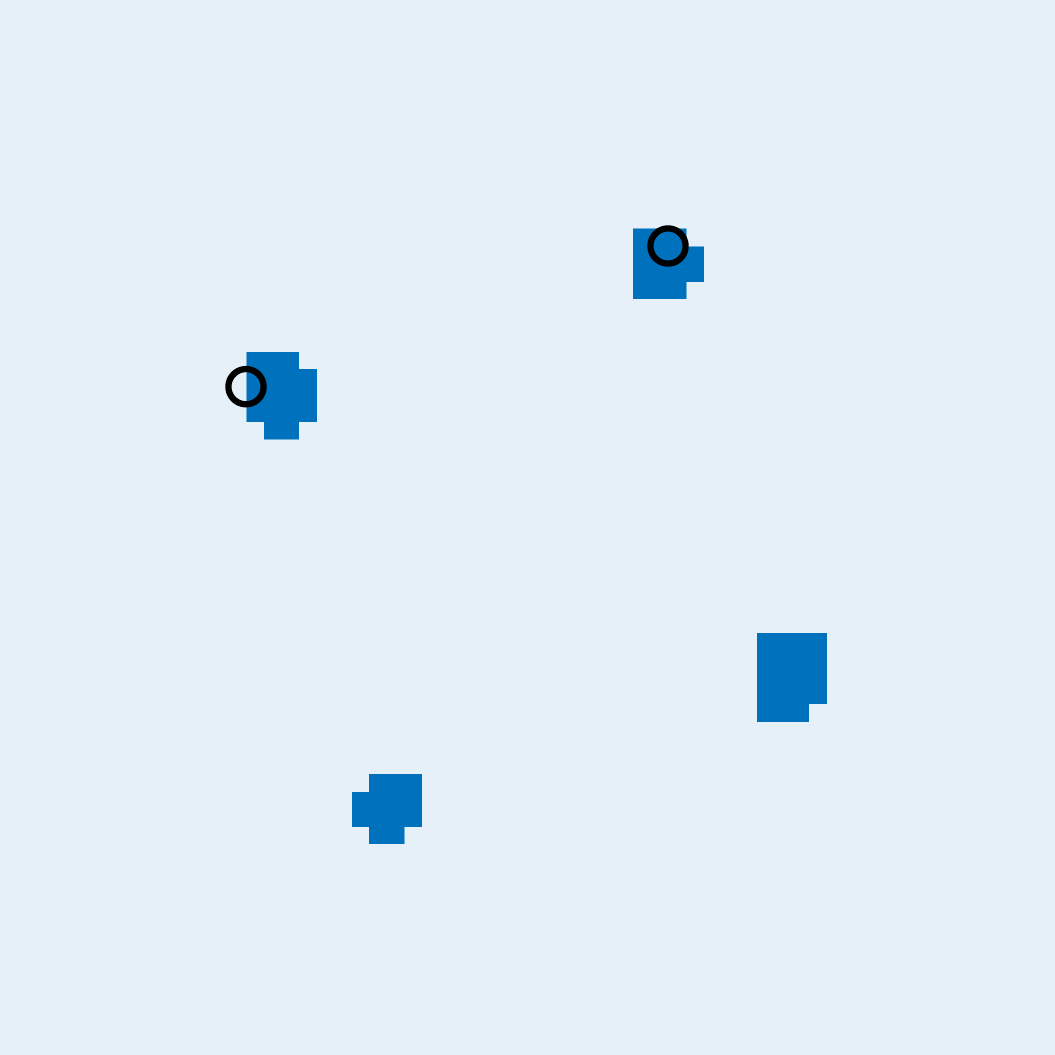}  &
        \includeinkscape[width=\w]{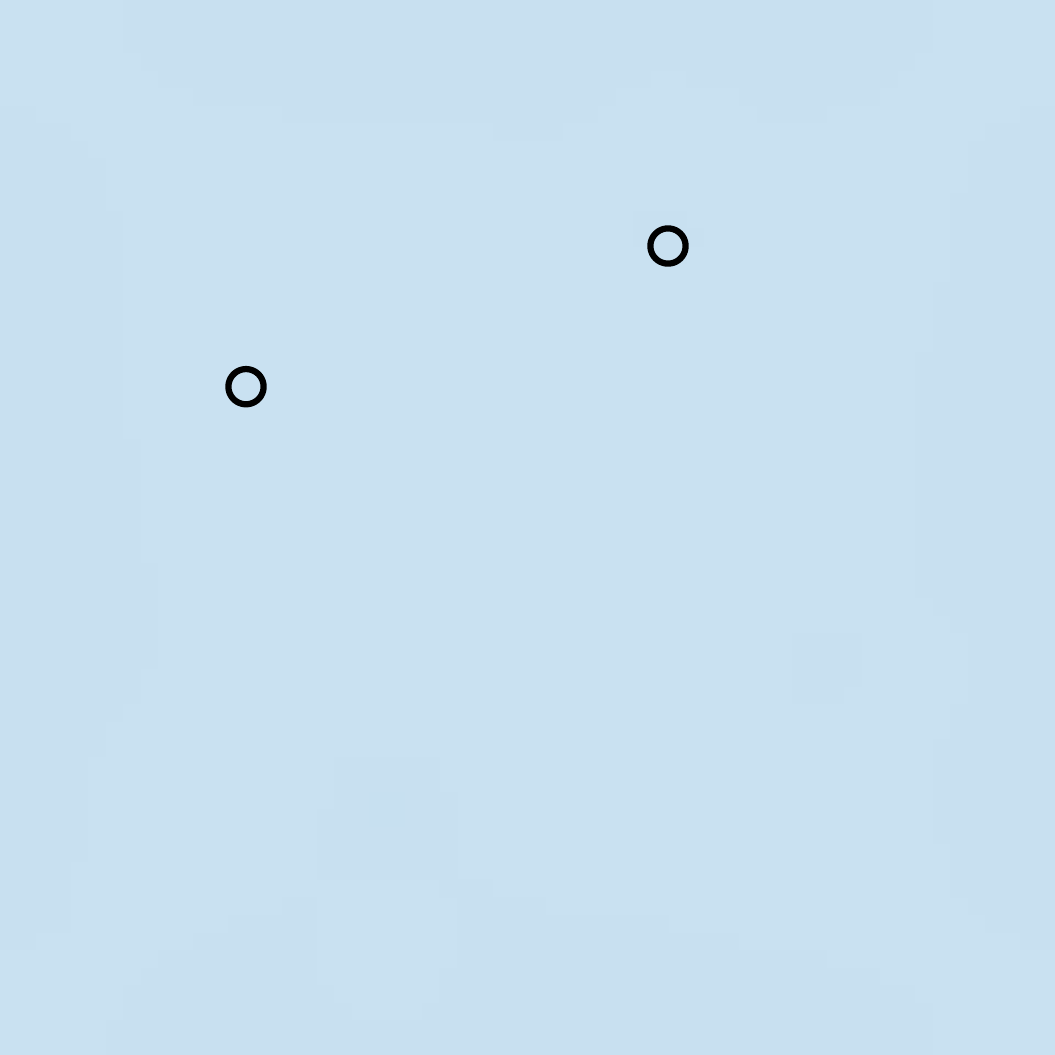} \\
        \thickhline
    \end{tabular}
\end{table}
From the results, several observations can be made:
\begin{enumerate}[nosep, leftmargin=*]
    \item Both \cref{alg:convex-optimization} and \cref{alg:mean-field-approx} are able to image defects as small as 1~mm. This surpasses the resolution limits imposed by the Nyquist theorem, given that the sensors are spaced 4~mm apart. This high imaging resolution is achieved through explicit modeling of the relationships between measurement quantities and defects, as well as the use of binary vector recovery algorithms.
    \item In contrast, \cref{alg:approx-message-passing} yields poor imaging result. The imaged defects are presence in only the topmost layer, with noticeable artifacts surrounding the defects. This outcome aligns with the discussion in \cref{sec:algorithmic-experiment}, which states that \cref{alg:approx-message-passing} relies on message passing without an explicit optimization objective. In comparison, \cref{alg:convex-optimization} and \cref{alg:mean-field-approx} leverage a convex optimization objective and a \ac{KL}-divergence minimization objective, respectively, resulting in superior imaging performance.
    \item In the imaging results produced by \cref{alg:convex-optimization} and \cref{alg:mean-field-approx}, some defects are imaged as being as deep as 2~mm, even when the actual defect is shallower. This discrepancy is a physical limitation rather than an algorithmic issue. It arises because the sensitivity—defined as the degree to which a defect in a given region affects the \ac{MFD} measurements—decreases for deeper regions within the metal plate. The presence of a deeper defect in the imaging results can better fit the linear measurement model, especially considering the presence of noise and the intrinsic nonlinear relationship between the measurements and the physical property fields.
    \item Using multiple excitation frequencies, selected to match the eddy current penetration depth to the depth of the layers, ensures that the imaging results accurately reveal the starting depth of the defects. This approach leverages the physical property that high-frequency eddy currents penetrate only shallow regions and thus provide information exclusively from those layers.
    \item Visual observations suggest that \cref{alg:convex-optimization} performs slightly better than \cref{alg:mean-field-approx}, as its imaging results deviate less from the ground truth. This difference is likely because \cref{alg:mean-field-approx} relies on mean field approximation, which neglects correlations between the posterior probabilities of entries, leading to a biased solution, whereas \cref{alg:convex-optimization} has \cref{thm:tight-convex} that provides recovery guarantees. However, it is worth noting that \cref{alg:mean-field-approx} has lower computational complexity than \cref{alg:convex-optimization}.
\end{enumerate}

\subsection{Imaging defects in metal pipes}
\begin{figure}
    \centering
    \subcaptionbox{\label{fig:pipe-sensing-schematic}}{
        \begin{tikzpicture}
            \tikzstyle{indicator} = [color=yellow, -{Circle[length=3pt]}, thick]
            \tikzstyle{label} = [anchor=west, rounded corners, fill=yellow, fill opacity=0.4, text=black, text opacity=1, align=left]

            \node[inner sep=0] at (0,0) {\includegraphics[width=11cm]{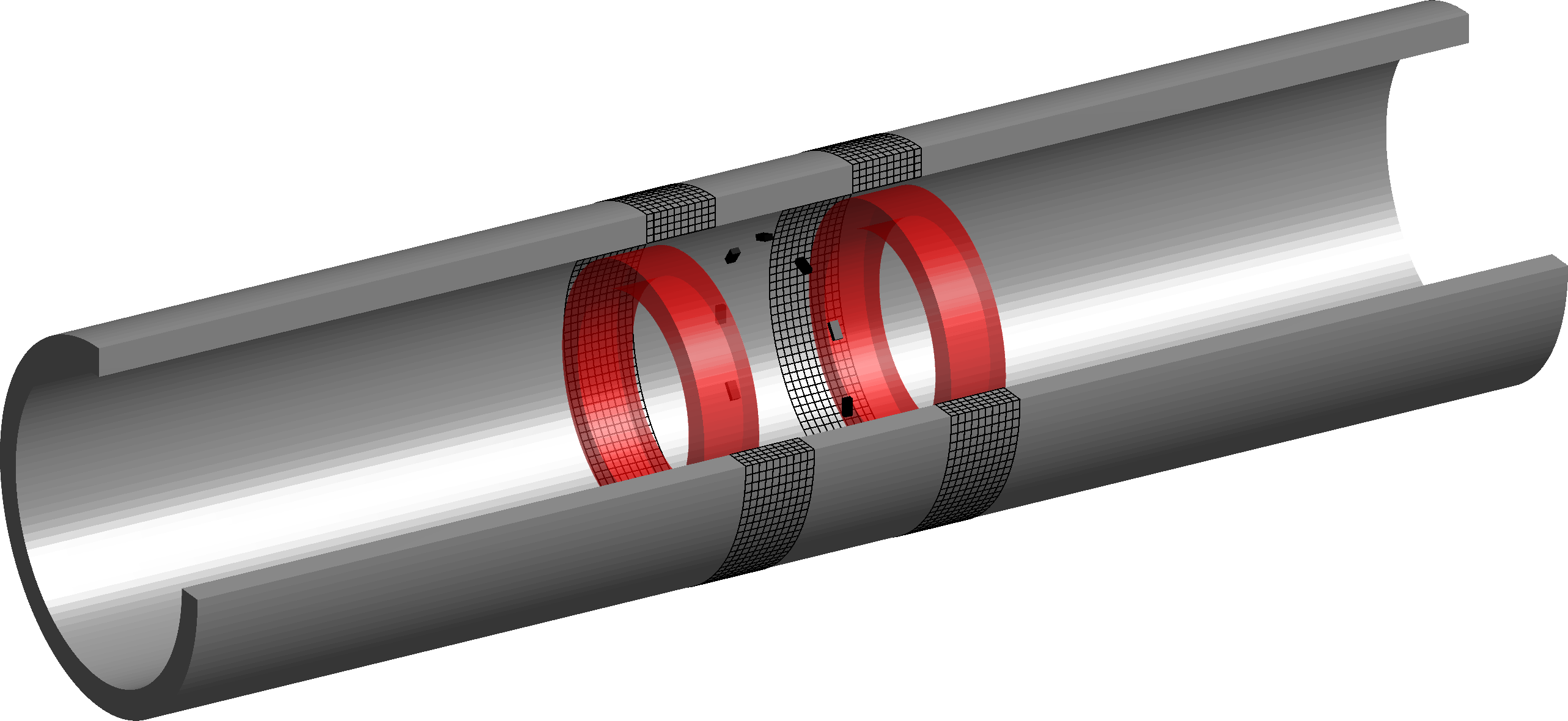}};

            \node[label] (label1) at (-4, 1.2) {Coils};
            \node[label] (label2) at (-4, 0.0) {Magnetic\\sensors};
            \node[label] (label3) at (-4,-1.2) {Metal pipe};

            \draw[indicator] (label1.east) -- (-0.6, 0.4);
            \draw[indicator] (label2.east) -- (0.37, 0.2);
            \draw[indicator] (label3.east) -- (-0.8,-1.1);
        \end{tikzpicture}
    } \\
    \subcaptionbox{\label{fig:pipe-sensing-simulation}}{
        \begin{tikzpicture}
            \node[inner sep=0] at (0,0) {\adjincludegraphics[width=15.17cm, trim={0 {0.31\height} 0 {0.31\height}}, clip]{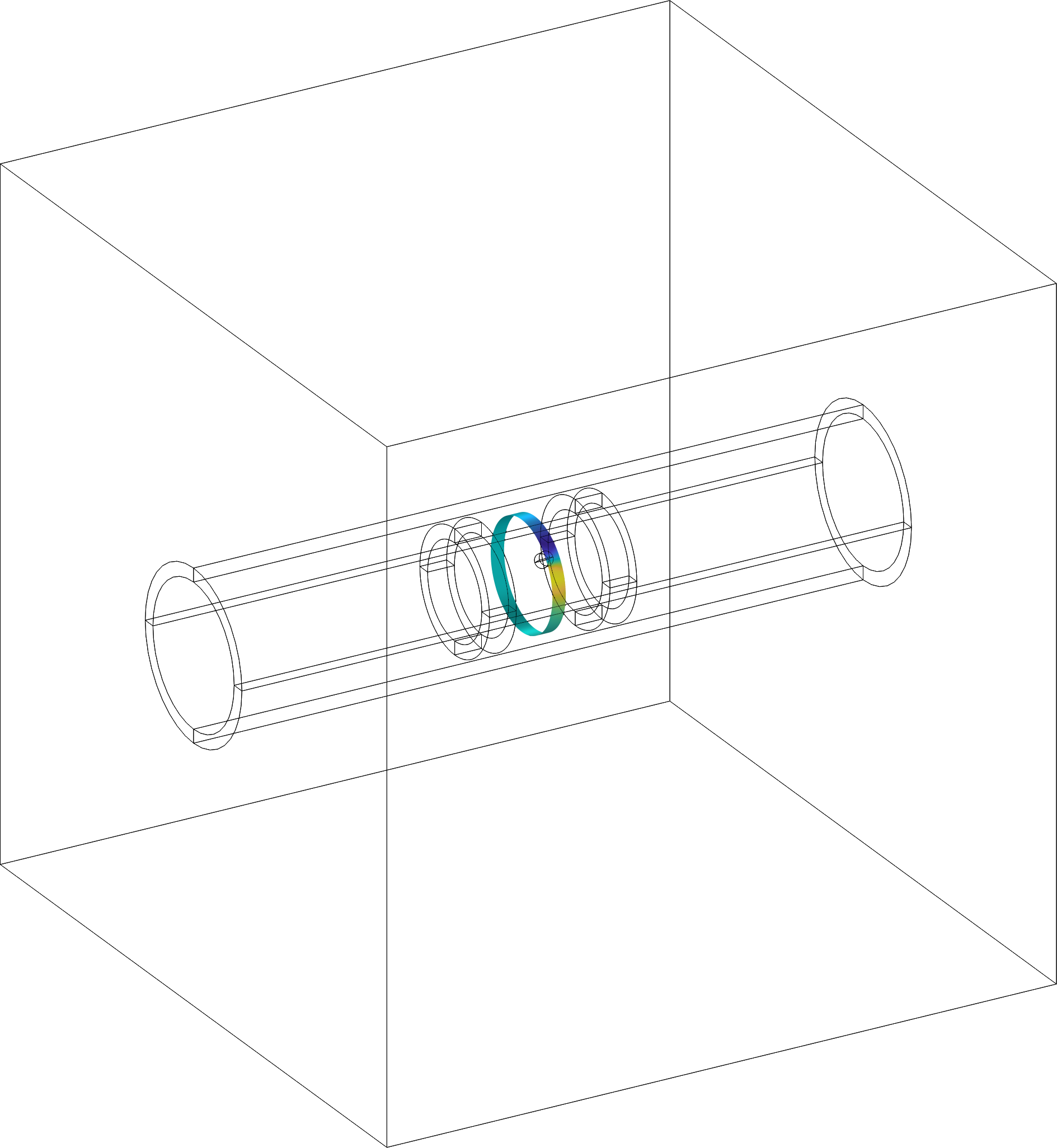}};
            \begin{scope}[shift={(6.8,0)}]
                \node[inner sep=0,draw,thick] at (0.03,0) {\includegraphics[height=3.0cm, width=1.5mm]{figures/colorbar.png}};
                \node[inner sep=0,anchor=east] at (0.6, 1.8) {\footnotesize ($\mu$T)};
                \node[inner sep=0,anchor=east] at (0.6, 1.5) {\footnotesize 2.2};
                \node[inner sep=0,anchor=east] at (0.6, 0.0) {\footnotesize 0\hspace*{2mm}};
                \node[inner sep=0,anchor=east] at (0.6,-1.5) {\footnotesize -2.2};
            \end{scope}
        \end{tikzpicture}
    }
    \caption{Eddy current sensing system for imaging defects in metal pipes.
             (a) Schematic of the bobbin probe and three-quarters of the pipe, along with the arrangement of voxels for reconstruction.
             (b) Finite element simulation setup and simulated \ac{MFD} component difference.}
\end{figure}

The seconds example concerns the imaging of defect in a metal pipe using a bobbin probe. This cylindrical probe, equipped with coils and magnetic sensors, moves through the interior of the pipe. At each position along the pipe, a computed tomography image of any nearby defects is obtained. As the probe advances, these individual images are merged, eventually producing a continuous defect map of the entire pipe length.

To merge the local images obtained at each position, binary Bayes filtering \cite{RN980} is used. In this method, each voxel covered by multiple local images is updated by summing the logit values from each recovery. Logit values are clamped to a large number to prevent the probability of defect presence from being exactly 0 or 1, which is a standard practice in Bayesian filtering to avoid irreversible probability states. As the bobbin probe moves along the pipe, each new local image provides updated evidence, which adjusts the voxel's probability of containing a defect based on both current observations and prior accumulated information. Since summation is commutative, the order in which each local image is merged does not affect the result.

Figure \ref{fig:pipe-sensing-schematic} shows the schematic of a bobbin probe, which comprises two coils carrying current in the same direction and a circular array of ten magnetic sensors with circumferential sensing axes. The bobbin probe is positioned inside the aluminum pipe. In this setup, regions with high sensitivity—where defects can significantly affect the measured \ac{MFD}—are discretized into voxels for imaging. The pipe has an inner diameter of 21~mm and an outer diameter of 25~mm. The full 2~mm wall thickness of the pipe is discretized into four layers, and excitation frequencies of 1000, 2000, 4000, and 16000~Hz are chosen to approximately match the eddy current penetration depth to the depth of each layer. The linearized sensitivities for these voxels across all frequencies are precomputed as described in \cref{thm:eddy-current-linearize} and \eqref{eq:measurement-perturb-discretize}.

In this study, the \ac{MFD} component measured by the magnetic sensors is simulated using commercial finite element software. The simulation setup, shown in \cref{fig:pipe-sensing-simulation}, includes the air regions both inside and outside the metal pipe. The simulated \ac{MFD} component in the circumferential direction is also illustrated in the figure, and the sensor signals equal to ten discrete samples taken along this strip. The simulations are conducted as the bobbin probe moves along the pipe in 4 mm increments. These spatially sparse samples are then processed using binary vector recovery algorithms, followed by refinement with binary Bayes filtering, to reconstruct the defect map of the entire pipe. The imaging results for a pipe with defects of 2~mm diameter at various depths are presented in \cref{tbl:pipe-result}.
\begin{table}
    \caption{Metal pipe imaging results.}
    \label{tbl:pipe-result}

    \setlength{\tabcolsep}{4pt}
    \def\w{45mm}

    \centering
    \begin{tabular}{|cc|c|c|c|}
        \thickhline
        && \multicolumn{3}{c|}{} \\[-9pt]
        \multicolumn{2}{|c|}{\multirow{-6.8}{*}{\rotatebox{90}{Ground truth}}} &
        \multicolumn{3}{c|}{
            \begin{tikzpicture}
                \node[inner sep=0] at (0,0) {\includegraphics[width=\w]{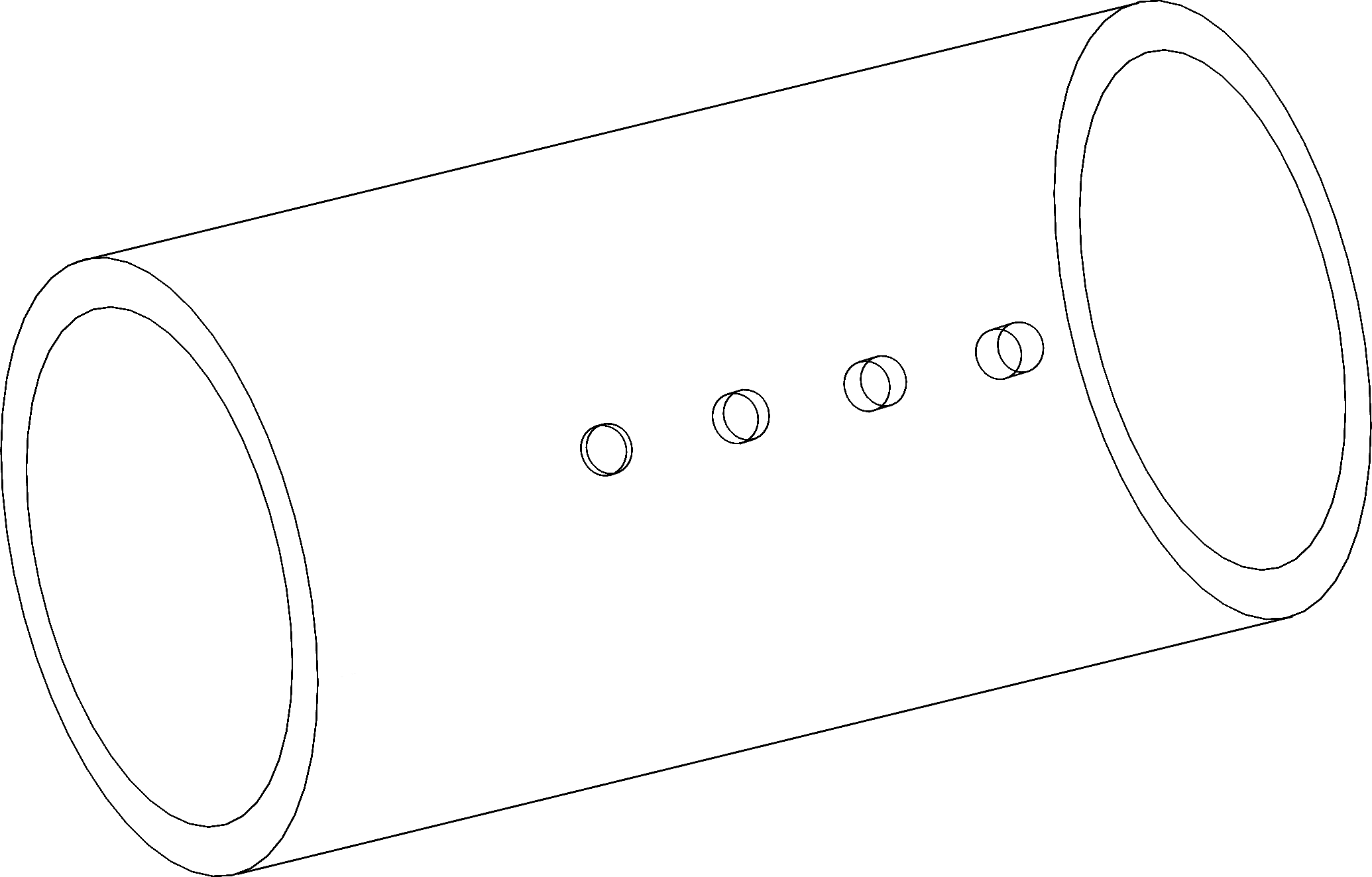}};
                \tikzstyle{label} = [rounded corners, fill=black, fill opacity=0.1, text opacity=1, inner sep=1.5pt]
                \node[label] at (0.1*\w,-0.2*\w) {\footnotesize 2~mm diameter defects};
            \end{tikzpicture}
        } \\
        \semithickhline
        & &
        \cref{alg:convex-optimization}    &
        \cref{alg:mean-field-approx}      &
        \cref{alg:approx-message-passing} \\
        \semithickhline
        \multirow{28.7}{*}{\hspace*{-1mm}\rotatebox{90}{Imaging result across radius}\hspace*{-2mm}} &
        {\footnotesize 10.5~mm} &&& \\[-9pt]
        & \multirow{1.6}{*}{\footnotesize 11.0~mm} &
        \includegraphics[width=\w]{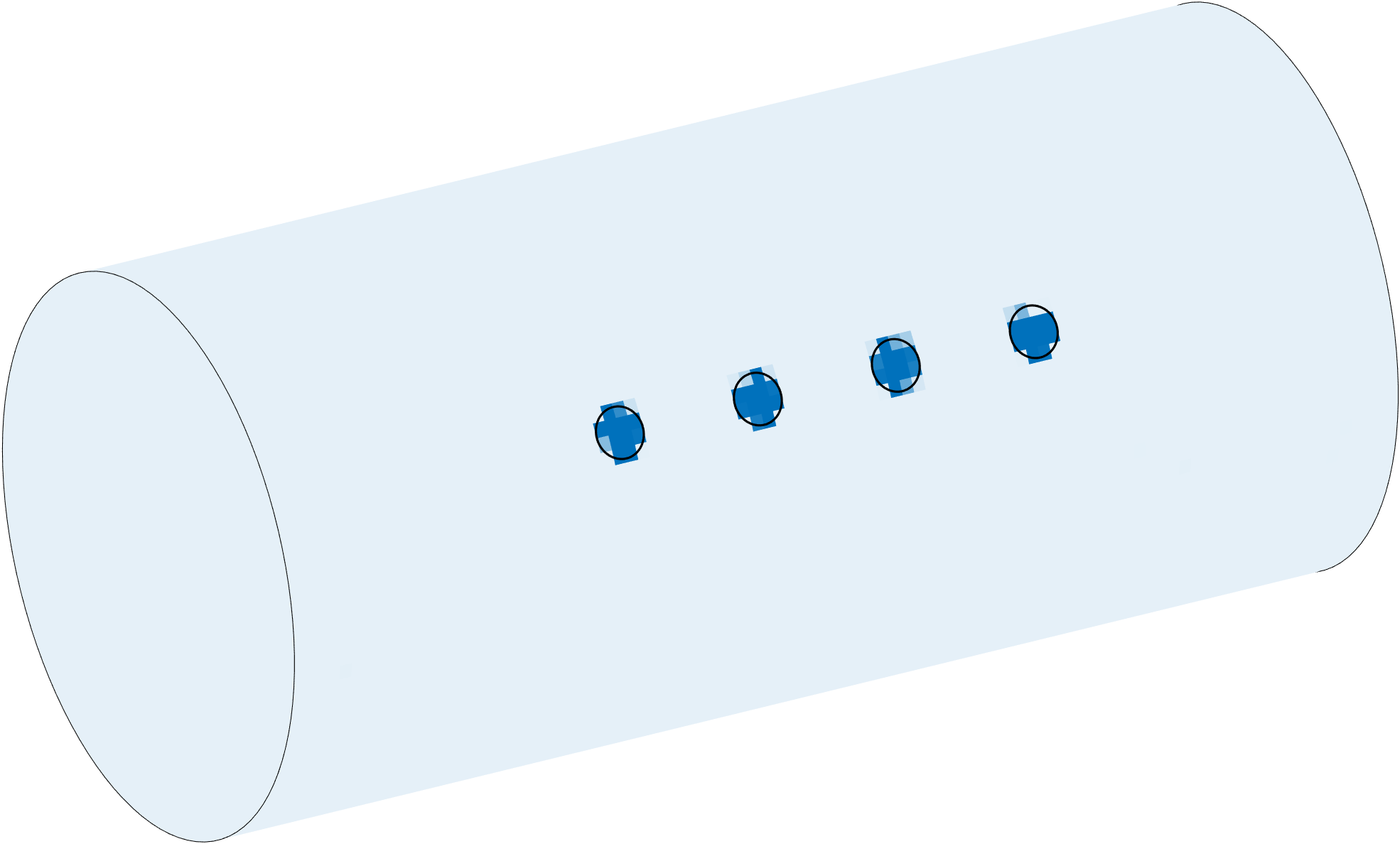} &
        \includegraphics[width=\w]{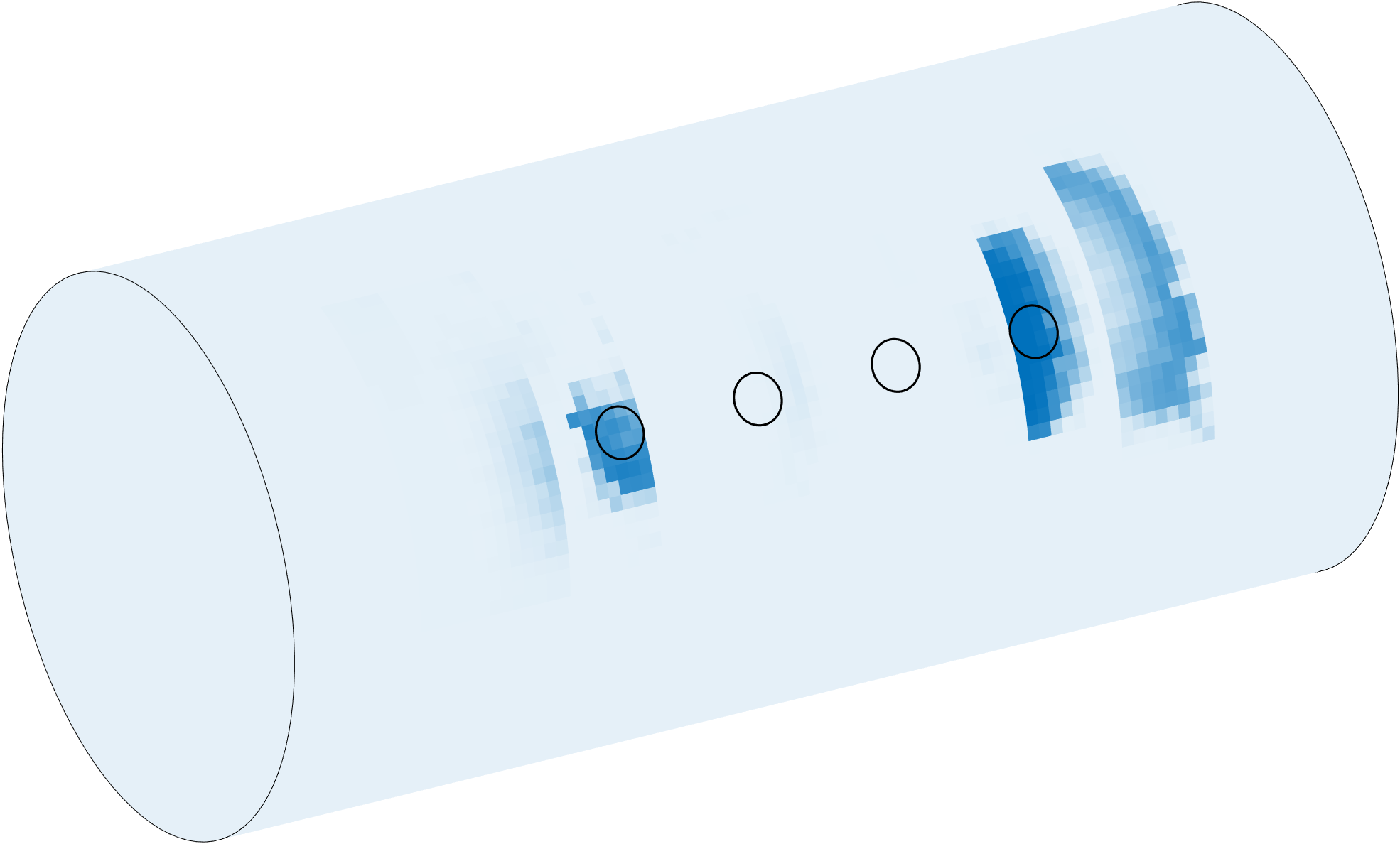}  &
        \includegraphics[width=\w]{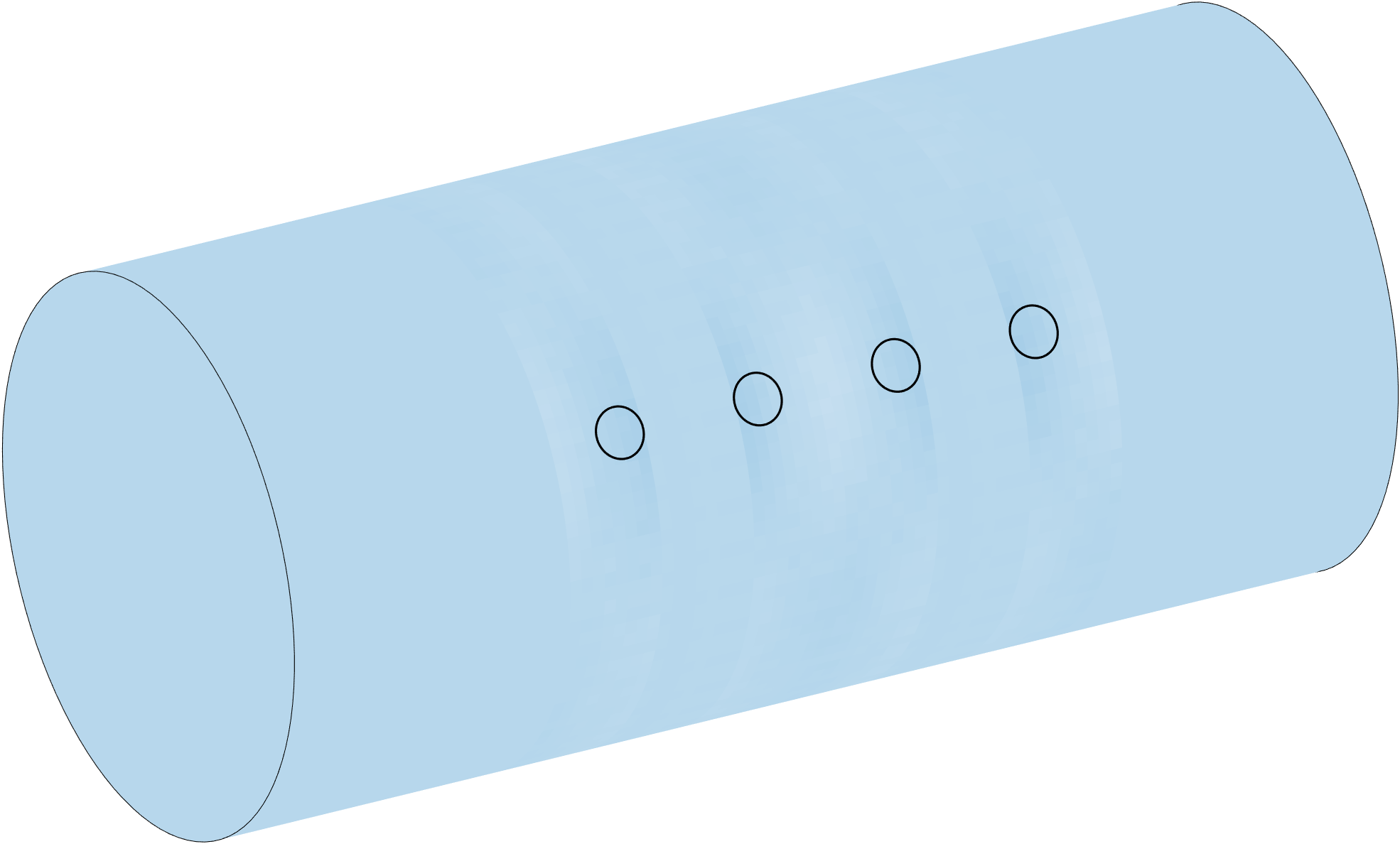} \\
        & \multirow{1.6}{*}{\footnotesize 11.5~mm} &
        \includegraphics[width=\w]{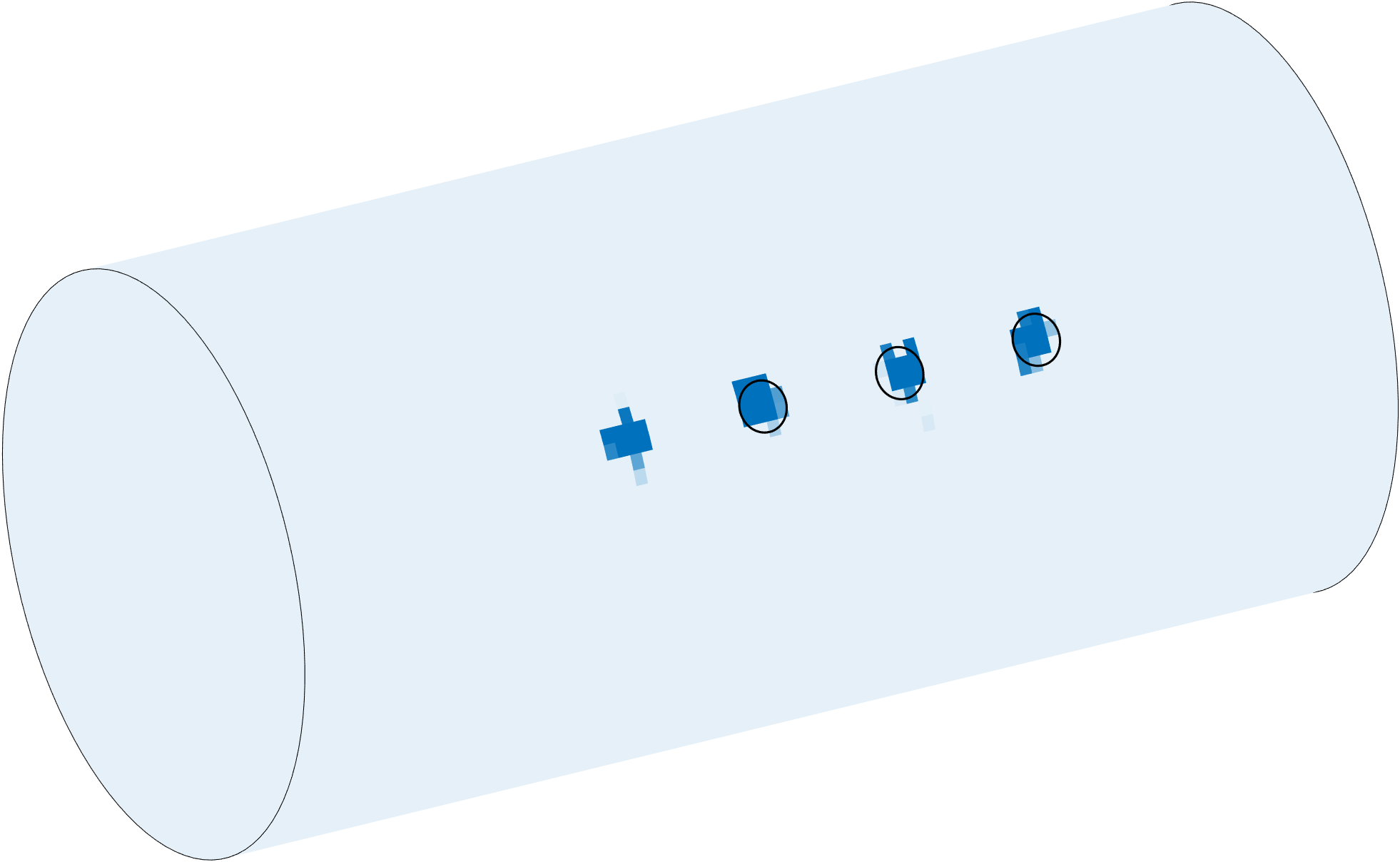} &
        \includegraphics[width=\w]{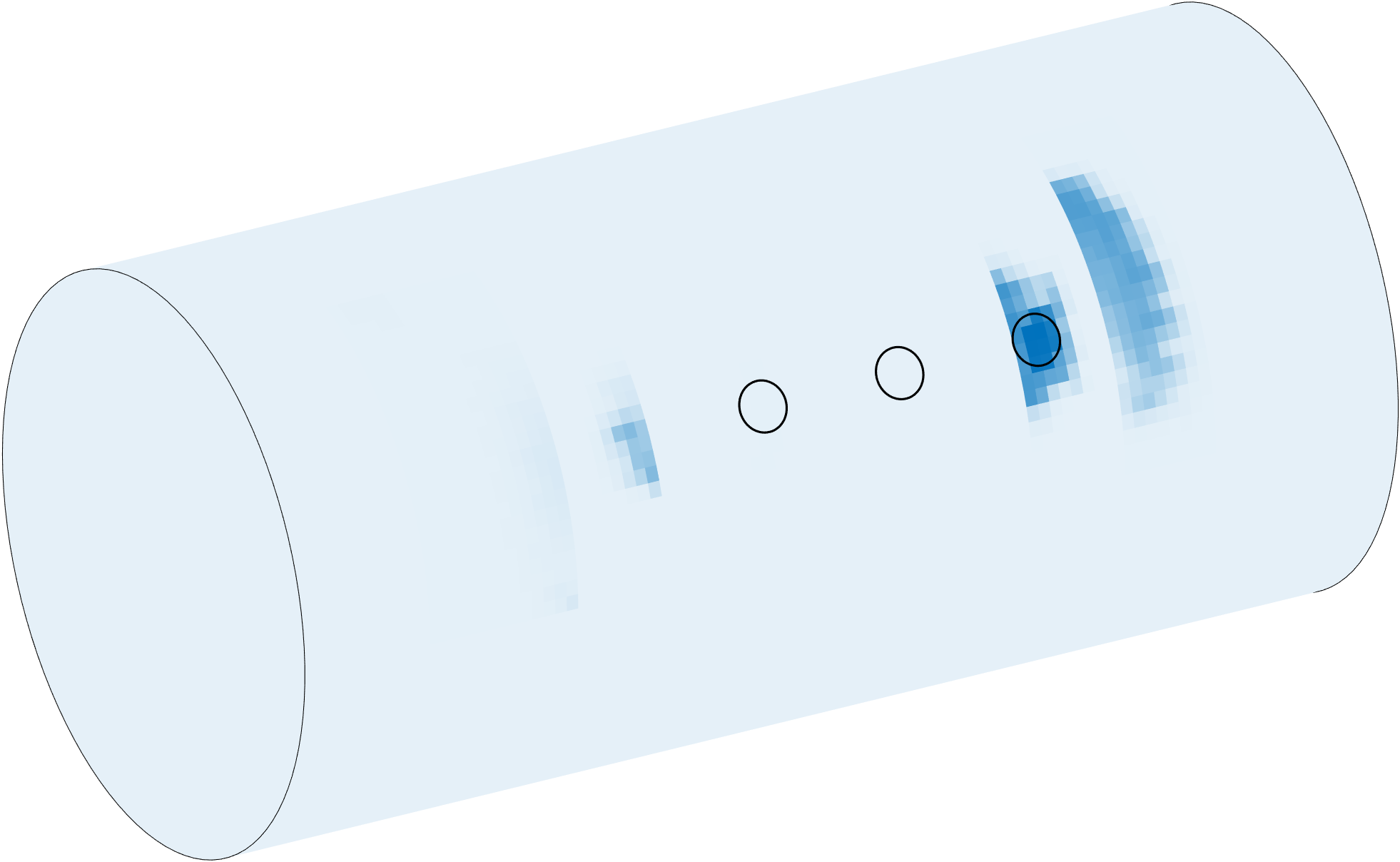}  &
        \includegraphics[width=\w]{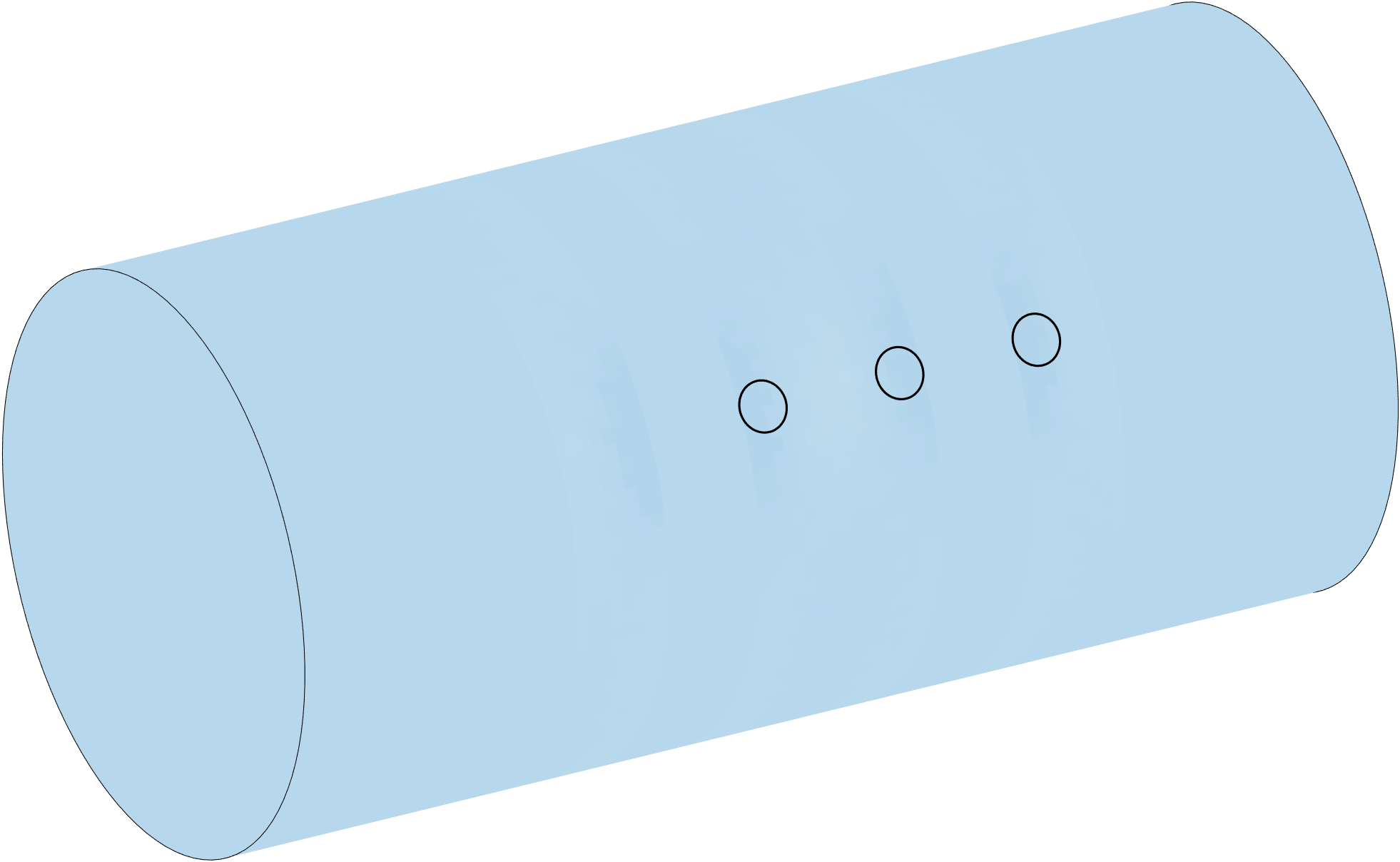} \\
        & \multirow{1.6}{*}{\footnotesize 12.0~mm} &
        \includegraphics[width=\w]{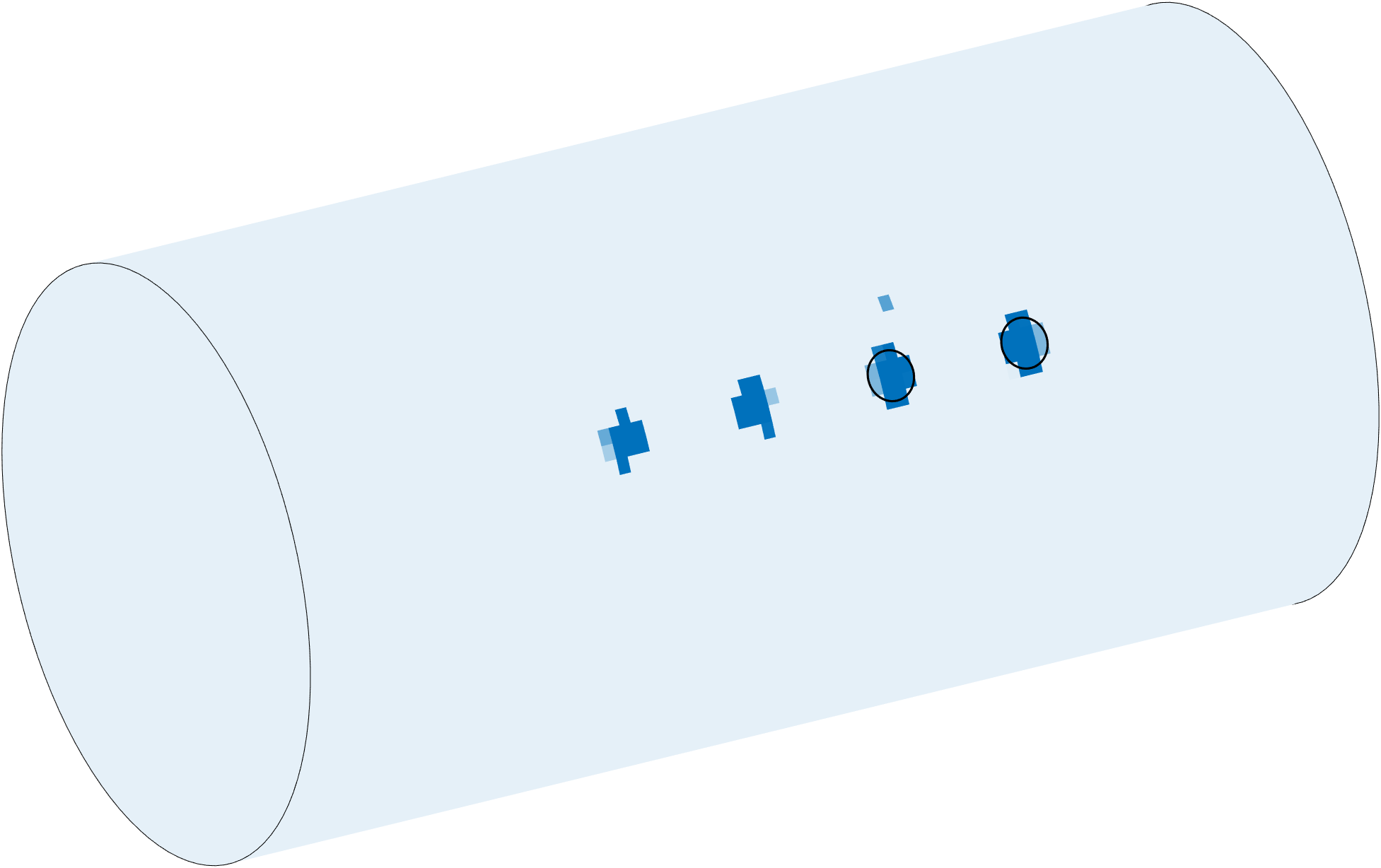} &
        \includegraphics[width=\w]{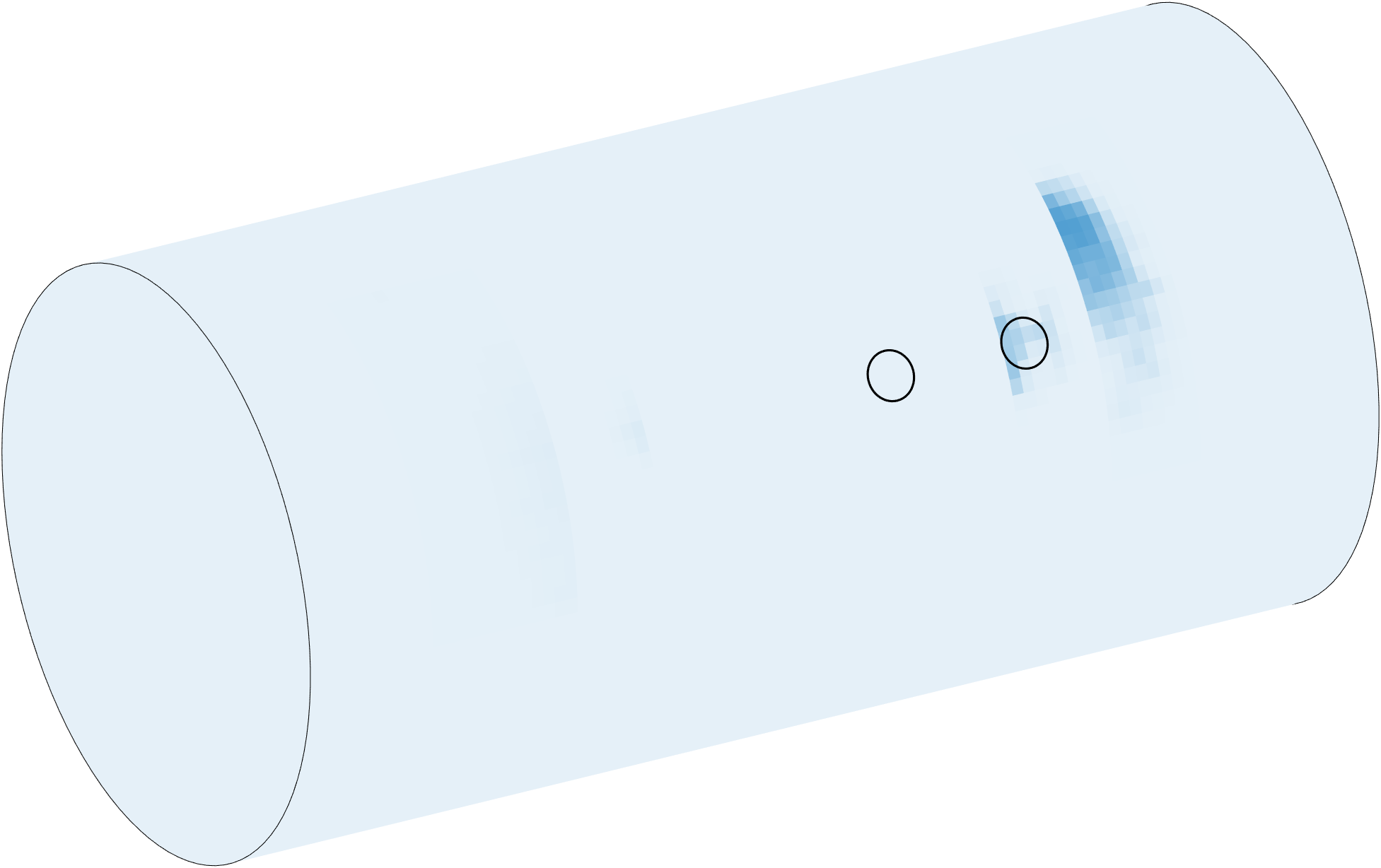}  &
        \includegraphics[width=\w]{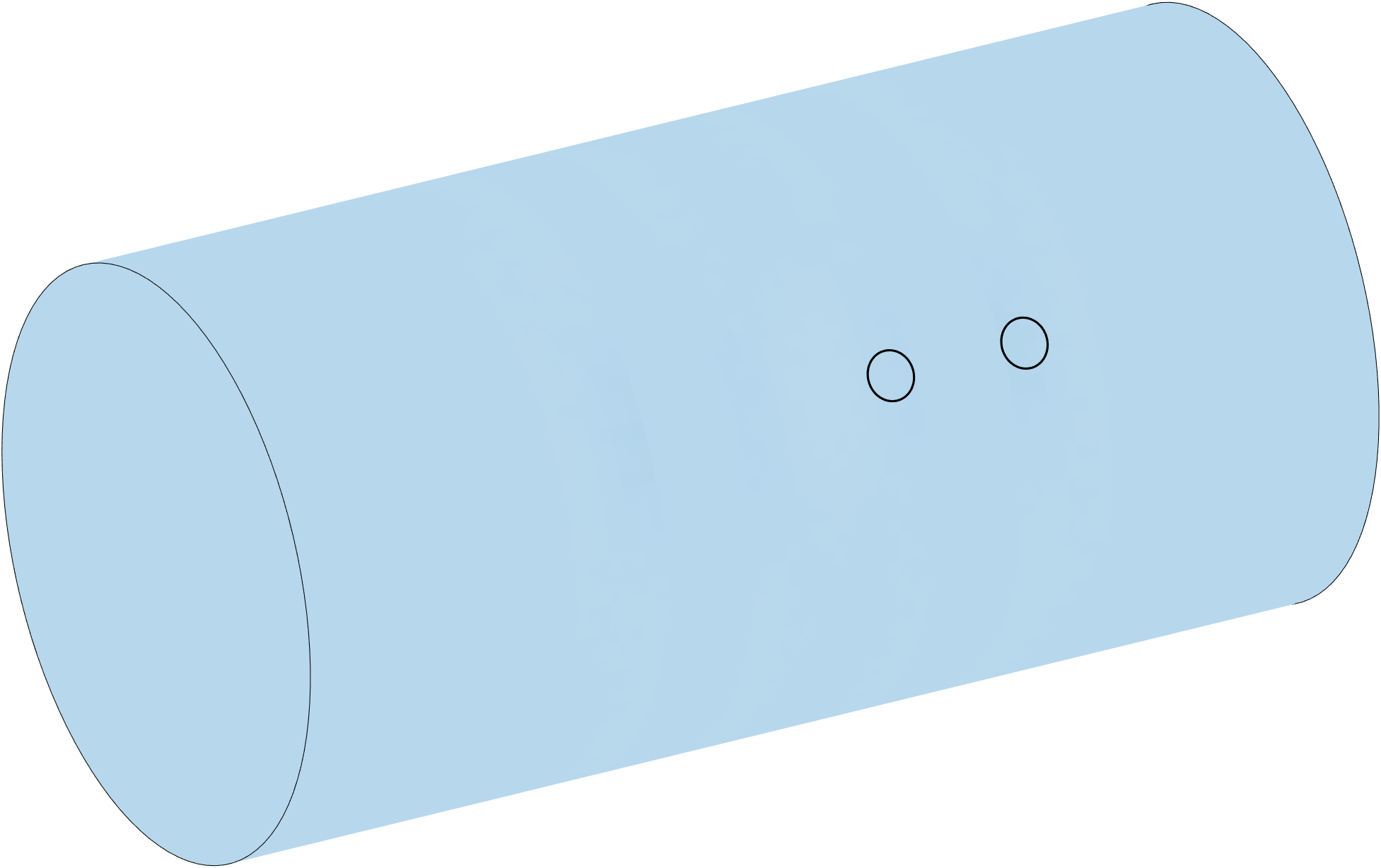} \\
        & \multirow{1.2}{*}{\footnotesize 12.5~mm} &
        \includegraphics[width=\w]{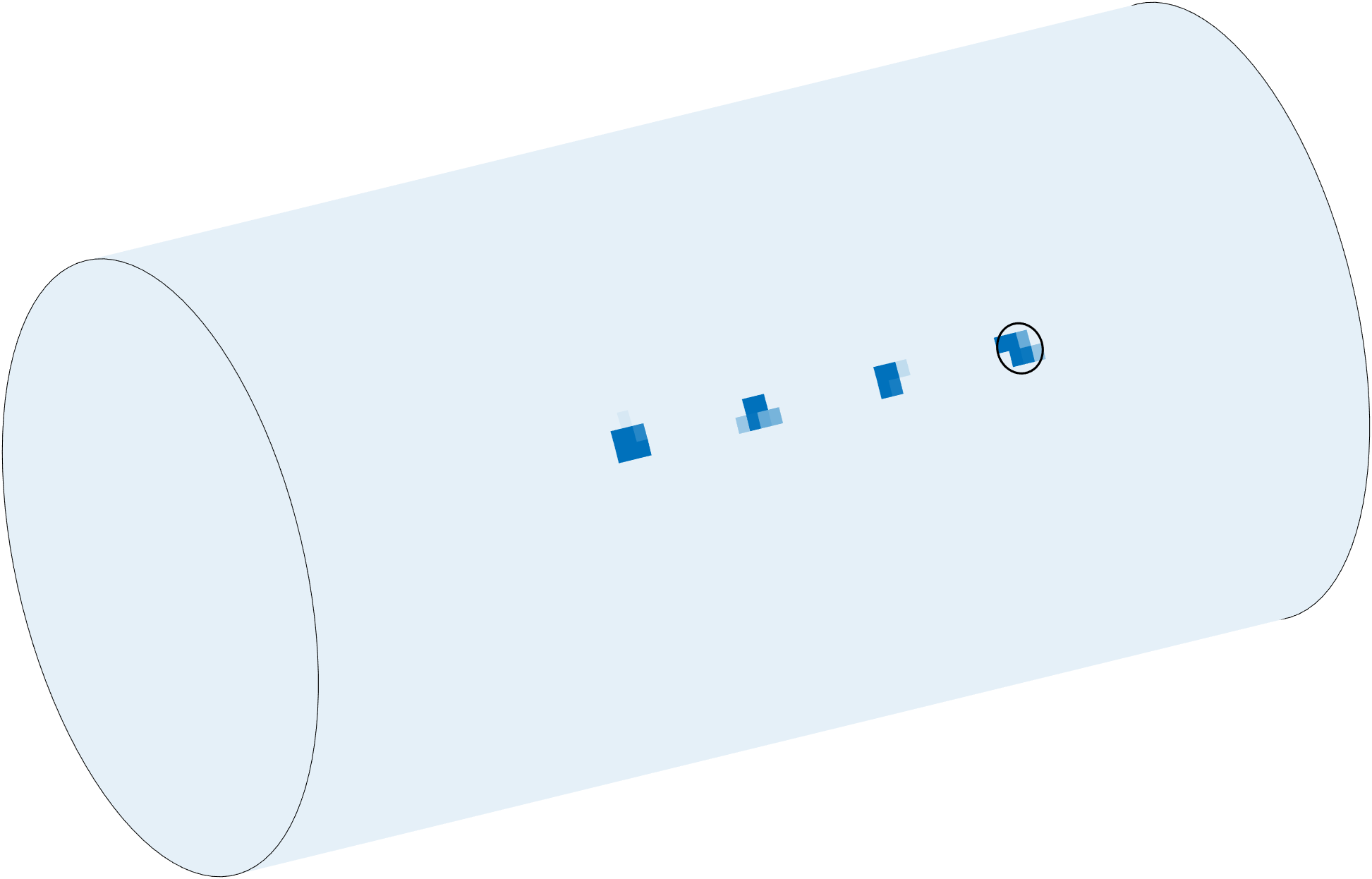} &
        \includegraphics[width=\w]{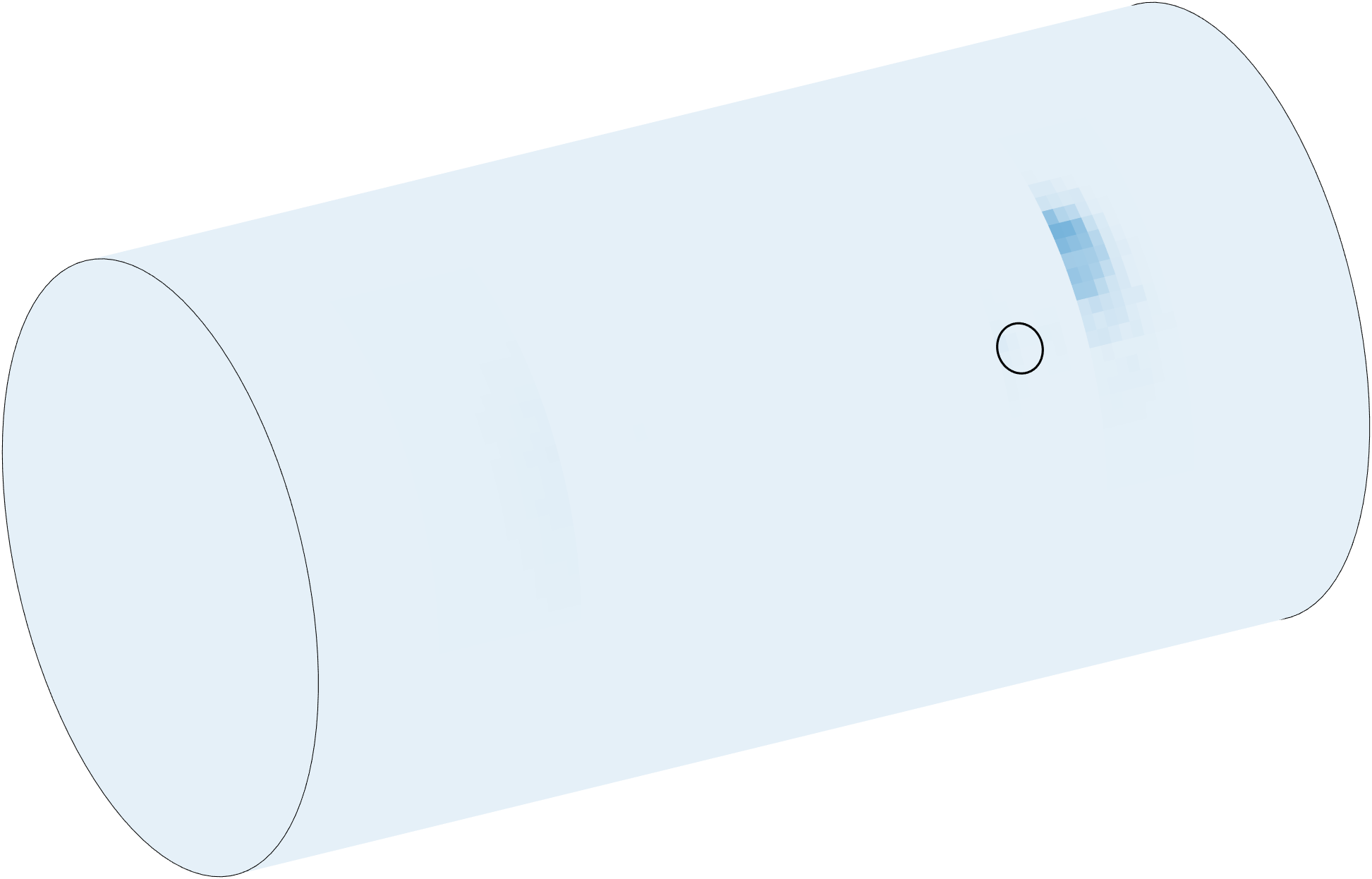}  &
        \includegraphics[width=\w]{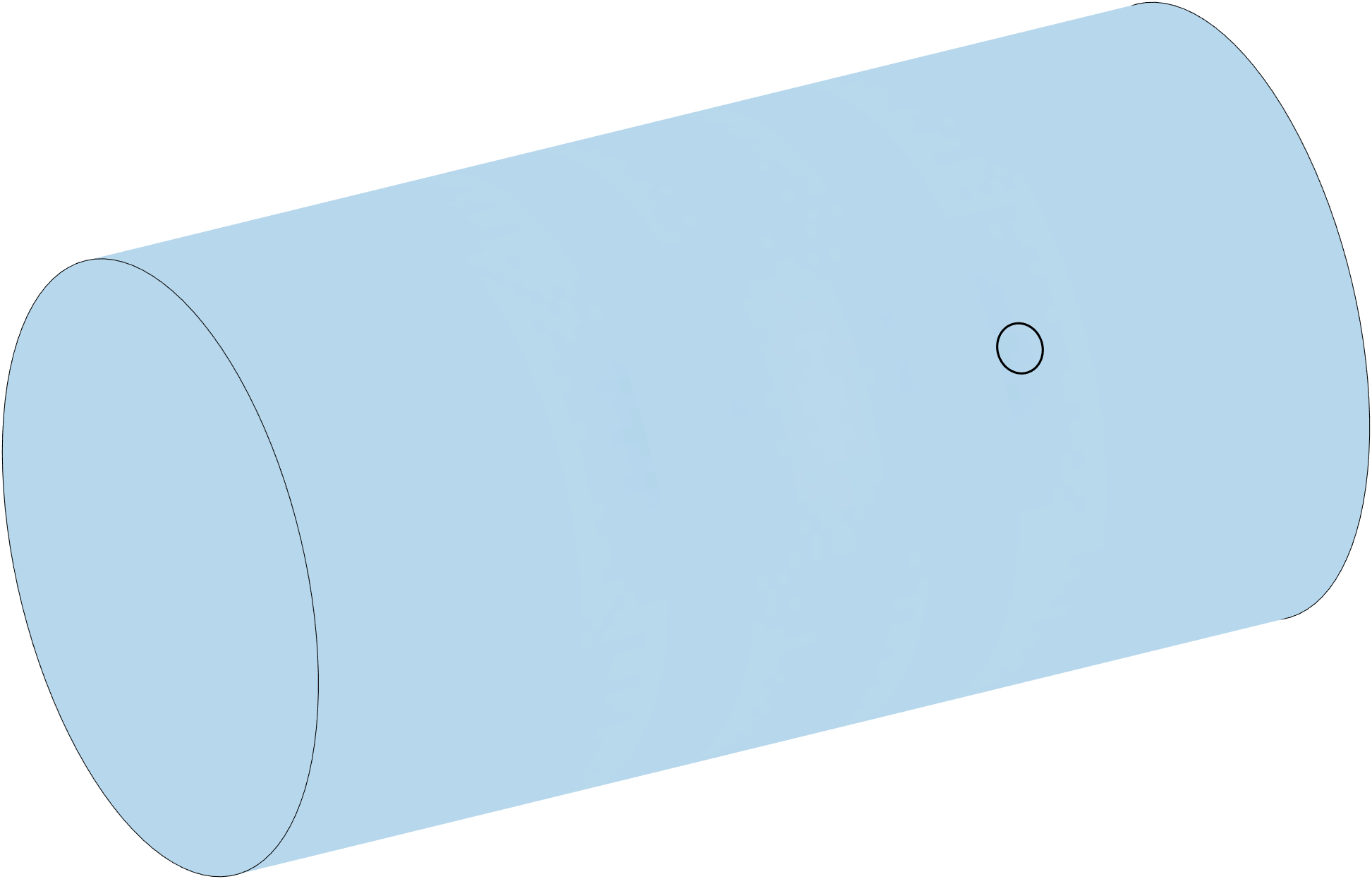} \\
        \thickhline
    \end{tabular}
\end{table}
From the results, several observations can be made:
\begin{enumerate}[nosep, leftmargin=*]
    \item In this example, sensor data from different axial or probe positions are combined using binary Bayes filtering, rather than assembling the data into a single vector and solving for the entire defect map simultaneously. The latter approach is similar to the method previously used for imaging metal plates. However, the binary Bayes filtering approach is computationally efficient, especially when the algorithm has a complexity higher than linear. By recovering individual images and merging them with binary Bayes filtering, only \cref{alg:convex-optimization} successfully identifies the defect locations. In contrast, the poor performance of \cref{alg:mean-field-approx} arises from its inability to construct images with sufficient resolution and accuracy from each local measurement.
    \item Similar to the results for metal plates, the result from \cref{alg:convex-optimization} here images defects as being deeper than their actual depth. This is primarily due to a mismatch between the model assumption of linear measurements and the actual nonlinearity of the physical measurement system.
\end{enumerate}

%% file: contents/section6.tex
\section{Conclusions}

This study addresses the problem of recovering binary vectors from linear measurements, an extension of compressive sampling that not only focuses on sparsity but also incorporates the additional constraint of binary values. Two approaches are presented to solve this problem: convex optimization and Bayesian inference.
The convex optimization approach is motivated by the fact that solving the optimization problem with the binary constraint is an NP-hard problem. Therefore, the binary constraint is relaxed to its convex hull. Theoretical results are provided, demonstrating that the convex relaxation is tight when the vector is sufficiently sparse and the measurement matrix satisfies the \ac{BRIP} condition. Additionally, the convex optimization approach can be extended to handle multiple measurement vectors with varying noise variances by augmenting the measurement vectors and matrices with proper scaling.
In the Bayesian approach, a Bayesian network is constructed to represent the linear measurements of a binary vector. It is shown that both \ac{MAP} inference and exact posterior inference on the Bayesian network are NP-hard problems. To address this, an approximate inference algorithm is developed that finds a mean field distribution close to the true posterior in the \ac{KL}-divergence measure. Another approximate inference algorithm is developed using message passing on a cluster graph.
The proposed algorithms are compared with existing compressive sampling methods and demonstrate superior performance in recovering binary vectors. These results suggest that binary recovery algorithms are preferable over sparse recovery algorithms when the binary constraint is applicable.

On the application side, this study demonstrates that the problem of imaging defects in metal structures using eddy current sensing can be formulated as a binary vector recovery problem. This involves analytically linearizing the relationship between the measured \ac{MFD} and the physical property fields. The binary vector recovery algorithms are then applied to perform the imaging, which are able to generate depthwise tomography images of defects in metal plates and pipes.

However, several limitations remain, particularly the challenge of scaling these algorithms to handle larger numbers of vector entries or imaging voxels. In this study, only the mean field approximate inference algorithm exhibits computational complexity that is linear with respect to the number of recovery entries, though it provides biased solutions. In contrast, the convex optimization method offers more accurate solutions but at the cost of higher computational complexity. Future research could focus on optimizing and combining these algorithms to maintain linear complexity while improving recovery performance. Additionally, exploring the application of binary recovery algorithms to other physical measurement systems would be a valuable direction for future work.

%% file: contents/acknowledgment.tex
\section*{Acknowledgement}
This work is supported by the National Science and Technology Council, Taiwan, under Grant NSTC \mbox{112-2221-E-002-155-MY2}.

%% file: contents/appendix1.tex
\section{Proof of \cref{thm:tight-convex}}
\label{app:tight-convex-proof}

Before proofing \cref{thm:tight-convex}, the following lemma is required.
\begin{lemma}  \label{thm:RIP-special-property}
    Suppose $\vec{x}$ and $\vec{z}$ are vectors with disjoint support, $|\support{\vec{x}}| + |\support{\vec{z}}| \leq k$, and $\|\vec{x}\|_\infty \leq 1$, $\|\vec{z}\|_\infty \leq 1$. Then,
    \begin{equation*}
        |\langle \mat{\Phi} \vec{x}, \mat{\Phi} \vec{z} \rangle| \leq \delta_s^\mathrm{b} \, \|\vec{x}\|_2 \, \|\vec{z}\|_2
    \end{equation*}
\end{lemma}
\begin{proof}
    Let $\tilde{\vec{x}} = \vec{x} / \|\vec{x}\|_2$ and $\tilde{\vec{z}} = \vec{z} / \|\vec{z}\|_2$.
    Since $\tilde{\vec{x}}$ and $\tilde{\vec{z}}$ have disjoint support, $\|\tilde{\vec{x}} + \tilde{\vec{z}}\|_2^2 = \|\tilde{\vec{x}} - \tilde{\vec{z}}\|_2^2 = 2$. Thus,
    \begin{equation*}
    \begin{split}
        |\langle \mat{\Phi} \tilde{\vec{x}}, \mat{\Phi} \tilde{\vec{z}} \rangle|
        &=    \frac{1}{4} \bigl| \|\mat{\Phi} \tilde{\vec{x}} + \mat{\Phi} \tilde{\vec{z}}\|_2^2 - \|\mat{\Phi} \tilde{\vec{x}} - \mat{\Phi} \tilde{\vec{z}}\|_2^2 \bigr|  \\
        &\leq \frac{1}{4} \bigl| (1 + \delta_s^\mathrm{b}) \, \|\tilde{\vec{x}} + \tilde{\vec{z}}\|_2^2 - (1 - \delta_s^\mathrm{b}) \, \|\tilde{\vec{x}} - \tilde{\vec{z}}\|_2^2 \bigr|  \\
        &= \delta_s^\mathrm{b} .
    \end{split}
    \end{equation*}
    Multiplying both sides by $\|\vec{x}\|_2 \, \|\vec{z}\|_2$ gives the result.
\end{proof}

The following is the proof for \cref{thm:tight-convex}.
\begin{proof}
    Set \mbox{$\vec{h} = \vec{x}^\star - \vec{x}_\text{true}$}. Given that $\vec{x}_\text{true}$ is binary and \mbox{$\vec{0} \preceq \vec{x}^\star \preceq \vec{1}$}, it follows that \mbox{$\|\vec{h}\|_\infty \leq 1$}.
    Define the disjoint subsets $\set{J}_0, \set{J}_1, \set{J}_2, \set{J}_3, \dots$ as follows:
    \begin{itemize}[nosep, leftmargin=*, label={}]
        \item $\set{J}_0 = \support{\vec{x}_\text{true}}$,
        \item $\set{J}_1$ indices the $s$ entries with the largest absolute value of $\vec{h}_{\set{J}_0^\cpm}$,
        \item $\set{J}_2$ indices the $s$ entries with the largest absolute value of $\vec{h}_{(\set{J}_0 \cup \set{J}_1)^\cpm}$,
        \item $\set{J}_3$ indices the $s$ entries with the largest absolute value of $\vec{h}_{(\set{J}_0 \cup \set{J}_1 \cup \set{J}_2)^\cpm}$,
        \item \vspace*{-2mm}\hspace*{4cm}$\vdots$
    \end{itemize}
    The proof is established by bounding $\|\vec{h}_{\set{J}_0 \cup \set{J}_1}\|_2$ and $\|\vec{h}_{(\set{J}_0 \cup \set{J}_1)^\cpm}\|_2$, and then summing these bounds.
    Since \mbox{$\vec{h}_{\set{J}_0 \cup \set{J}_1} = \vec{h}_{\set{J}_0} + \vec{h}_{\set{J}_1} = \vec{h} - \sum_{\ell\geq 2} \vec{h}_{\set{J}_\ell}$}, if follows that
    \begin{equation}  \label{eq:tight-convex-proof-1}
    \begin{split}
        (1 - \delta_{2s}^\mathrm{b}) \, \|\vec{h}_{\set{J}_0 \cup \set{J}_1}\|_2^2
        &\leq \|\mat{\Phi} \vec{h}_{\set{J}_0 \cup \set{J}_1}\|_2^2  \\
        &=    \bigl\langle \mat{\Phi} \vec{h}_{\set{J}_0 \cup \set{J}_1}, \mat{\Phi} \vec{h} \bigr\rangle + \bigl\langle \mat{\Phi} \vec{h}_{\set{J}_0} + \mat{\Phi} \vec{h}_{\set{J}_1}, -\sum_{\ell\geq 2} \mat{\Phi} \vec{h}_{\set{J}_\ell} \bigr\rangle  \\
        &\leq \|\mat{\Phi} \vec{h}_{\set{J}_0 \cup \set{J}_1}\|_2 \, \|\mat{\Phi} \vec{h}\|_2 + \sum_{\ell \geq 2} \bigl( | \langle \mat{\Phi} \vec{h}_{\set{J}_0}, \mat{\Phi} \vec{h}_{\set{J}_\ell} \rangle | + | \langle \mat{\Phi} \vec{h}_{\set{J}_1}, \mat{\Phi} \vec{h}_{\set{J}_\ell} \rangle | \bigr) .
    \end{split}
    \end{equation}
    The first term on the right-hand side is bounded by
    \begin{equation}  \label{eq:tight-convex-proof-2}
    \begin{split}
        \|\mat{\Phi} \vec{h}_{\set{J}_0 \cup \set{J}_1}\|_2 \, \|\mat{\Phi} \vec{h}\|_2
        &\leq \sqrt{1 + \delta_{2s}^\text{b}} \, \|\vec{h}_{\set{J}_0 \cup \set{J}_1}\|_2 \, \|\mat{\Phi} \vec{h}\|_2  \\
        &\overset{2}{\leq} 2 \epsilon \sqrt{1 + \delta_{2s}^\text{b}} \, \|\vec{h}_{\set{J}_0 \cup \set{J}_1}\|_2 ,
    \end{split}
    \end{equation}
    where inequality 2 holds because $\|\mat{\Phi} \vec{h}\|_2 = \|\mat{\Phi} (\vec{x}^\star - \vec{x}_\text{true})\|_2 \leq \|\mat{\Phi} \vec{x}^\star - \vec{y}\|_2 + \|\vec{y} - \mat{\Phi} \vec{x}_\text{true}\|_2 \leq 2\epsilon$.
    The second term on the right-hand side is bounded by
    \begin{equation}  \label{eq:tight-convex-proof-3}
    \begin{split}
        \sum_{\ell \geq 2}\bigl( |\langle\mat{\Phi}\vec{h}_{\set{J}_0},\mat{\Phi}\vec{h}_{\set{J}_\ell}\rangle| + |\langle\mat{\Phi}\vec{h}_{\set{J}_1},\mat{\Phi}\vec{h}_{\set{J}_\ell}\rangle| \bigr)
        &\overset{1}{\leq} \delta_{2s}^\text{b} \left(\|\vec{h}_{\set{J}_0}\|_2 + \|\vec{h}_{\set{J}_1}\|_2\right) \sum_{\ell \geq 2} \|\vec{h}_{\set{J}_\ell}\|_2 \\
        &\overset{2}{\leq} \delta_{2s}^\text{b} \sqrt{2} \|\vec{h}_{\set{J}_0\cup\set{J}_1}\|_2 \sum_{\ell \geq 2} \|\vec{h}_{\set{J}_\ell}\|_2 ,
    \end{split}
    \end{equation}
    where
    \begin{itemize}[nosep, leftmargin=*, label={}]
        \item inequality 1 uses \cref{thm:RIP-special-property},
        \item inequality 2 holds because $\bigl(\|\vec{h}_{\set{J}_0}\|_2 + \|\vec{h}_{\set{J}_1}\|_2\bigr)^2 \leq 2 \|\vec{h}_{\set{J}_0}\|_2^2 + 2 \|\vec{h}_{\set{J}_1}\|_2^2 = 2 \|\vec{h}_{\set{J}_0\cup\set{J}_1}\|_2^2$.
    \end{itemize}
    Furthermore, $\sum_{\ell \geq 2} \|\vec{h}_{\set{J}_\ell}\|_2$ is bounded as follows:
    \begin{equation}  \label{eq:tight-convex-proof-4}
    \begin{split}
        \sum_{\ell \geq 2} \|\vec{h}_{\set{J}_\ell}\|_2
        &\overset{1}{\leq} \sum_{\ell \geq 2} \sqrt{s} \, \|\vec{h}_{\set{J}_\ell}\|_\infty  \\
        &\overset{2}{\leq} \sum_{\ell \geq 2} \frac{1}{\sqrt{s}} \, \|\vec{h}_{\set{J}_{\ell-1}}\|_1
         =                 \frac{1}{\sqrt{s}} \|\vec{h}_{\set{J}_0^\cpm}\|_1  \\
        &\overset{3}{\leq} \frac{1}{\sqrt{s}} \|\vec{h}_{\set{J}_0}\|_1  \\
        &\overset{4}{\leq} \|\vec{h}_{\set{J}_0}\|_2
         \leq              \|\vec{h}_{\set{J}_0\cup\set{J}_1}\|_2 ,
    \end{split}
    \end{equation}
    where
    \begin{itemize}[nosep, leftmargin=*, label={}]
        \item inequality 1 holds because $\|\vec{z}\|_2 \leq \sqrt{s} \|\vec{z}\|_\infty$ for any $s$-sparse $\vec{z}$,
        \item inequality 2 holds because $\|\vec{h}_{\set{J}_\ell}\|_\infty \geq \|\vec{h}_{\set{J}_{\ell-1}}\|_1 / s$,
        \item inequality 3 holds because $\|\vec{x}_\text{true}\|_1 \geq \|\vec{x}^\star\|_1 = \|\vec{x}_\text{true} + \vec{h}\|_1 = \|\vec{x}_\text{true} + \vec{h}_{\set{J}_0}\|_1 + \|\vec{h}_{\set{J}_0^\cpm}\|_1 \geq \|\vec{x}_\text{true}\|_1 - \|\vec{h}_{\set{J}_0}\|_1 + \|\vec{h}_{\set{J}_0^\cpm}\|_1$,
        \item inequality 4 holds because $\|\vec{z}\|_1 \leq \sqrt{s} \|\vec{z}\|_2$ for any $s$-sparse $\vec{z}$.
    \end{itemize}
    Substituting \eqref{eq:tight-convex-proof-2}, \eqref{eq:tight-convex-proof-3}, and \eqref{eq:tight-convex-proof-4} into \eqref{eq:tight-convex-proof-1} and rearranging yields
    \begin{equation}  \label{eq:tight-convex-proof-5}
        \bigl( 1 - (\sqrt{2}+1) \, \delta_{2s}^\text{b} \bigr) \, \|\vec{h}_{\set{J}_0\cup\set{J}_1}\|_2
        \leq 2\epsilon \sqrt{1 + \delta_{2s}^\text{b}} .
    \end{equation}
    The bound for $\|\vec{h}_{(\set{J}_0 \cup \set{J}_1)^\cpm}\|_2$ is simpler,
    \begin{equation}  \label{eq:tight-convex-proof-6}
        \|\vec{h}_{(\set{J}_0 \cup \set{J}_1)^\cpm}\|_2
        =                 \| \sum_{\ell\geq 2} \vec{h}_{\set{J}_\ell} \|_2
        \leq              \sum_{\ell\geq 2} \|\vec{h}_{\set{J}_\ell}\|_2
        \overset{2}{\leq} \|\vec{h}_{\set{J}_0\cup\set{J}_1}\|_2 ,
    \end{equation}
    where inequality 2 directly comes from \eqref{eq:tight-convex-proof-4}.
    Combining \eqref{eq:tight-convex-proof-5} and \eqref{eq:tight-convex-proof-6} leads to the following result
    \begin{equation}
        \|\vec{h}\|_2
        \leq \|\vec{h}_{\set{J}_0\cup\set{J}_1}\|_2 + \|\vec{h}_{(\set{J}_0 \cup \set{J}_1)^\cpm}\|_2
        \leq \frac{4 \sqrt{1 + \delta_{2s}^\text{b}}}{1 - (\sqrt{2}+1) \, \delta_{2s}^\text{b}} \epsilon .
    \end{equation}
    This inequality holds as long as $1 - (\sqrt{2}+1) \, \delta_{2s}^\text{b} > 0$, which implies $\delta_{2s}^\text{b} < \sqrt{2}-1$.
\end{proof}

%% file: contents/appendix2.tex
\section{Linearization of the eddy current system}
\label{app:linearization-proof}

The governing equations for the eddy current system, previously given in \eqref{eq:electromagnetic-system1}, are repeated here for convenience.
\begin{subequations}  \label{eq:electromagnetic-system}
\begin{align}
    \nabla \times \vec{H} &= \sigma \vec{E} + \vec{J}_s,  \\
    \nabla \times \vec{E} &= -j \omega \mu \vec{H} .
\end{align}
\end{subequations}
Perturbations in the electrical conductivity field, $\delta \sigma$, and in the magnetic permeability field, $\delta \mu$, induce corresponding perturbations in the electric field, $\delta \vec{E}$, and in the magnetic field intensity, $\delta \vec{H}$. The perturbed eddy current system is then governed by
\begin{subequations}  \label{eq:electromagnetic-system-perturbed}
\begin{align}
    \nabla \times (\vec{H} + \delta \vec{H}) &= (\sigma + \delta \sigma)(\vec{E} + \delta \vec{E}) + \vec{J}_s,  \\
    \nabla \times (\vec{E} + \delta \vec{E}) &= -j \omega (\mu + \delta \mu)(\vec{H} + \delta \vec{H}) .
\end{align}
\end{subequations}
Next, consider an alternative electromagnetic system, governed by
\begin{subequations}  \label{eq:electromagnetic-system-alternate}
\begin{align}
    \nabla \times \acute{\vec{H}} &= \sigma \acute{\vec{E}} + \acute{\vec{J}}_s ,  \\
    \nabla \times \acute{\vec{E}} &= -j \omega \mu \acute{\vec{H}} - \acute{\vec{L}}_s ,
\end{align}
\end{subequations}
where $\acute{\vec{L}}_s$ represents a source magnetic current density, and the accent symbol $\ \acute{}\ $ indicates that these equations describe a distinct system from \eqref{eq:electromagnetic-system}.

By applying the divergence theorem and the vector calculus identity \mbox{$\nabla \cdot (\vec{a} \times \vec{b}) = (\nabla \times \vec{a}) \cdot \vec{b} - \vec{a} \cdot (\nabla \times \vec{b})$}, \eqref{eq:electromagnetic-system} and \eqref{eq:electromagnetic-system-alternate} are combined to yield
\begin{equation}  \label{eq:perturbation-analysis1}
\begin{split}
       \oint_{\partial V} ( \vec{E} \times \acute{\vec{H}} ) \cdot \hat{\vec{n}} \, dS
    &= \int_V \nabla \cdot ( \vec{E} \times \acute{\vec{H}} ) \, dV  \\
    &= \int_V \bigl( ( \nabla \times \vec{E}) \cdot \acute{\vec{H}} - \vec{E} \cdot (\nabla \times \acute{\vec{H}}) \bigr) \, dV  \\
    &= \int_V \bigl( -j \omega \mu \vec{H} \cdot \acute{\vec{H}} - \vec{E} \cdot (\sigma \acute{\vec{E}} + \acute{\vec{J}}_s) \bigr) \, dV .
\end{split}
\end{equation}
Similarly, \eqref{eq:electromagnetic-system-perturbed} and \eqref{eq:electromagnetic-system-alternate} are combined to yield
\begin{equation} \label{eq:perturbation-analysis2}
\begin{split}
     \oint_{\partial V} ( (\vec{E} + \delta\vec{E}) \times \acute{\vec{H}} ) \cdot \hat{\vec{n}} \, dS
  &= \int_V \nabla \cdot ( (\vec{E} + \delta\vec{E}) \times \acute{\vec{H}} ) \, dV  \\
  &= \int_V \bigl( ( \nabla \times (\vec{E} + \delta\vec{E}) ) \cdot \acute{\vec{H}} - (\vec{E} + \delta\vec{E}) \cdot (\nabla \times \acute{\vec{H}}) \bigr) \, dV  \\
  &= \int_V \bigl( -j \omega (\mu + \delta\mu) (\vec{H} + \delta\vec{H}) \cdot \acute{\vec{H}} - (\vec{E} + \delta\vec{E}) \cdot (\sigma \acute{\vec{E}} + \acute{\vec{J}}_s) \bigr) \, dV .
\end{split}
\end{equation}
Subtracting \eqref{eq:perturbation-analysis2} by \eqref{eq:perturbation-analysis1} yields
\begin{equation}  \label{eq:perturbation-analysis3}
    \oint_{\partial V} ( \delta\vec{E} \times \acute{\vec{H}} ) \cdot \hat{\vec{n}} \, dS
    = \int_V \bigl( -j \omega (\mu \delta\vec{H} + \delta\mu \vec{H}) \cdot \acute{\vec{H}} - \delta\vec{E} \cdot (\sigma \acute{\vec{E}} + \acute{\vec{J}}_s) \bigr) \, dV .
\end{equation}
Following a similar approach to expand $\oint_{\partial V} ( \acute{\vec{E}} \times \vec{H} ) \cdot \hat{\vec{n}} dS$ and $\oint_{\partial V} ( \acute{\vec{E}} \times (\vec{H} + \delta\vec{H}) ) \cdot \hat{\vec{n}} dS$, then subtracting the two yields
\begin{equation}  \label{eq:perturbation-analysis4}
    \oint_{\partial V} ( \acute{\vec{E}} \times \delta\vec{H} ) \cdot \hat{\vec{n}} \, dS
    = \int_V \bigl( (-j \omega \mu \acute{\vec{H}} - \acute{\vec{L}}_s) \cdot \delta\vec{H} - \acute{\vec{E}} \cdot (\sigma \delta\vec{E} + \delta\sigma \vec{E}) \bigr) \, dV .
\end{equation}
Finally, subtracting \eqref{eq:perturbation-analysis3} by \eqref{eq:perturbation-analysis4} and rearranging terms leads to
\begin{equation}  \label{eq:perturbation-analysis5}
    \int_V \bigl( \acute{\vec{J}}_s \cdot \delta\vec{E} - \acute{\vec{L}}_s \cdot \delta\vec{H} \bigr) \, dV
    = \int_V \bigl( \vec{E} \cdot \acute{\vec{E}} \delta\sigma - j\omega \vec{H} \cdot \acute{\vec{H}} \delta\mu \bigr) \, dV
    + \oint_{\partial V} \bigl( \acute{\vec{H}} \times \delta\vec{E} + \acute{\vec{E}} \times \delta\vec{H} \bigr) \cdot \hat{\vec{n}} \, dS .
\end{equation}
In \eqref{eq:perturbation-analysis5}, the surface integral term is zero because $\partial V$ extends to infinity in all directions, with $\acute{\vec{H}}$ and $\acute{\vec{E}}$ both approaching zero at infinity.

Now, denote the MFD component measurement of a magnetic sensor as the volume integral of a function $f$ of the electromagnetic fields:
\begin{equation}  \label{eq:measurement}
    B = \int_V f(\vec{E}, \vec{H}) \, dV ,
\end{equation}
where from \eqref{eq:MFD-measurement},
\begin{equation*}
    f(\vec{E}, \vec{H}) = \mu_0 \vec{H} \cdot \hat{\vec{n}}_\text{sensor} \, \delta(\vec{r} - \vec{r}_\text{sensor}) .
\end{equation*}
A perturbation in the measurement results from changes in the electromagnetic fields, expressed as:
\begin{equation}  \label{eq:measurement-perturb-field}
    \delta B
    = \int_V \biggl( \frac{\partial f(\vec{E}, \vec{H})}{\partial \vec{E}} \cdot \delta\vec{E} + \frac{\partial f(\vec{E}, \vec{H})}{\partial \vec{H}} \cdot \delta\vec{H} \biggr) \, dV .
\end{equation}
The measurement perturbation can also be represented as a linear function of perturbations in the electric conductivity and magnetic permeability fields, expressed as:
\begin{equation}  \label{eq:measurement-perturb-material}
    \delta B
    = \int_V \bigl( S_\sigma \delta\sigma + S_\mu \delta\mu \bigr) \, dV ,
\end{equation}
where $S_\sigma$ and $S_\mu$ are the linearization coefficients.
Equating \eqref{eq:measurement-perturb-field} and \eqref{eq:measurement-perturb-material}, then comparing with \eqref{eq:perturbation-analysis5}, it becomes apparent that by selecting
\begin{subequations}
\begin{align}
    \acute{\vec{J}}_s &=  \frac{\partial f(\vec{E}, \vec{H})}{\partial \vec{E}} = \vec{0} ,  \\
    \acute{\vec{L}}_s &= -\frac{\partial f(\vec{E}, \vec{H})}{\partial \vec{H}} = -\mu_0 \hat{\vec{n}}_\text{sensor} \, \delta(\vec{r} - \vec{r}_\text{sensor}) ,
\end{align}
\end{subequations}
the linearization coefficients in \eqref{eq:measurement-perturb-material} are
\begin{equation}
    S_\sigma = \vec{E} \cdot \acute{\vec{E}} , \quad
    S_\mu    = - j\omega \vec{H} \cdot \acute{\vec{H}} .
\end{equation}

%% file: contents/appendix3.tex
\section{Semi-analytical solution for electromagnetic fields induced by inner sources in pipes}
\label{app:electromagnetic-field-pipe-semianalytical}

\cref{fig:pipe-schematic} shows the schematic of an infinity long pipe with excitation on the inner side of the pipe. Region 1 extends from the outer boundary of the coil to the inner radius of the pipe. Only region 2 differs from air; it has an electrical conductivity of $\sigma^{(2)}$ and magnetic permeability of $\mu^{(2)} = \mu_0 \mu_r^{(2)}$, where $\mu_r$ represents the relative magnetic permeability.

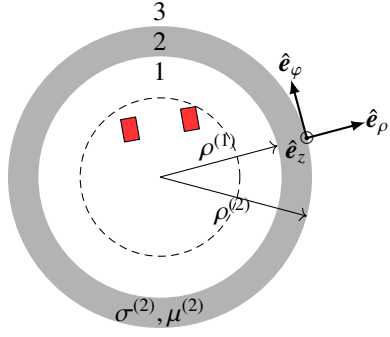
\begin{figure}
    \centering
    \input{figures/pipe-schematic.tex}
    \caption{Schematic of a pipe with inner excitation.}
    \label{fig:pipe-schematic}
\end{figure}

Analytical solutions for eddy current systems are often derived using the \ac{SOVP} \cite{RN40}. In cylindrical coordinates, the \ac{SOVP}, denoted as $\vec{W}$ is given by
\begin{equation}
\begin{split}
    \vec{W}
    &= W_a \hat{\vec{e}}_z + \hat{\vec{e}}_z \times \nabla W_b  \\
    &= -\frac{1}{\rho}\frac{\partial W_b}{\partial \varphi} \hat{\vec{e}}_\rho + \frac{\partial W_b}{\partial \rho} \hat{\vec{e}}_\varphi + W_a \hat{\vec{e}}_z ,
\end{split}
\end{equation}
where $W_a$ and $W_b$ are scalar fields governed by:
\begin{subequations} \label{eq:Wa-Wb-governing}
\begin{align}
    \frac{1}{\rho} \frac{\partial}{\partial \rho} \bigg( \rho \frac{\partial W_a}{\partial \rho} \bigg) + \frac{1}{\rho^2} \frac{\partial^2 W_a}{\partial \varphi^2} + \frac{\partial^2 W_a}{\partial z^2} + k^2 W_a &= 0  \label{eq:Wa-governing} , \\
    \frac{1}{\rho} \frac{\partial}{\partial \rho} \bigg( \rho \frac{\partial W_b}{\partial \rho} \bigg) + \frac{1}{\rho^2} \frac{\partial^2 W_b}{\partial \varphi^2} + \frac{\partial^2 W_b}{\partial z^2} + k^2 W_b &= 0  \label{eq:Wb-governing} ,
\end{align}
\end{subequations}
with $k^2 = -j \omega \mu \sigma$.
The magnetic vector potential $\vec{A}$ and \ac{MFD} $\vec{B}$ relate to $W_a$ and $W_b$ as follows:
\begin{subequations}
\begin{align}
    \vec{A} &=
    \biggl(  \frac{1}{\rho} \frac{\partial W_a}{\partial \varphi} - \frac{\partial^2 W_b}{\partial \rho \partial z} \biggr) \, \hat{\vec{e}}_\rho +
    \biggl( -\frac{\partial W_a}{\partial \rho} - \frac{1}{\rho} \frac{\partial^2 W_b}{\partial \varphi \partial z} \biggr) \, \hat{\vec{e}}_\varphi +
    \biggl(  \frac{\partial^2 W_b}{\partial \rho^2} + \frac{1}{\rho} \frac{\partial W_b}{\partial \rho} + \frac{1}{\rho^2} \frac{\partial^2 W_b}{\partial \varphi^2} \biggr) \, \hat{\vec{e}}_z ,  \label{eq:MVP}  \\
    \vec{B} &=
    \biggl( -k^2 \frac{1}{\rho} \frac{\partial W_b}{\partial \varphi} + \frac{\partial^2 W_a}{\partial \rho \partial z} \biggr) \, \hat{\vec{e}}_\rho +
    \biggl(  k^2 \frac{\partial W_b}{\partial \rho} + \frac{1}{\rho} \frac{\partial^2 W_a}{\partial \varphi \partial z} \biggr) \, \hat{\vec{e}}_\varphi +
    \biggl(  k^2 W_a + \frac{\partial^2 W_a}{\partial z^2} \biggr) \, \hat{\vec{e}}_z .  \label{eq:MFD}
\end{align}
\end{subequations}

Since the pipe is infinity long in the z-axis, Fourier transform is applied in this direction:
\begin{subequations}
\begin{align}
    \tilde{W}_a(\rho,\varphi,\kappa) = \int_{-\infty}^\infty W_a(\rho,\varphi,z) \, e^{-j\kappa z} \, dz ,  \label{eq:Wa-ft}  \\
    \tilde{W}_b(\rho,\varphi,\kappa) = \int_{-\infty}^\infty W_b(\rho,\varphi,z) \, e^{-j\kappa z} \, dz .  \label{eq:Wb-ft}
\end{align}
\end{subequations}
For clarity, denote $\tilde{W}_a$ and $\tilde{W}_b$ for each region as $\tilde{W}_a^{(t)}$ and $\tilde{W}_b^{(t)}$, where $t=1,2,3$. With the transformed \ac{SOVP}, \eqref{eq:Wa-Wb-governing} reduces to
\begin{subequations} \label{eq:Wa-Wb-ft-governing}
\begin{align}
    \frac{1}{\rho} \frac{\partial}{\partial \rho} \bigg( \rho \frac{\partial \tilde{W}_a^{(t)}}{\partial \rho} \bigg) + \frac{1}{\rho^2} \frac{\partial^2 \tilde{W}_a^{(t)}}{\partial \varphi^2} &= (\lambda^{(t)})^2 \tilde{W}_a^{(t)}  \label{eq:Wa-ft-governing} ,  \\
    \frac{1}{\rho} \frac{\partial}{\partial \rho} \bigg( \rho \frac{\partial \tilde{W}_b^{(t)}}{\partial \rho} \bigg) + \frac{1}{\rho^2} \frac{\partial^2 \tilde{W}_b^{(t)}}{\partial \varphi^2} &= (\lambda^{(t)})^2 \tilde{W}_b^{(t)}  \label{eq:Wb-ft-governing} ,
\end{align}
\end{subequations}
where $(\lambda^{(t)})^2 = -(k^{(t)})^2 + \kappa^2$.
Using separation of variables, the general solution to \eqref{eq:Wa-Wb-ft-governing} for region 2 is
\begin{subequations}   \label{eq:Wa-Wb-ft-2-sol}
\begin{align}
    \tilde{W}_a^{(2)} &= \sum_{\nu=-\infty}^\infty \left[ C_a^{(2)} I_\nu(\lambda^{(2)} \rho) + D_a^{(2)} K_\nu(\lambda^{(2)} \rho) \right] e^{j \nu \varphi} ,   \\
    \tilde{W}_b^{(2)} &= \sum_{\nu=-\infty}^\infty \left[ C_b^{(2)} I_\nu(\lambda^{(2)} \rho) + D_b^{(2)} K_\nu(\lambda^{(2)} \rho) \right] e^{j \nu \varphi} ,
\end{align}
\end{subequations}
where $(\lambda^{(2)})^2 = j \omega \mu^{(2)} \sigma^{(2)} + \kappa^2$, and $I_\nu(\cdot)$ and $K_\nu(\cdot)$ are the modified Bessel function of the first and second kind of order $\nu$, respectively.

In regions 1 and 3, the electrical conductivity is zero, which results in $k^2 = -j \omega \mu \sigma = 0$. Since $W_b$ is always associated with a $k^2$ term in \eqref{eq:MFD}, $W_b$ can be set to zero.
With region 3 extending to $\rho=\infty$, the solution to \eqref{eq:Wa-ft-governing} for region 3 is
\begin{equation}   \label{eq:Wa-ft-3-sol}
    \tilde{W}_a^{(3)} = \sum_{\nu=-\infty}^\infty D^{(3)} K_\nu(|\kappa| \rho) \, e^{j \nu \varphi} .
\end{equation}
In region 1, the solution to \eqref{eq:Wa-ft-governing} is
\begin{equation}   \label{eq:Wa-ft-1-sol}
    \tilde{W}_a^{(1)} = \sum_{\nu=-\infty}^\infty \left[ C^{(ec)} I_\nu(|\kappa| \rho) + D^{(s)} K_\nu(|\kappa| \rho) \right] e^{j \nu \varphi} .
\end{equation}
The terms $C^{(ec)} I_\nu(|\kappa| \rho) e^{j \nu \varphi}$ are attributed to eddy current because the terms decreases with decreasing $\rho$. The remaining terms $D^{(s)} K_\nu(|\kappa| \rho) e^{j \nu \varphi}$ are attributed to coil current.

Taking the Fourier transform of \eqref{eq:MVP} and \eqref{eq:MFD} in the z-direction results in
\begin{subequations} \label{eq:fields-sol}
\begin{align}
    \tilde{\vec{A}} &=
    \biggl(  \frac{1}{\rho} \frac{\partial \tilde{W}_a}{\partial \varphi} - j\kappa \frac{\partial \tilde{W}_b}{\partial \rho} \biggr) \, \hat{\vec{e}}_\rho +
    \biggl( -\frac{\partial \tilde{W}_a}{\partial \rho} - j\kappa \frac{1}{\rho} \frac{\partial \tilde{W}_b}{\partial \varphi} \biggr) \, \hat{\vec{e}}_\varphi +
    \biggl(  \frac{\partial^2 \tilde{W}_b}{\partial \rho^2} + \frac{1}{\rho} \frac{\partial \tilde{W}_b}{\partial \rho} + \frac{1}{\rho^2} \frac{\partial^2 \tilde{W}_b}{\partial \varphi^2} \biggr) \, \hat{\vec{e}}_z ,  \\
    \tilde{\vec{B}} &=
    \biggl( -k^2 \frac{1}{\rho} \frac{\partial \tilde{W}_b}{\partial \varphi} + j\kappa \frac{\partial \tilde{W}_a}{\partial \rho} \biggr) \, \hat{\vec{e}}_\rho +
    \biggl(  k^2 \frac{\partial \tilde{W}_b}{\partial \rho} + j\kappa \frac{1}{\rho} \frac{\partial \tilde{W}_a}{\partial \varphi} \biggr) \, \hat{\vec{e}}_\varphi+
    \biggl( -\lambda^2 \tilde{W}_a \biggr) \, \hat{\vec{e}}_z \label{eq:MFD-ft-sol} .
\end{align}
\end{subequations}
Using the relation $\tilde{\vec{E}} = -j \omega \tilde{\vec{A}}$ and the constitutive equation $\tilde{\vec{H}} = \frac{1}{\mu} \tilde{\vec{B}}$, the $\vec{E}$ and $\vec{H}$ fields are obtained from the inverse Fourier transform:
\begin{subequations}  \label{eq:E-ift}
\begin{align}
    \vec{E}(\rho,\varphi,z) &= \frac{1}{2\pi} \int_{-\infty}^\infty     -j \omega \tilde{\vec{A}}(\rho,\varphi,\kappa) \, e^{j\kappa z} \, d\kappa ,  \\
    \vec{H}(\rho,\varphi,z) &= \frac{1}{2\pi} \int_{-\infty}^\infty \frac{1}{\mu} \tilde{\vec{B}}(\rho,\varphi,\kappa) \, e^{j\kappa z} \, d\kappa .
\end{align}
\end{subequations}

To determine the coefficients $C^{(ec)}$, $D^{(s)}$, $C_a^{(2)}$, $D_a^{(2)}$, $C_b^{(2)}$, $D_b^{(2)}$, $D^{(3)}$ in \eqref{eq:Wa-Wb-ft-2-sol}, \eqref{eq:Wa-ft-3-sol}, and \eqref{eq:Wa-ft-1-sol}, proceed as follows:

\paragraph*{Solve $D^{(s)}$ for an arbitrary coil}
For any arbitrary coil with coil current density $\vec{J}_s$, its contributing \ac{MFD} in free space can be calculated using the Biot-Savart law. Retaining the $\rho$ component and take the Fourier transform in the z -irection yields
\begin{equation} \label{eq:Bs-ft}
    \tilde{B}_\rho^{(s)}(\rho,\varphi,\kappa) = \int_{-\infty}^\infty B_\rho^{(s)}(\rho,\varphi,z) \, e^{-j\kappa z} \, dz .
\end{equation}
Based on \eqref{eq:MFD-ft-sol} and \eqref{eq:Wa-ft-1-sol}, $\tilde{B}_\rho^{(s)}$ can be expressed as
\begin{equation} \label{eq:Bs-ft-2}
    \tilde{B}_\rho^{(s)}(\rho,\varphi,\kappa)
    = j\kappa \frac{\partial}{\partial \rho} \biggl( \sum_{\nu=-\infty}^\infty D^{(s)} K_\nu(|\kappa| \rho) \, e^{j \nu \varphi} \biggr)
    = \sum_{\nu=-\infty}^\infty j\kappa D^{(s)} |\kappa| \, K'_\nu(|\kappa| \rho) \, e^{j \nu \varphi}  .
\end{equation}
Evaluating \eqref{eq:Bs-ft-2} at $\rho=\rho^{(1)}$ and utilizing the orthogonality of the exponential functions, $D^{(s)}$ can be computed from
\begin{equation} \label{eq:Ds-arbitrary-coil}
    D^{(s)} = \frac{1}{j\kappa |\kappa| \, K'_\nu(|\kappa| \rho^{(1)})} \frac{1}{2\pi} \int_{-\pi}^{\pi} \tilde{B}_\rho^{(s)}(\rho^{(1)},\varphi,\kappa) \, e^{-j \nu \varphi} \, d\varphi   .
\end{equation}

\paragraph*{Solve $D^{(s)}$ corresponding to a magnetic sensor}
For a magnetic sensor with position $\vec{r}_\text{sensor}$ at $(\rho,\varphi,z) = (\rho_\text{s},\varphi_\text{s},z_\text{s})$ and sensing axis $\hat{\vec{n}}_\text{sensor} = n_{\text{s}\rho} \hat{\vec{e}}_\rho + n_{\text{s}\varphi} \hat{\vec{e}}_\varphi + n_{\text{s}z} \hat{\vec{e}}_z$, the electromagnetic system described in \eqref{eq:electromagnetic-system2} needs to be solved. It can be shown that the corresponding $D^{(s)}$ coefficient is
\begin{equation} \label{eq:Ds-magnetic-sensor}
    D^{(s)} =
    \frac{j \mu_0}{\pi^3 \omega} \Bigl( j |\kappa| \, I'_\nu(|\kappa| \rho_\text{s}) \, n_{\text{s}\rho} + \frac{\nu}{\kappa} \frac{1}{\rho_\text{s}} I_\nu(|\kappa| \rho_\text{s}) \, n_{\text{s}\varphi} + I_\nu(|\kappa| \rho_\text{s}) \, n_{\text{s}z} \Bigr) \, e^{-j\nu \varphi_\text{s} -j\kappa z_\text{s}} .
\end{equation}

\paragraph*{Solve $C^{(ec)}$, $C_a^{(2)}$, $D_a^{(2)}$, $C_b^{(2)}$, $D_b^{(2)}$, $D^{(3)}$ from interface conditions}
Continuity of the normal component of the \ac{MFD} and the tangential components of the magnetic field intensity across any interface provides conditions to solve for the unknown coefficients. For the interfaces $\rho = \rho^{(1)}$ and $\rho = \rho^{(2)}$,
\begin{equation*} \label{eq:interface-conditions1}
\begin{array}{r}
    \tilde{\vec{B}}^{(1)} \cdot \hat{\vec{e}}_\rho    =                       \tilde{\vec{B}}^{(2)} \cdot \hat{\vec{e}}_\rho     \big\vert_{\rho=\rho^{(1)}} ,  \\
    \tilde{\vec{B}}^{(1)} \cdot \hat{\vec{e}}_\varphi = \frac{1}{\mu_r^{(2)}} \tilde{\vec{B}}^{(2)} \cdot \hat{\vec{e}}_\varphi  \big\vert_{\rho=\rho^{(1)}} ,  \\
    \tilde{\vec{B}}^{(1)} \cdot \hat{\vec{e}}_z       = \frac{1}{\mu_r^{(2)}} \tilde{\vec{B}}^{(2)} \cdot \hat{\vec{e}}_z        \big\vert_{\rho=\rho^{(1)}} ,
\end{array} \quad
\begin{array}{r}
    \tilde{\vec{B}}^{(2)} \cdot \hat{\vec{e}}_\rho                          = \tilde{\vec{B}}^{(3)} \cdot \hat{\vec{e}}_\rho     \big\vert_{\rho=\rho^{(2)}} ,  \\
    \frac{1}{\mu_r^{(2)}} \tilde{\vec{B}}^{(2)} \cdot \hat{\vec{e}}_\varphi = \tilde{\vec{B}}^{(3)} \cdot \hat{\vec{e}}_\varphi  \big\vert_{\rho=\rho^{(2)}} ,  \\
    \frac{1}{\mu_r^{(2)}} \tilde{\vec{B}}^{(2)} \cdot \hat{\vec{e}}_z       = \tilde{\vec{B}}^{(3)} \cdot \hat{\vec{e}}_z        \big\vert_{\rho=\rho^{(2)}} .
\end{array}
\end{equation*}
By substituting in \eqref{eq:MFD-ft-sol} then \eqref{eq:Wa-Wb-ft-2-sol}, \eqref{eq:Wa-ft-3-sol}, and \eqref{eq:Wa-ft-1-sol}, the unknown coefficients $C^{(ec)}$, $C_a^{(2)}$, $D_a^{(2)}$, $C_b^{(2)}$, $D_b^{(2)}$, $D^{(3)}$ are related to $D^{(s)}$ as follows:
\begin{equation} \label{eq:coef}
    \mat{\Lambda}
    \begin{bmatrix}
        C^{(ec)} \\ C_a^{(2)} \\ D_a^{(2)} \\ C_b^{(2)} \\ D_b^{(2)} \\ D^{(3)}
    \end{bmatrix}
    =
    \begin{bmatrix}
        j\kappa |\kappa| K'_\nu(|\kappa| \rho^{(1)}) \\
        -\kappa \nu \frac{1}{\rho^{(1)}} K_\nu(|\kappa| \rho^{(1)}) \\
        -\kappa^2 K_\nu(|\kappa| \rho^{(1)}) \\
        0 \\ 0 \\ 0
    \end{bmatrix}
    D^{(s)} ,
\end{equation}
where
\begin{equation*}
    \mat{\Lambda} =
    {\small
    \begin{bmatrix}
        -j\kappa |\kappa| I'_\nu(|\kappa| \rho^{(1)})               &  j\kappa \lambda^{(2)} I'_\nu(\lambda^{(2)} \rho^{(1)})                        &  j\kappa \lambda^{(2)} K'_\nu(\lambda^{(2)} \rho^{(1)})                        &  -j\nu (k^{(2)})^2 \frac{1}{\rho^{(1)}} I_\nu(\lambda^{(2)} \rho^{(1)})            &  -j\nu (k^{(2)})^2 \frac{1}{\rho^{(1)}} K_\nu(\lambda^{(2)} \rho^{(1)})            &  0                                                           \\
        \kappa \nu \frac{1}{\rho^{(1)}} I_\nu(|\kappa| \rho^{(1)})  &  -\kappa \nu \frac{1}{\mu_r^{(2)} \rho^{(1)}} I_\nu(\lambda^{(2)} \rho^{(1)})  &  -\kappa \nu \frac{1}{\mu_r^{(2)} \rho^{(1)}} K_\nu(\lambda^{(2)} \rho^{(1)})  &  (k^{(2)})^2 \lambda^{(2)} \frac{1}{\mu_r^{(2)}} I'_\nu(\lambda^{(2)} \rho^{(1)})  &  (k^{(2)})^2 \lambda^{(2)} \frac{1}{\mu_r^{(2)}} K'_\nu(\lambda^{(2)} \rho^{(1)})  &  0                                                           \\
        \kappa^2 I_\nu(|\kappa| \rho^{(1)})                         &  -(\lambda^{(2)})^2 \frac{1}{\mu_r^{(2)}} I_\nu(\lambda^{(2)} \rho^{(1)})      &  -(\lambda^{(2)})^2 \frac{1}{\mu_r^{(2)}} K_\nu(\lambda^{(2)} \rho^{(1)})      &  0                                                                                 &  0                                                                                 &  0                                                           \\
        0                                                           &  j\kappa \lambda^{(2)} I'_\nu(\lambda^{(2)} \rho^{(2)})                        &  j\kappa \lambda^{(2)} K'_\nu(\lambda^{(2)} \rho^{(2)})                        &  -j\nu (k^{(2)})^2 \frac{1}{\rho^{(2)}} I_\nu(\lambda^{(2)} \rho^{(2)})            &  -j\nu (k^{(2)})^2 \frac{1}{\rho^{(2)}} K_\nu(\lambda^{(2)} \rho^{(2)})            &  -j\kappa |\kappa| K'_\nu(|\kappa| \rho^{(2)})               \\
        0                                                           &  -\kappa \nu \frac{1}{\mu_r^{(2)} \rho^{(2)}} I_\nu(\lambda^{(2)} \rho^{(2)})  &  -\kappa \nu \frac{1}{\mu_r^{(2)} \rho^{(2)}} K_\nu(\lambda^{(2)} \rho^{(2)})  &  (k^{(2)})^2 \lambda^{(2)} \frac{1}{\mu_r^{(2)}} I'_\nu(\lambda^{(2)} \rho^{(2)})  &  (k^{(2)})^2 \lambda^{(2)} \frac{1}{\mu_r^{(2)}} K'_\nu(\lambda^{(2)} \rho^{(2)})  &  \kappa \nu \frac{1}{\rho^{(2)}} K_\nu(|\kappa| \rho^{(2)})  \\
        0                                                           &  -(\lambda^{(2)})^2 \frac{1}{\mu_r^{(2)}} I_\nu(\lambda^{(2)} \rho^{(2)})      &  -(\lambda^{(2)})^2 \frac{1}{\mu_r^{(2)}} K_\nu(\lambda^{(2)} \rho^{(2)})      &  0                                                                                 &  0                                                                                 &  -\kappa^2 K_\nu(|\kappa| \rho^{(2)})                        \\
    \end{bmatrix}} .
\end{equation*}

%% file: figures/pipe-schematic.tex
\begin{tikzpicture}
    \def\ra{1.6}
    \def\rb{2.0}
    \def\anga{15}
    \def\angb{-15}

    \def\r{2.0}
    \def\ang{15}
    \def\l{0.8}

    \fill[fill=gray!60,even odd rule]  (0,0) circle (\rb)  (0,0) circle (\ra);

    \node at (0,  1.5*\ra-0.5*\rb) {$1$};
    \node at (0,  0.5*\ra+0.5*\rb) {$2$};
    \node at (0, -0.5*\ra+1.5*\rb) {$3$};

    \node at (0, -0.5*\ra-0.5*\rb) {$\sigma^{(2)}, \mu^{(2)}$};

    \draw[->] (0,0) -- node[anchor=south, inner sep=0]{$\rho^{(1)}$} ({\ra*cos(\anga)},{\ra*sin(\anga)});
    \draw[->] (0,0) -- node[anchor=north, inner sep=0]{$\rho^{(2)}$} ({\rb*cos(\angb)},{\rb*sin(\angb)});

    \draw[thick, ->, >=latex] ({\r*cos(\ang)},{\r*sin(\ang)}) -- ({\r*cos(\ang)+\l*cos(\ang)},{\r*sin(\ang)+\l*sin(\ang)}) node[anchor=west, inner sep=0] {$\hat{\vec{e}}_\rho$};
    \draw[thick, ->, >=latex] ({\r*cos(\ang)},{\r*sin(\ang)}) -- ({\r*cos(\ang)-\l*sin(\ang)},{\r*sin(\ang)+\l*cos(\ang)}) node[anchor=south, inner sep=0] {$\hat{\vec{e}}_\varphi$};

    \node at ({\r*cos(\ang)},{\r*sin(\ang)}) {\footnotesize $\odot$};
    \fill ({\r*cos(\ang)},{\r*sin(\ang)}) circle (0.03);
    \node at ({\r*cos(\ang)-0.14},{\r*sin(\ang)-0.18}) {$\hat{\vec{e}}_z$};

    \begin{scope}[yshift=7mm]
        \draw [fill=red!80, rotate around={10:(0,0)}] (-0.5,-0.15) rectangle (-0.3,0.15);
        \draw [fill=red!80, rotate around={10:(0,0)}] ( 0.5,-0.15) rectangle ( 0.3,0.15);
    \end{scope}
    \draw[densely dashed] (0,0) circle (1.05);
\end{tikzpicture}

%% file: main.bbl
\begin{thebibliography}{10}
\expandafter\ifx\csname url\endcsname\relax
  \def\url#1{\texttt{#1}}\fi
\expandafter\ifx\csname urlprefix\endcsname\relax\def\urlprefix{URL }\fi
\expandafter\ifx\csname href\endcsname\relax
  \def\href#1#2{#2} \def\path#1{#1}\fi

\bibitem{RN1159}
M.~Lustig, D.~L. Donoho, J.~M. Santos, J.~M. Pauly, Compressed sensing {MRI}, IEEE Signal Processing Magazine 25~(2) (2008) 72--82.
\newblock \href {https://doi.org/10.1109/Msp.2007.914728} {\path{doi:10.1109/Msp.2007.914728}}.

\bibitem{ma2017}
L.~Ma, M.~Soleimani, \href{https://dx.doi.org/10.1088/1361-6501/aa7107}{Magnetic induction tomography methods and applications: A review}, Measurement Science and Technology 28~(7) (2017) 072001.
\newblock \href {https://doi.org/10.1088/1361-6501/aa7107} {\path{doi:10.1088/1361-6501/aa7107}}.
\newline\urlprefix\url{https://dx.doi.org/10.1088/1361-6501/aa7107}

\bibitem{RN119}
G.~S. Alberti, H.~Ammari, B.~T. Jin, J.~K. Seo, W.~L. Zhang, The linearized inverse problem in multifrequency electrical impedance tomography, SIAM Journal on Imaging Sciences 9~(4) (2016) 1525--1551.
\newblock \href {https://doi.org/10.1137/16m1061564} {\path{doi:10.1137/16m1061564}}.

\bibitem{RN271}
E.~J. Candès, Compressive sampling, in: Proceedings of the International Congress of Mathematicians, Vol.~3, 2006, pp. 1433--1452.

\bibitem{RN276}
D.~L. Donoho, Compressed sensing, IEEE Transactions on Information Theory 52~(4) (2006) 1289--1306.
\newblock \href {https://doi.org/10.1109/Tit.2006.871582} {\path{doi:10.1109/Tit.2006.871582}}.

\bibitem{RN1226}
E.~C. Marques, N.~Maciel, L.~Naviner, H.~Cai, J.~Yang, A review of sparse recovery algorithms, IEEE Access 7 (2019) 1300--1322.
\newblock \href {https://doi.org/10.1109/Access.2018.2886471} {\path{doi:10.1109/Access.2018.2886471}}.

\bibitem{RN1217}
M.~Rani, S.~B. Dhok, R.~B. Deshmukh, A systematic review of compressive sensing: Concepts, implementations and applications, IEEE Access 6 (2018) 4875--4894.
\newblock \href {https://doi.org/10.1109/Access.2018.2793851} {\path{doi:10.1109/Access.2018.2793851}}.

\bibitem{natarajan95}
B.~K. Natarajan, Sparse approximate solutions to linear systems, SIAM Journal on Computing 24~(2) (1995) 227--234.
\newblock \href {https://doi.org/10.1137/S0097539792240406} {\path{doi:10.1137/S0097539792240406}}.

\bibitem{RN1099}
J.~Wright, Y.~Ma, High-Dimensional Data Analysis with Low-Dimensional Models: Principles, Computation, and Applications, Cambridge University Press, 2022, p.~57.
\newblock \href {https://doi.org/10.1017/9781108779302} {\path{doi:10.1017/9781108779302}}.

\bibitem{RN227}
E.~J. Candès, T.~Tao, Decoding by linear programming, IEEE Transactions on Information Theory 51~(12) (2005) 4203--4215.
\newblock \href {https://doi.org/10.1109/Tit.2005.858979} {\path{doi:10.1109/Tit.2005.858979}}.

\bibitem{RN230}
E.~J. Candès, J.~Romberg, T.~Tao, Robust uncertainty principles: Exact signal reconstruction from highly incomplete frequency information, IEEE Transactions on Information Theory 52~(2) (2006) 489--509.
\newblock \href {https://doi.org/10.1109/Tit.2005.862083} {\path{doi:10.1109/Tit.2005.862083}}.

\bibitem{RN1200}
E.~J. Candès, The restricted isometry property and its implications for compressed sensing, Comptes Rendus. Mathématique 346~(9-10) (2008) 589--592.
\newblock \href {https://doi.org/10.1016/j.crma.2008.03.014} {\path{doi:10.1016/j.crma.2008.03.014}}.

\bibitem{RN64}
M.~E. Tipping, Sparse {Bayesian} learning and the relevance vector machine, Journal of Machine Learning Research 1~(3) (2001) 211--244.
\newblock \href {https://doi.org/10.1162/15324430152748236} {\path{doi:10.1162/15324430152748236}}.

\bibitem{RN211}
S.~H. Ji, Y.~Xue, L.~Carin, Bayesian compressive sensing, IEEE Transactions on Signal Processing 56~(6) (2008) 2346--2356.
\newblock \href {https://doi.org/10.1109/Tsp.2007.914345} {\path{doi:10.1109/Tsp.2007.914345}}.

\bibitem{RN33}
Z.~L. Zhang, B.~D. Rao, Extension of {SBL} algorithms for the recovery of block sparse signals with intra-block correlation, IEEE Transactions on Signal Processing 61~(8) (2013) 2009--2015.
\newblock \href {https://doi.org/10.1109/Tsp.2013.2241055} {\path{doi:10.1109/Tsp.2013.2241055}}.

\bibitem{RN10}
J.~Fang, Y.~N. Shen, H.~B. Li, P.~Wang, Pattern-coupled sparse {Bayesian} learning for recovery of block-sparse signals, IEEE Transactions on Signal Processing 63~(2) (2015) 360--372.
\newblock \href {https://doi.org/10.1109/Tsp.2014.2375133} {\path{doi:10.1109/Tsp.2014.2375133}}.

\bibitem{RN127}
L.~Wang, L.~F. Zhao, L.~Yu, J.~J. Wang, G.~A. Bi, Structured {Bayesian} learning for recovery of clustered sparse signal, Signal Processing 166 (2020) 107255.
\newblock \href {https://doi.org/10.1016/j.sigpro.2019.107255} {\path{doi:10.1016/j.sigpro.2019.107255}}.

\bibitem{RN96}
Z.~Zhang, T.~P. Jung, S.~Makeig, Z.~Pi, B.~D. Rao, Spatiotemporal sparse {Bayesian} learning with applications to compressed sensing of multichannel physiological signals, IEEE Transactions on Neural Systems and Rehabilitation Engineering 22~(6) (2014) 1186--1197.
\newblock \href {https://doi.org/10.1109/Tnsre.2014.2319334} {\path{doi:10.1109/Tnsre.2014.2319334}}.

\bibitem{RN1440}
S.~Rangan, Generalized approximate message passing for estimation with random linear mixing, in: IEEE International Symposium on Information Theory Proceedings, 2011, pp. 2168--2172.
\newblock \href {https://doi.org/10.1109/ISIT.2011.6033942} {\path{doi:10.1109/ISIT.2011.6033942}}.

\bibitem{RN1444}
D.~L. Donoho, A.~Maleki, A.~Montanari, Message-passing algorithms for compressed sensing, Proceedings of the National Academy of Sciences 106~(45) (2009) 18914--18919.
\newblock \href {https://doi.org/10.1073/Pnas.0909892106} {\path{doi:10.1073/Pnas.0909892106}}.

\bibitem{RN1219}
D.~L. Donoho, J.~Tanner, Precise undersampling theorems, Proceedings of the IEEE 98~(6) (2010) 913--924.
\newblock \href {https://doi.org/10.1109/Jproc.2010.2045630} {\path{doi:10.1109/Jproc.2010.2045630}}.

\bibitem{RN1157}
T.~S. Jayram, S.~Pal, V.~Arya, Recovery of a sparse integer solution to an underdetermined system of linear equations, arXiv:1112.1757 (2011).
\newblock \href {https://doi.org/10.48550/arXiv.1112.1757} {\path{doi:10.48550/arXiv.1112.1757}}.

\bibitem{RN1155}
O.~L. Mangasarian, B.~Recht, Probability of unique integer solution to a system of linear equations, European Journal of Operational Research 214~(1) (2011) 27--30.
\newblock \href {https://doi.org/10.1016/j.ejor.2011.04.010} {\path{doi:10.1016/j.ejor.2011.04.010}}.

\bibitem{RN1467}
J.~Kleinberg, E.~Tardos, Algorithm Design, Pearson, 2005, pp. 491--493.

\bibitem{RN157}
D.~G. Tzikas, A.~C. Likas, N.~P. Galatsanos, The variational approximation for {Bayesian} inference: Life after the {EM} algorithm, IEEE Signal Processing Magazine 25~(6) (2008) 131--146.
\newblock \href {https://doi.org/10.1109/Msp.2008.929620} {\path{doi:10.1109/Msp.2008.929620}}.

\bibitem{RN1028}
C.~M. Bishop, Pattern Recognition and Machine Learning, Springer, 2006, pp. 55--57.

\bibitem{koller2009pgm}
D.~Koller, N.~Friedman, Probabilistic Graphical Models: Principles and Techniques, MIT press, 2009, pp. 277--279, 439--442.

\bibitem{minka2001}
T.~P. Minka, A family of algorithms for approximate {Bayesian} inference, Ph.D. thesis, Massachusetts Institute of Technology (2001).
\newblock \href {https://doi.org/1721.1/86583} {\path{doi:1721.1/86583}}.

\bibitem{RN93}
T.~Goldstein, S.~Osher, The split {Bregman} method for {L1}-regularized problems, SIAM Journal on Imaging Sciences 2~(2) (2009) 323--343.
\newblock \href {https://doi.org/10.1137/080725891} {\path{doi:10.1137/080725891}}.

\bibitem{RN1311}
S.~Boyd, N.~Parikh, E.~Chu, B.~Peleato, J.~Eckstein, Distributed optimization and statistical learning via the alternating direction method of multipliers, Foundations and Trends in Machine Learning 3~(1) (2011) 1--122.
\newblock \href {https://doi.org/10.1561/2200000016} {\path{doi:10.1561/2200000016}}.

\bibitem{RN1414}
W.-C. Li, C.-Y. Lin, Sparse magnetic array for the imaging of defects in multilayer metals, IEEE Sensors Journal 24~(9) (2024) 14082--14092.
\newblock \href {https://doi.org/10.1109/JSEN.2024.3381623} {\path{doi:10.1109/JSEN.2024.3381623}}.

\bibitem{RN1284}
X.~F. Mao, Y.~Z. Lei, Analytical solutions to eddy current field excited by a probe coil near a conductive pipe, NDT \& E International 54 (2013) 69--74.
\newblock \href {https://doi.org/10.1016/j.ndteint.2012.11.010} {\path{doi:10.1016/j.ndteint.2012.11.010}}.

\bibitem{RN980}
S.~Thrun, W.~Burgard, D.~Fox, Probabilistic Robotics, MIT Press, 2005, pp. 94--96.

\bibitem{RN40}
T.~P. Theodoulidis, E.~E. Kriezis, Eddy Current Canonical Problems (with Applications to Nondestructive Evaluation), Tech Science Press, 2006, pp. 6--14.

\end{thebibliography}
